\documentclass[draftclsnofoot,onecolumn]{IEEEtran}

\usepackage[font=footnotesize,format=plain,skip=15pt,labelsep=period]{caption}

\usepackage[numbers,sort&compress]{natbib}

\makeatletter

\usepackage{etex}
\usepackage{enumerate}
\usepackage{zref-savepos}

\newcounter{mnote}%[page]

\def\xmarginnote{%
  \xymarginnote{\hskip -\marginparsep \hskip -\marginparwidth}}

\def\ymarginnote{%
  \xymarginnote{\hskip\columnwidth \hskip\marginparsep}}

\long\def\xymarginnote#1#2{%
\vadjust{#1%
\smash{\hbox{{%
        \hsize\marginparwidth
        \@parboxrestore
        \@marginparreset
\footnotesize #2}}}}}

\def\mnoteson{%
\gdef\mnote##1{\refstepcounter{mnote}\label{##1}%
  \zsavepos{##1}%
  \ifnum20432158>\number\zposx{##1}%
  \xmarginnote{{\color{blue}\bf $\langle$\arabic{mnote}$\rangle$}}% 
  \else
  \ymarginnote{{\color{blue}\bf $\langle$\arabic{mnote}$\rangle$}}%
  \fi%
}
  }
\gdef\mnotesoff{\gdef\mnote##1{}}
\mnoteson
\mnotesoff

\usepackage{framed}
\usepackage{comment}
\usepackage[svgnames,dvipsnames]{xcolor}
\usepackage[all]{xy}
\usepackage{tikz}
\usetikzlibrary{positioning,matrix,through,calc,arrows,fit,shapes,decorations.pathreplacing,decorations.markings,trees}

\tikzstyle{block} = [draw,fill=blue!20,minimum size=2em]

\usepackage{qsymbols,amssymb,mathrsfs}
\usepackage{amsmath}
\usepackage[standard,thmmarks]{ntheorem}
\theoremstyle{plain}
\theoremsymbol{\ensuremath{_\vartriangleleft}}
\theorembodyfont{\itshape}
\theoremheaderfont{\normalfont\bfseries}
\theoremseparator{}

\theoremstyle{nonumberplain}
\theoremheaderfont{\scshape}
\theorembodyfont{\normalfont}
\theoremsymbol{\ensuremath{_\blacktriangleleft}}

\theoremnumbering{arabic}
\theoremstyle{plain}
\usepackage{latexsym}
\theoremsymbol{\ensuremath{_\Box}}
\theorembodyfont{\itshape}
\theoremheaderfont{\normalfont\bfseries}
\theoremseparator{}
\newtheorem{Conjecture}{Conjecture}

\theorembodyfont{\upshape}
\theoremprework{\bigskip\hrule}
\theorempostwork{\hrule\bigskip}
\usepackage[overload]{empheq} % no \intertext and \displaybreak

\let\iftwocolumn\if@twocolumn
\g@addto@macro\@twocolumntrue{\let\iftwocolumn\if@twocolumn}
\g@addto@macro\@twocolumnfalse{\let\iftwocolumn\if@twocolumn}

\mathtoolsset{showonlyrefs=false,showmanualtags}
\let\underbrace\LaTeXunderbrace % adapt spacing to font sizes

\renewcommand{\eqref}[1]{\textup{(\refeq{#1})}} % eqref was not allowed in
\newtagform{brackets}[]{(}{)}   % new tags for equations
\usetagform{brackets}
\usepackage[Smaller]{cancel}

\PassOptionsToPackage{breaklinks,hyperindex=true,backref=false,bookmarksnumbered,bookmarksopen,linktocpage,colorlinks,linkcolor=BrickRed,citecolor=OliveGreen,urlcolor=Blue,pdfstartview=FitH}{hyperref}
\usepackage{hyperref}

\usepackage{graphicx,psfrag}
\graphicspath{{figure/}{image/}} % Search path of figures

\usepackage{diagbox} % \backslashbox{}{} for slashed entries
\usepackage{algorithm2e}
\usepackage{listings} 
\lstdefinelanguage{Maple}{
  morekeywords={proc,module,end, for,from,to,by,while,in,do,od
    ,if,elif,else,then,fi ,use,try,catch,finally}, sensitive,
  morecomment=[l]\#,
  morestring=[b]",morestring=[b]`}[keywords,comments,strings]
\DeclareMathAlphabet{\mathpzc}{OT1}{pzc}{m}{it}
\usepackage{upgreek} % \upalpha,\upbeta, ...
\usepackage{dsfont}  % \mathds

\def\multi@nostar#1#2{%
  \expandafter\def\csname multi#1\endcsname##1{%
    \if ##1.\let\next=\relax \else
    \def\next{\csname multi#1\endcsname}     
    \expandafter\newcommand\csname #1##1\endcsname{#2}
    \fi\next}}

\def\multi@star#1#2{%
  \expandafter\def\csname #1\endcsname##1{#2}
  \multi@nostar{#1}{#2}
}

\newcommand{\multi}{%
  \@ifstar \multi@star \multi@nostar}

\multi*{rm}{\mathrm{#1}}
\multi*{mc}{\mathcal{#1}}
\multi*{op}{\mathop {\operator@font #1}}
\multi*{ds}{\mathds{#1}}
\multi*{set}{\mathcal{#1}}
\multi*{rsfs}{\mathscr{#1}}
\multi*{pz}{\mathpzc{#1}}
\multi*{M}{\boldsymbol{#1}}
\multi*{R}{\mathsf{#1}}
\multi*{RM}{\M{\R{#1}}}
\multi*{bb}{\mathbb{#1}}
\multi*{td}{\tilde{#1}}
\multi*{tR}{\tilde{\mathsf{#1}}}
\multi*{trM}{\tilde{\M{\R{#1}}}}
\multi*{tset}{\tilde{\mathcal{#1}}}
\multi*{tM}{\tilde{\M{#1}}}
\multi*{baM}{\bar{\M{#1}}}
\multi*{ol}{\overline{#1}}

\multirm  ABCDEFGHIJKLMNOPQRSTUVWXYZabcdefghijklmnopqrstuvwxyz.
\multiol  ABCDEFGHIJKLMNOPQRSTUVWXYZabcdefghijklmnopqrstuvwxyz.
\multitR   ABCDEFGHIJKLMNOPQRSTUVWXYZabcdefghijklmnopqrstuvwxyz.
\multitd   ABCDEFGHIJKLMNOPQRSTUVWXYZabcdefghijklmnopqrstuvwxyz.
\multitset ABCDEFGHIJKLMNOPQRSTUVWXYZabcdefghijklmnopqrstuvwxyz.
\multitM   ABCDEFGHIJKLMNOPQRSTUVWXYZabcdefghijklmnopqrstuvwxyz.
\multibaM   ABCDEFGHIJKLMNOPQRSTUVWXYZabcdefghijklmnopqrstuvwxyz.
\multitrM   ABCDEFGHIJKLMNOPQRSTUVWXYZabcdefghijklmnopqrstuvwxyz.
\multimc   ABCDEFGHIJKLMNOPQRSTUVWXYZabcdefghijklmnopqrstuvwxyz.
\multiop   ABCDEFGHIJKLMNOPQRSTUVWXYZabcdefghijklmnopqrstuvwxyz.
\multids   ABCDEFGHIJKLMNOPQRSTUVWXYZabcdefghijklmnopqrstuvwxyz.
\multiset  ABCDEFGHIJKLMNOPQRSTUVWXYZabcdefghijklmnopqrstuvwxyz.
\multirsfs ABCDEFGHIJKLMNOPQRSTUVWXYZabcdefghijklmnopqrstuvwxyz.
\multipz   ABCDEFGHIJKLMNOPQRSTUVWXYZabcdefghijklmnopqrstuvwxyz.
\multiM    ABCDEFGHIJKLMNOPQRSTUVWXYZabcdefghijklmnopqrstuvwxyz.
\multiR    ABCDEFGHIJKL NO QRSTUVWXYZabcd fghijklmnopqrstuvwxyz.
\multibb   ABCDEFGHIJKLMNOPQRSTUVWXYZabcdefghijklmnopqrstuvwxyz.
\multiRM   ABCDEFGHIJKLMNOPQRSTUVWXYZabcdefghijklmnopqrstuvwxyz.

\newcommand{\dotleq}{\buildrel \textstyle  .\over {\smash{\lower
      .2ex\hbox{\ensuremath\leqslant}}\vphantom{=}}}
\newcommand{\dotgeq}{\buildrel \textstyle  .\over {\smash{\lower
      .2ex\hbox{\ensuremath\geqslant}}\vphantom{=}}}

\newcommand{\bM}{\begin{bmatrix}}
\newcommand{\eM}{\end{bmatrix}}
\newcommand{\bSM}{\left[\begin{smallmatrix}}
\newcommand{\eSM}{\end{smallmatrix}\right]}
\renewcommand*\env@matrix[1][*\c@MaxMatrixCols c]{%
  \hskip -\arraycolsep
  \let\@ifnextchar\new@ifnextchar
  \array{#1}}

\newqsymbol{`N}{\mathbb{N}}
\newqsymbol{`R}{\mathbb{R}}
\newqsymbol{`Z}{\mathbb{Z}}

\newqsymbol{`|}{\mid}
\newqsymbol{`8}{\infty}
\newqsymbol{`1}{\left}
\newqsymbol{`2}{\right}
\newqsymbol{`6}{\partial}
\newqsymbol{`0}{\emptyset}
\newqsymbol{`-}{\leftrightarrow}
\newcommand{\sgn}{\operatorname{sgn}}

\DeclareFontFamily{OMX}{MnSymbolE}{}
\DeclareSymbolFont{MnLargeSymbols}{OMX}{MnSymbolE}{m}{n}
\SetSymbolFont{MnLargeSymbols}{bold}{OMX}{MnSymbolE}{b}{n}
\DeclareFontShape{OMX}{MnSymbolE}{m}{n}{
    <-6>  MnSymbolE5
   <6-7>  MnSymbolE6
   <7-8>  MnSymbolE7
   <8-9>  MnSymbolE8
   <9-10> MnSymbolE9
  <10-12> MnSymbolE10
  <12->   MnSymbolE12
}{}
\DeclareFontShape{OMX}{MnSymbolE}{b}{n}{
    <-6>  MnSymbolE-Bold5
   <6-7>  MnSymbolE-Bold6
   <7-8>  MnSymbolE-Bold7
   <8-9>  MnSymbolE-Bold8
   <9-10> MnSymbolE-Bold9
  <10-12> MnSymbolE-Bold10
  <12->   MnSymbolE-Bold12
}{}

\let\llangle\@undefined
\let\rrangle\@undefined
\DeclareMathDelimiter{\llangle}{\mathopen} {MnLargeSymbols}{'164}{MnLargeSymbols}{'164}
\DeclareMathDelimiter{\rrangle}{\mathclose} {MnLargeSymbols}{'171}{MnLargeSymbols}{'171}

\DeclarePairedDelimiter\Set{\{}{\}}
\newcommand{\imod}[1]{\allowbreak\mkern10mu({\operator@font mod}\,\,#1)}

\newcommand{\threecols}[3]{
\hbox to \textwidth{%
      \normalfont\rlap{\parbox[b]{\textwidth}{\raggedright#1\strut}}%
        \hss\parbox[b]{\textwidth}{\centering#2\strut}\hss
        \llap{\parbox[b]{\textwidth}{\raggedleft#3\strut}}%
    }% hbox 
}

\newcommand{\reason}[2][\relax]{
  \ifthenelse{\equal{#1}{\relax}}{
    \left(\text{#2}\right)
  }{
    \left(\parbox{#1}{\raggedright #2}\right)
  }
}

\newcommand{\utag}[2]{\mathop{#2}\limits^{\text{(#1)}}}
\newcommand{\uref}[1]{(#1)}

\let\SavedDoubleVert\relax
\let\protect\relax
{\catcode`\|=\active
  \xdef\extendvert{\protect\expandafter\noexpand\csname extendvert \endcsname}
  \expandafter\gdef\csname extendvert \endcsname#1{\mskip-5mu \left.%
      \ifx\SavedDoubleVert\relax \let\SavedDoubleVert\|\fi
     \:{\let\|\SetDoubleVert
       \mathcode`\|32768\let|\SetVert
     #1}\:\right.\mskip-5mu}
}
\def\SetVert{\@ifnextchar|{\|\@gobble}% turn || into \|
    {\egroup\;\mid@vertical\;\bgroup}}
\def\SetDoubleVert{\egroup\;\mid@dblvertical\;\bgroup}

\begingroup
 \edef\@tempa{\meaning\middle}
 \edef\@tempb{\string\middle}
\expandafter \endgroup \ifx\@tempa\@tempb
 \def\mid@vertical{\middle|}
 \def\mid@dblvertical{\middle\SavedDoubleVert}
\else
 \def\mid@vertical{\mskip1mu\vrule\mskip1mu}
 \def\mid@dblvertical{\mskip1mu\vrule\mskip2.5mu\vrule\mskip1mu}
\fi

\makeatother

\usepackage{fouridx}
\usepackage{framed}
\usetikzlibrary{positioning,matrix}

\usepackage{paralist}
\usepackage{enumerate}

\usepackage[normalem]{ulem}

{\endMakeFramed}

\newenvironment{ybox}{
	\setlength{\FrameSep}{1.5mm}
	\setlength{\FrameRule}{0mm}
  \MakeFramed {\FrameRestore}}%
{\endMakeFramed}

\newenvironment{gbox}{
	\setlength{\FrameSep}{1.5mm}
\setlength{\FrameRule}{0mm}
  \MakeFramed {\FrameRestore}}%
{\endMakeFramed}

{\endMakeFramed}

 {\endMakeFramed}

\usepackage{enumitem}

\usepackage{etoolbox}

\let\theparentequation\theequation
\patchcmd{\theparentequation}{equation}{parentequation}{}{}

\renewenvironment{subequations}[1][]{%              optional argument: label-name for (first) parent equation
	\refstepcounter{equation}%
	\setcounter{parentequation}{\value{equation}}%    parentequation = equation
	\setcounter{equation}{0}%                         (sub)equation  = 0
	\def\theequation{\theparentequation\alph{equation}}% 
	\let\parentlabel\label%                           Evade sanitation performed by amsmath
	\ifx\\#1\\\relax\else\label{#1}\fi%               #1 given: \label{#1}, otherwise: nothing
	\ignorespaces
}{%
	\setcounter{equation}{\value{parentequation}}%    equation = subequation
	\ignorespacesafterend
}

\newcommand*{\nextParentEquation}[1][]{%            optional argument: label-name for (first) parent equation
	\refstepcounter{parentequation}%                  parentequation++
	\setcounter{equation}{0}%                         equation = 0
	\ifx\\#1\\\relax\else\parentlabel{#1}\fi%         #1 given: \label{#1}, otherwise: nothing
}

\DeclareMathOperator{\rank}{rank}

\newcommand{\wskc}{C_{\op{W}}}
\newcommand{\skc}{C_{\op{S}}}
\newcommand{\pkc}{C_{\op{P}}}
\newcommand{\rco}{R_{\op{CO}}}
\newcommand{\rl}{R_{\op{L}}}

\newcommand{\Fq}{\mathbb{F}_q}

\title{Wiretap Secret Key Agreement \\Via Secure Omniscience}
\author{Praneeth Kumar Vippathalla, Chung Chan, Navin Kashyap and Qiaoqiao Zhou
\thanks{N.\ Kashyap (nkashyap@iisc.ac.in) and Praneeth Kumar V.\ (praneethv@iisc.ac.in) are with the Department of Electrical Communication Engineering, Indian Institute of Science, Bangalore 560012. Their work was supported in part by a Swarnajayanti Fellowship awarded to N.\ Kashyap by the Department of Science \& Technology (DST), Government of India.}
\thanks{C.\ Chan (email: chung.chan@cityu.edu.hk) is with the Department of Computer Science, City University of Hong Kong. His work is supported by a grant from the University Grants Committee of the Hong Kong Special Administrative Region, China (Project No. 21203318).}
 \thanks{Q.\ Zhou (email: zhouqq@comp.nus.edu.sg) is with the Department of Computer Science, National University of Singapore.}
 \thanks{Corresponding author: C.\ Chan}
\thanks{This work was presented in part at the 2020 IEEE International Symposium on Information Theory, and in part at the 2021 IEEE International Symposium on Information Theory.}}

\setlist{leftmargin=1.2em}
\begin{document}
% \newif\ifPAGELIMIT
% \PAGELIMITfalse
%\PAGELIMITtrue
\IEEEoverridecommandlockouts

\maketitle
\begin{abstract}
    In this paper, we explore the connection between secret key agreement and secure omniscience within the setting of the multiterminal source model with a wiretapper who has side information. While the secret key agreement problem considers the generation of a maximum-rate secret key through public discussion, the secure omniscience problem is concerned with communication protocols for omniscience that minimize the rate of information leakage to the wiretapper. The starting point of our work is a lower bound on the minimum leakage rate for omniscience, $\rl$, in terms of the wiretap secret key capacity, $\wskc$. Our interest is in identifying broad classes of sources for which this lower bound is met with equality, in which case we say that there is a duality between secure omniscience and secret key agreement. We show that this duality holds in the case of certain finite linear source (FLS) models, such as two-terminal FLS models and pairwise independent network models on trees with a linear wiretapper. Duality also holds for any FLS model in which $\wskc$ is achieved by a perfect linear secret key agreement scheme. We conjecture that the duality in fact holds unconditionally for any FLS model. On the negative side, we give an example of a (non-FLS) source model for which duality does not hold if we limit ourselves to communication-for-omniscience protocols with at most two (interactive) communications.  We also address the secure function computation problem and explore the connection between the minimum leakage rate for computing a function and the wiretap secret key capacity.

\end{abstract} 

\begin{IEEEkeywords}
Information theoretic security, secret key generation, secure omniscience, leakage rate for omniscience, tree-PIN model, finite linear sources
\end{IEEEkeywords}

\section{Introduction} \label{sec:introduction}
In the setting of the multiterminal source model for secure computation, users who privately observe correlated random variables from a source try to compute functions of these private observations through interactive public discussion. The goal of the users is to keep these computed functions secure from a wiretapper who has some side information (a random variable possibly correlated with the source), and noiseless access to the public discussion. A well-studied problem within this model is that of secret key agreement, where users try to agree on a key that is kept secure from the wiretapper. In other words, users try to compute a common function that is independent of the public discussion and the wiretapper's side information. 

The secret key agreement problem was first studied for two users by Maurer \cite{maurer93}, and Ahlswede and Csisz\'ar \cite{ahlswedeCRpart1}. These works attempted to characterize the \emph{wiretap secret key capacity} $\wskc$, which is defined as the maximum secret key rate possible with unlimited public discussion. They were able to do this in certain special cases, for instance, in the case when only one user is allowed to communicate \cite[Theorem~1]{ahlswedeCRpart1}, and in the case when the wiretapper's side information is conditionally independent of one user's private information, given that of the other user \cite[Theorems~2 and 3]{maurer93}. In particular, when the wiretapper has no side information, $\wskc$ was shown to be equal to the mutual information between the random variables observed by the two users. But, for the two-user setting without additional assumptions, only upper and lower bounds on $\wskc$ were given. Subsequently, there have been multiple efforts, notably \cite{renner_wolf03,csiszar04,aminsource}, to strengthen and extend these bounds to the general setting of two or more users, but finding a single-letter expression remains a fundamental open problem in this domain.

In the course of extending the earlier results to the setting of multiple users, Csisz{\'a}r and Narayan \cite{csiszar04} gave a single-letter expression for the secret key capacity in the case when the wiretapper has no side information. They did this by establishing an equivalence or ``duality'' between the secret key agreement problem and the source coding problem of communication for \emph{omniscience}, which is attained when each user is able to recover (with high probability) the private observations of all the other users. They observed that a secret key of maximum rate can be extracted from a protocol that involves public discussion at the minimum rate required to attain omniscience. They were thus able to relate the secret key capacity to $\rco$, the minimum rate of communication required for omniscience, which can be obtained as the solution to a relatively simple linear program. 

Subsequently, Gohari and Anantharam \cite{aminsource} succeeded in establishing a similar duality in the more general setting of a wiretapper having side information. They showed an equivalence between the wiretap secret key agreement problem (in the presence of a wiretapper having side information) and a problem of communication for omniscience at a neutral observer. In the latter problem, there is (in addition to the users and the wiretapper) a neutral observer who is given access to the wiretapper's side information. The goal here is for the users to communicate in public to create a shared random variable which when provided to the neutral observer, allows the observer to reconstruct all the users' private observations. Theorem~3 of \cite{aminsource} relates $\wskc$ to the minimum rate of public communication required for omniscience at the neutral observer. However, this does not lead to a single-letter expression for $\wskc$, as it is not known how to compute the minimum rate of communication for omniscience at the neutral observer.

Motivated in part by the results of \cite{csiszar04} and \cite{aminsource}, we explore the possibility of an alternative duality existing between the wiretap secret key problem and a certain \emph{secure omniscience} problem, in the hope of obtaining additional insight on $\wskc$, potentially leading to its evaluation in settings where it still remains unknown. In the secure omniscience problem we consider, we stay within the original setting of the multiterminal source model with a wiretapper having side information. The users communicate interactively in public so as to attain omniscience, but now the aim is not necessarily to minimize the rate of communication needed for this. Instead, the goal is to minimize the rate at which the communication for omniscience leaks information about the source to the wiretapper. We give the formal definition of $\rl$, the minimum information leakage rate of any communication for omniscience, in Section~\ref{sec:problem}. 

%We try to make inroads into this problem by connecting secret key generation with secure source coding namely secure omniscience. In the secure omniscience problem, every user tries to attain omniscience by communicating interactively using their private observations from a correlated source; however, the goal here is to minimize the information leakage rate $\rl$ to a wiretapper who has side information about the source. This is motivated by the work of Csisz\'{a}r and Narayan \cite{csiszar04} on the multiterminal setting with no wiretapper side information, where they gave a single-letter expression for the secret key capacity by connecting it with source coding problem of communication for omniscience. They observed that it is enough for the users to attain omniscience by communicating at the minimum rate possible to extract a secret key of maximum rate. This kind of a duality has been further explored in the case with wiretapper side information by Gohari and Anantharam, in \cite{aminsource}. They proved a duality between secret key agreement and the problem of communication for omniscience by a neutral observer, where the neutral observer attains omniscience given the common function agreed by the users and the wiretapper side information. We would like to explore an alternate duality without the neutral observer by relating $\wskc$ and $\rl$. If we are able to establish a duality between these two problems, then we can use an optimum secure omniscience protocol to achieve the wiretap secret key capacity.

\subsection{Main Contributions}

The starting point of our paper is an inequality that relates the wiretap secret key capacity and the minimum leakage rate for omniscience for a source $(\RZ_V,\RZ_{\opw})$. Here, $V:=\{1,\ldots, m\}$ denotes the set of users, $\RZ_V:=(\RZ_i \mid i \in V)$ is the collection of user observations, and $\RZ_{\opw}$ denotes the wiretapper's side information. We then have 
     \begin{equation}
      H(\RZ_V|\RZ_{\opw}) - \wskc \leq \rl ,
      \label{basic_ineq}
        \end{equation}
%The minimum leakage rate for omniscience is upper bounded by $H(\RZ_V|\RZ_{\opw})$.
  where $\rl$ is the information leakage rate, which we formally define in Section~\ref{sec:problem} as $\limsup_n\frac{1}{n}I(\RF^{(n)} \wedge \RZ_V^n|\RZ_{\opw}^n)$, minimized over all communications $\RF^{(n)}$ for omniscience.
The inequality follows from a standard argument: once the users attain omniscience via a communication protocol that achieves the minimum leakage rate $\rl$, they can extract a secret key of rate $H(\RZ_V | \RZ_{\opw}) - \rl$ from the reconstruction of $\RZ_V$ available to each of them. 
%We give the simple proof of this inequality in Section~\ref{sec:problem}.

If the inequality in \eqref{basic_ineq} holds with equality, then we refer to it as a duality between secure omniscience and wiretap secret key agreement.  Essentially, whenever this duality holds, a secret key of maximum rate can be extracted from a communication for omniscience protocol that minimizes the leakage rate. Note that equality in \eqref{basic_ineq} yields an expression for $\wskc$ in terms of $\rl$, but its utility towards computing $\wskc$ is unclear, as it is not known whether $\rl$ admits a single-letter expression.

We first address the question of whether there is always a duality between secure omniscience and wiretap secret key agreement for any multiterminal source model with wiretapper. Note that if equality holds in \eqref{basic_ineq}, then it must be the case that $\wskc = 0$ iff $\rl = H(\RZ_V | \RZ_{\opw})$. Now, it is easily shown that, for any multiterminal source model, $\wskc = 0$ implies $\rl = H(\RZ_V | \RZ_{\opw})$. This follows directly from \eqref{basic_ineq} and the upper bound $\rl \le H(\RZ_V | \RZ_{\opw})$, which always holds, as is easily seen from the definition of $\rl$ --- see Theorem~\ref{thm:RL:lb} in Section~\ref{sec:problem}. It is not so clear whether the converse is also true, namely, that $\rl = H(\RZ_V | \RZ_{\opw})$ implies $\wskc = 0$. We conjecture that  the converse does not always hold, i.e., there are sources for which $\rl = H(\RZ_V | \RZ_{\opw})$, yet $\wskc > 0$. We make partial progress in this direction by showing that this is the case if we restrict ourselves to omniscience protocols in which at most two communications are allowed. We give an example of a two-user source model for which $\wskc > 0$, but the leakage rate equals $H(\RZ_V | \RZ_{\opw})$ for any omniscience protocol involving at most two messages. While our example does not definitively resolve the issue of duality between secure omniscience and wiretap secret key agreement, it seems to indicate that this duality may not always hold.

Next, we consider a broad class of sources, namely, \emph{finite linear sources}, for which we believe the duality must hold. In a finite linear source (FLS) model, each user's observations, as well as the wiretapper's side information, is given by a linear transformation of an underlying random vector consisting of finitely many i.i.d.\ uniform random variables. This class of sources has received some prior attention \cite{chan11itw,chan11delay,chan19oneshot}. We prove that \eqref{basic_ineq} holds with equality for FLS models in which the wiretap secret key capacity $\wskc$ is achieved by a \emph{perfect} key agreement protocol involving public communications that are linear functions of the users' observations. It is an open question as to whether $\wskc$ can always be achieved through linear communication protocols for any FLS model, but it is reasonable to expect that this is the case. We also give two unconditionally positive results: duality holds in the case of two-user FLS models, and in the case of pairwise independent network (PIN) models on trees \cite{sirinpin, sirinperfect} in which the wiretapper's side information is a linear function of the source. In both these cases, we obtain explicit expressions for $\rl$ and $\wskc$. In fact, in the case of tree-PIN models with a linear wiretapper, we are able to explicitly determine the maximum secret key rate achievable when the total rate of public communication is constrained to be at most $R$.

{ Finally, we consider the problem of secure function computation, where users try to compute a given function $\RG=g(\RZ_1,\ldots, \RZ_m)$ without revealing much information about the computed value to the wiretapper. For this problem, we generalize the inequality \eqref{basic_ineq} to $H(\RG|\RZ_{\opw}) - \wskc \leq \rl^{\RG}$, where $\rl^{\RG}$ is the minimum information leakage rate for computing the function $\RG$, which we define as $\limsup_n\frac{1}{n}I(\RF^{(n)} \wedge \RG^n|\RZ_{\opw}^n)$, minimized over all communications $\RF^{(n)}$ for computing $\RG$. We also give the conditions under which the previous inequality holds with equality. Furthermore, for finite linear sources, we prove a result that says that any information leakage rate in computing a linear function can also be achieved by an omniscience scheme.}

\subsection{Related Work}

Our work is closely related to that of Prabhakaran and Ramchandran \cite{vinod07}. In their work, they considered the problem of secure source coding in a two-user model with a wiretapper where only one user is allowed to communicate to the other. This kind of communication is commonly referred to as \emph{one-way communication}. The goal here is to communicate in such a way that the receiving user recovers the observations of the transmitting user while minimizing the rate of information leaked to the wiretapper about the transmitting user's source. In this case, they  obtained a single-letter characterization of the minimum leakage rate for recovering one terminal's observation by the other terminal by using conventional information-theoretic techniques. Moreover, they used this quantity to lower bound the wiretap secret key capacity. Our work, in fact, generalizes this result by considering the minimum leakage rate for omniscience instead in the multi-user setting where interactive communication is allowed.

The secure source coding problem considered in \cite{vinod07},  has been generalized and studied extensively in the direction of characterizing the minimum rate of leakage of transmitter's source \cite{gunduz08, villard13} by incorporating various constraints. For instance, Villard and Piantanida \cite{villard13} considered a similar model as in \cite{vinod07}, but the receiving user observes  coded side information from a third party. Since uncoded side information is a special case of coded side information, this framework subsumes the model of \cite{vinod07}. For this model, they studied the problem in a broad generality by considering a lossy recovery of the transmitter's observations at the receiving terminal in the presence of a wiretapper. They gave a characterization of the rate-distortion-leakage rate region which is the set of all achievable tuples of communication rate, distortion and leakage rate.

Recently, in \cite{wenwentu19},  Tu and Lai have considered the same model but studied the problem of lossy function computation by the receiving terminal, which is a further generalization of the model of \cite{villard13}. They considered even the privacy aspect (leakage of the transmitting user's source to the receiving user) and studied it along with the rate-distortion-leakage rate region. They were able to give an explicit characterization of the entire achievable rate region.

This problem falls in the class of source coding for distributed function computation; see, for e.g., \cite{han_dichotomy, alon_roche, ma11, tyagi11, wenwentu19}. In this problem, each user has access to a private random variable, and they wish to compute functions of these private random variables by communicating in public, possibly interactively or/and in the presence of a wiretapper. For instance, in \cite{ma11},  Ma and Ishwar have considered a two-user model without a wiretapper, where users, after observing private random variables, interactively communicate to compute functions of these private random variables. They studied the interactive communication rates needed for the computation of functions and completely characterized the rate region. Subsequently, this work has been extended  by \cite{yassaee15} for randomized function computation in the two-user case. Recently, \cite{data_interactive_securefunction} has studied the randomized function computation even by including privacy constraints on the users' observation.

One work that studies the function computation in the context of  multi-user source model with a wiretapper is \cite{tyagi11}. In their work, Tyagi, Narayan, and Gupta assumed that the wiretapper has no side information and addressed the question: when can a common function be computed securely? Here we say a function is \emph{securely computable} if it is kept asymptotically independent of the communication that is needed to compute this function. It means that the wiretapper can gain almost no knowledge of the function output even with access to the communication. They answered this question by 
relating it with the secret key capacity of the source model. The precise result is that a common function is securely computable by all the terminals if and only if  the entropy of the function is less than the secret key capacity. 

Secure omniscience is also a problem of source coding for distributed function computation. Here, all the users try to recover the users' source, and the quantity of interest is the minimum rate of information about the source that gets leaked to the wiretapper through the communication. A problem that is closely related to secure omniscience is the coded cooperative data exchange (CCDE) problem with a secrecy constraint; see, for e.g., \cite{sprinston13, courtade16}. In the problem of CCDE, we consider a hypergraphical source and study one-shot omniscience. The hypergraphical model generalizes the PIN model within the class of FLSs. \cite{courtade16} studied the secret key agreement in the CCDE context and characterized the number of transmissions required versus the number of SKs generated. On the other hand, \cite{sprinston13} considered the same model but with wiretapper side information and explored the leakage aspect of an omniscience protocol. However, the security notion considered therein does not allow the eavesdropper to recover even one hyperedge of the source from the communication except what is already available. However, the communication scheme can still reveal information about the source. In this paper, we are interested in minimizing the rate of information leakage to the wiretapper. Though we consider the asymptotic notion, the  designed optimal communication scheme uses only a finite number of realizations of the source. Hence our scheme can find applications even in CCDE problems. 

The role of omniscience in the multi-user secret key agreement (with wiretapper side information) was highlighted in the work of Csisz{\'a}r and Narayan \cite{csiszar04}. They showed that a maximum key rate could be achieved  by communicating at a minimum rate for omniscience. This led to the question of whether the omniscience is optimal even in terms of the minimum communication rate needed to achieve secret key capacity. The works \cite{chan18, mukherjee16} have addressed this question by giving sufficient conditions for general sources and equivalent conditions for hypergraphical sources. 

Though the characterization of secret key capacity (without wiretapper side information) is known, and its connection with omniscience is well studied, the characterization of wiretap secret key capacity is still an open problem. Results are known only in special sources \cite{ahlswedeCRpart1, maurer93}. However, there has been some progress in this direction in recent times.  For instance,  Gohari, G\"{u}nl\"{u} and Kramer, in \cite{amin2020}, sought for the characterization of the class of two-user sources for which wiretap secret key capacity is positive. They were able to find an equivalent characterization in terms of R\'enyi divergence. Its usefulness has been demonstrated on  sources with an erasure model on the wiretapper side information   by deriving a sufficient condition for the positivity of $\wskc$. In the direction of characterizing $\wskc$,  Poostindouz and Safavi-Naini, in \cite{alireza19}, have made an effort in the case of some special source models. In particular, they considered tree-PIN models with a  wiretapper side information containing  noisy versions of the edge random variables. They obtained a characterization of $\wskc$ in terms of the conditional minimum rate of communication for omniscience which is a solution to a certain linear program.

% Despite all these efforts,  to our knowledge, there is no work that addresses the relationship between secure omniscience and wiretap secret key agreement for multiterminal setting with interactive communication allowed.

\subsection{Organization}
This paper is organized as follows. In Section~\ref{sec:problem}, we introduce the problem and notations. In this section, we also establish an inequality relating the minimum leakage rate for omniscience and wiretap secret key capacity for general source models. Section~\ref{sec:counterexamaple_duality} contains an example showing that the duality does not hold between secure omniscience and secret key agreement in the case of limited interaction (with two messages allowed). This result suggests that the duality need not hold in the general case. In Section~\ref{sec:duality_fls}, we first formally define the finite linear source models and prove a duality result concerning linear protocols. Furthermore, we establish an unconditional result in the two-user FLS. In Section~\ref{sec:treepin}, we prove the duality in the case of the tree-PIN model with linear wiretapper. Moreover, for this model,  we determine the rate region containing all achievable secret key rate and total communication rate pairs. In fact, we use a secure omniscience scheme for a part of the source to obtain this result. 
%In Section~\ref{sec:positivity}, we obtain some equivalent conditions for the positivity of $\wskc$ for multi-user case using \eqref{basic_ineq}. This generalizes the two-user result of \cite{amin2020}. 
{ In Section~\ref{sec:funcomp}, we introduce the problem of secure function computation and derive an inequality between the minimum leakage rate for computing a function and the wiretap secret key capacity. Furthermore, in this section, we explore the conditions under which this inequality holds with equality.} Finally, we discuss the open problems and challenges in establishing duality in Section~\ref{sec:discussion}.

\section{Problem formulation}\label{sec:problem}
In this section, we describe two different scenarios, namely wiretap secret key agreement and secure omniscience, in the context of the multiterminal source model. In this model,  the terminals communicate publicly using their correlated observations to compute functions securely from the eavesdropper, who has access to the public communication along with some side information.  More precisely, let $V=[m]:=\left\lbrace1, \ldots, m\right\rbrace$ be the set of users,  and let $\opw$ denote the wiretapper.  Let  $\RZ_1,\ldots, \RZ_m$ and $\RZ_{\opw}$ be the random variables  taking values in finite alphabets $\mc{Z}_1,\ldots, \mc{Z}_m$ and $\mc{Z}_{\opw}$ respectively, and their joint distribution is given by $P_{\RZ_1 \ldots \RZ_m \RZ_{\opw}}$. Let $\RZ_V := (\RZ_i: i \in V)$ and $\RZ_i^n$ denote the $n$ i.i.d. realizations  of $\RZ_i$. For $i \in V$, user $i$ has access to the random variable $\RZ_i$, and the wiretapper observes $\RZ_{\opw}$. Upon observing $n$ i.i.d. realizations, the users communicate interactively using their observations, and possibly independent private randomness, on the noiseless and authenticated channel. In other words, the communication made by a user in any round depends on all the previous rounds' communications and  the user's own observations.  Let $\RF^{(n)}$ denotes this interactive communication. We say $\RF^{(n)}$ is \emph{non-interactive}, if it is of the form $(\tRF_i^{(n)}: i \in V)$, where $\tRF_i^{(n)}$ depends only on $\RZ_i^n$ and the  private randomness of user $i$. Note that the eavesdropper has access to the pair $(\RF^{(n)}, \RZ_{\opw}^n)$. At the end of the communication, each user outputs a value in a finite set using its observations and $\RF^{(n)}$. For example, user $i$ outputs $\RE_i^{(n)}$ using $(\RF^{(n)}, \RZ_i^n)$ and its private randomness { $\RS_i$, i.e., $\RE_i^{(n)} = \psi_i(\RF^{(n)}, \RZ_i^n, \RS_i)$ for some function $\psi_i$}. See Fig.~\ref{fig:system}.
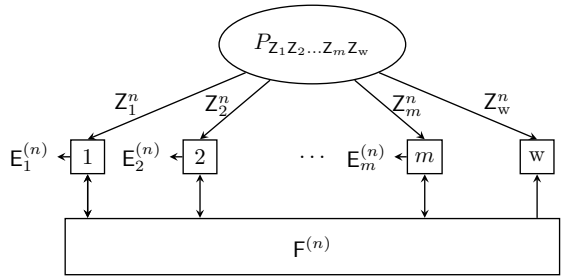
\begin{figure}[h]
\centering
\resizebox{0.85\width}{!}{\begin{tikzpicture}[>=stealth,semithick, node distance = 5 em]
\tikzstyle{source}=[ellipse, draw, semithick, minimum height=3.5 em, minimum width=7.5 em];
\tikzstyle{user}=[rectangle, draw, semithick, minimum height=1.5 em, minimum width=1.5 em];
\tikzstyle{chnl}=[rectangle, draw, semithick, minimum height=2.5 em, minimum width=22 em];
\node   (dots)   {$\ldots$};
\node[source]      (source)       [above of = dots] {$P_{\RZ_1\RZ_2\ldots \RZ_m\RZ_{\opw}}$};
\node[user]      (2)       [left of = dots] {$2$};
\node[user]      (1)       [left of = 2] {$1$};
\node[user]      (m)       [right of = dots] {$m$};
\node[user]      (w)       [right of = m] {${\opw}$};
\node[chnl]      (chnl)       [below of =dots, yshift=1 em] {$\RF^{(n)}$};
\node (k1) [left of = 1, xshift=2.4 em] {$\RE_1^{(n)}$};
\node (k2) [left of = 2, xshift=2.4 em] {$\RE_2^{(n)}$};
\node (km) [left of = m, xshift=2.4 em] {$\RE_m^{(n)}$};
\draw [->] (source) edge node[above,near end] {$\RZ_1^n$} (1.north);
\draw [->] (source) edge node[above,near end]{$\RZ_2^n$}(2.north);
\draw [->] (source) edge node[above,near end]{$\RZ_m^n$} (m.north);
\draw [->] (source) edge node[above,near end]{$\RZ_{\opw}^n$} (w.north);
\draw [<->]  (1.south)--(1|-chnl.north);
\draw [<->]  (2.south)--(2|-chnl.north);
\draw [<->]  (m.south)--(m|-chnl.north);
\draw [->]  (w|-chnl.north)--(w.south);
\draw [<->]  (1.south)--(1|-chnl.north);
\draw [->]  (1.west)--(k1.east);
\draw [->]  (2.west)--(k2.east);
\draw [->]  (m.west)--(km.east);
\end{tikzpicture}}
\caption{Multiterminal source model with wiretapper side information. The terminals interactively discuss over a public channel using their observations from a correlated source to  compute their respective functions.}
\label{fig:system}
 \end{figure}
\subsection{Secure Omniscience}\label{subsec:omniscience}
In the secure omniscience scenario, each user tries to recover the observations of all the users other than the wiretapper. We say that $(\RF^{(n)}, \RE_1^{(n)}, \ldots, \RE_m^{(n)})_ {n \geq 1}$  is an \emph{omniscience scheme} if it satisfies the recoverability condition for omniscience:
\begin{align}\label{eq:omn:recoverability}
\liminf_{n \to \infty} \Pr(\RE_1^{(n)} = \dots =\RE_m^{(n)} = \RZ_V^n) = 1.
\end{align}
The communication $\RF^{(n)}$ in an omniscience scheme is called an \emph{omniscience communication}.

The \emph{minimum leakage rate for omniscience} is defined as 
\begin{align}
\begin{split}
 \rl&:= \inf  \biggl\lbrace \limsup_{n \to \infty} \frac{1}{n}I(\RF^{(n)} \wedge \RZ_V^n|\RZ_{\opw}^n) \biggr\rbrace, \label{eq:rl}
 \end{split}
\end{align}
where the infimum is over all omniscience schemes. We sometimes use $\rl(\RZ_V\|\RZ_{\opw})$ instead of $\rl$ to make the source explicit. The \emph{minimum rate of communication for omniscience} $\rco(\RZ_V)$, or simply $\rco$, is defined as \cite{csiszar04}
  \begin{align}
    \begin{split}
 \rco&:= \inf  \biggl\lbrace \limsup_{n \to \infty} \frac{1}{n}\log|\mc{F}^{(n)}| \biggr\rbrace, \label{eq:rco_def}
 \end{split}
\end{align}
where $\mc{F}^{(n)}$ is the range of $\RF^{(n)}$, and the infimum is over all omniscience schemes. It is known \cite[Proposition~1]{csiszar04} that $\rco$ is given by the solution to the following linear program:
\begin{align*}
    \rco = \min\left\{\sum \limits_{i \in V} R_i \Bigm\vert \sum \limits_{i \in B}R_i \geq H(\RZ_B|\RZ_{B^c}), \: \forall B\subsetneq V  \right\}.
    \label{rco_expression}
\end{align*}

The \emph{conditional} minimum rate of communication for omniscience, $\rco(\RZ_V|\RJ)$, is used in situations where all the users have access to a common random variable $\RJ^n$ along with  their private observations. This means that user $i$ observes $(\RJ^n, \RZ_i^n)$. 

Observe that for any source, we have
\begin{align}\label{eq:rl_rco}
    \rl  \leq \rco,
\end{align}
which follows easily from \eqref{eq:rl} and \eqref{eq:rco_def} as $I(\RF^{(n)} \wedge \RZ_V^n|\RZ_{\opw}^n)\leq H(\RF^{(n)}) \leq \log|\mc{F}^{(n)}|$, where $\RF^{(n)}$ is an omniscience communication taking values in the set $\mc{F}^{(n)}$. 

To get a sense of how $\rl(\RZ_V\|\RZ_{\opw})$ behaves with respect to the correlation between $\RZ_{\opw}$ and $\RZ_V$, consider another wiretapper side information $\tRZ_{\opw}$ that is less correlated with $\RZ_{V}$ than  $\RZ_{\opw}$, in the sense that they form the Markov chain $\tRZ_{\opw} \textrm{ -- } \RZ_{\opw} \textrm{ -- } \RZ_V$. Using the data processing inequality and the fact that any omniscience communication $\RF^{(n)}$ is a function only of $\RZ_V^n$ and private randomness (which is independent of $\RZ_V^n, \RZ_{\opw}^n, \tRZ_{\opw}^n$), we can infer that $I(\RF^{(n)} \wedge \RZ_V^n|\RZ_{\opw}^n)\leq I(\RF^{(n)} \wedge \RZ_V^n|\tRZ_{\opw}^n)$. We conclude, via \eqref{eq:rl} and \eqref{eq:rl_rco}, that
\begin{align*}
    \rl(\RZ_V\|\RZ_{\opw})\leq \rl(\RZ_V\|\tRZ_{\opw}) \leq \rco(\RZ_V).
\end{align*}
At the end of the next subsection, we will show that when $\RZ_{\opw}$ is independent of $\RZ_V$, \eqref{eq:rl_rco} holds with equality.

\subsection{Wiretap Secret Key Agreement}\label{sec:wska:def}
In the wiretap secret key agreement, each user tries to compute a common function, which is called a \emph{key}, that is kept secure from the wiretapper. Specifically, we say that $(\RF^{(n)}, \RE_1^{(n)}, \ldots, \RE_m^{(n)})_ {n \geq 1}$  is a \emph{wiretap secret key agreement (SKA) scheme} if there exists a sequence $(\RK^{(n)})_{n \geq 1}$  such that
\begin{subequations}
\label{eq:sk:constraints}
\begin{align}
\liminf_{n \to \infty} \Pr(\RE_1^{(n)} = \dots =\RE_m^{(n)} = \RK^{(n)}) = 1 \label{eq:sk:recoverability},\\
\limsup_{n \to \infty}\left[\log |\mc{K}^{(n)}| - H(\RK^{(n)}| \RF^{(n)},\RZ_{\opw}^n)\right] =0\label{eq:sk:secrecy},
\end{align}
where $|\mc{K}^{(n)}|$ denotes the cardinality of the range of $\RK^{(n)}$. Conditions \eqref{eq:sk:recoverability} and \eqref{eq:sk:secrecy} are referred to as the key recoverability condition and the secrecy condition of the key, respectively. 
\end{subequations}
The \emph{wiretap secret key capacity}  is defined as
\begin{align}
 \wskc:= \sup \left\lbrace \liminf_{n \to \infty} \frac{1}{n} \log |\mc{K}^{(n)}| \label{eq:wskc}\right\rbrace
\end{align}
where the supremum is over all SKA schemes. The quantity $\wskc$ is also sometimes written as $\wskc(\RZ_V\|\RZ_{\opw})$. In \eqref{eq:wskc}, we use $\skc$ instead of $\wskc$, when the wiretap side information is set to a constant.   Similarly, we use $\pkc(\RZ_V| \RJ)$  in the case when wiretap side information is  $\RZ_{\opw}= \RJ$ and all the users have the shared random variable $\RJ$ along with  their private observations $\RZ_i$. The quantities $\skc$ and $\pkc(\RZ_V|\RJ)$ are referred to  as \emph{secret key capacity} of  $\RZ_V$, and \emph{private key capacity} of $\RZ_V$ with compromised-helper side information $\RJ$ respectively. 

The following theorem gives lower and upper bounds on the minimum leakage rate for omniscience for a general source $(\RZ_V,\RZ_{\opw})$. 

\begin{theorem}\label{thm:RL:lb}
    For a general source $(\RZ_V,\RZ_{\opw})$,
      \begin{align}
      H(\RZ_V|\RZ_{\opw}) - \wskc \leq \rl \leq H(\RZ_V|\RZ_{\opw}).\label{eq:RL:lb}
        \end{align}
\end{theorem}
 \begin{IEEEproof}[Proof sketch]
The upper bound on $\rl$ follows from \eqref{eq:rl}, upon noting that  $\frac{1}{n}I(\RF^{(n)} \wedge \RZ_V^n|\RZ_{\opw}^n)\leq H(\RZ_V|\RZ_{\opw})$. 
For the lower bound, the underlying idea of the proof is that given a discussion scheme that achieves $\rl$, one can apply privacy amplification to extract a secret key of rate $H(\RZ_V|\RZ_{\opw})-\rl$ from the recovered source. The details of this argument are given in Appendix~\ref{app:thm:proof_rl_bound}.
\end{IEEEproof}

% Given a discussion scheme that achieves $\rl$, one can apply privacy amplification~\cite[Lemma~B.2]{csiszar04} to extract a secret key of rate $H(\RZ_V|\RZ_{\opw})-\rl$ from the recovered source. Since the secret key rate thus achieved is bounded above by $\wskc$, we obtain the lower bound on $\rl$. The upper bound on $\rl$ follows from \eqref{eq:rl}, upon noting that  $\frac{1}{n}I(\RF^{(n)} \wedge \RZ_V^n|\RZ_{\opw}^n)\leq H(\RZ_V|\RZ_{\opw})$.

\begin{remark}\renewcommand{\qed}{}
Note that the achievable key rate is, intuitively, the total amount of randomness in the recovered source $\RZ_V$ that is not in the wiretapper's side information $\RZ_{\opw}$ nor revealed in public. 
\end{remark}

% Therefore, we have 
% \begin{align*}
%     \rl  \leq {\min} \{\rco, H(\RZ_V|\RZ_{\opw})\}.
% \end{align*}

In Theorems~1-3 of \cite{csiszar04}, Csisz\'ar and Narayan showed that $\skc(\RZ_V)=H(\RZ_V)-\rco(\RZ_V)$ and $ \pkc(\RZ_V | \RZ_{\opw}) = H(\RZ_V|\RZ_{\opw}) - \rco(\RZ_V|\RZ_{\opw})$. They also proved, in Theorem~4 of \cite{csiszar04}, that $\wskc(\RZ_V \| \RZ_{\opw}) \leq \min\{\skc(\RZ_V),\pkc(\RZ_V | \RZ_{\opw})\}$, which implies that $\rl(\RZ_V \| \RZ_{\opw}) \stackrel{\eqref{eq:RL:lb}}{\geq} H(\RZ_V|\RZ_{\opw}) - \wskc (\RZ_V \| \RZ_{\opw}) \geq H(\RZ_V|\RZ_{\opw}) - \pkc (\RZ_V | \RZ_{\opw})=\rco(\RZ_V|\RZ_{\opw})$. By combining this inequality and \eqref{eq:rl_rco}, we get
$\rco(\RZ_V|\RZ_{\opw})\leq \rl(\RZ_V \| \RZ_{\opw}) \leq \rco(\RZ_V)$.

When $\RZ_{\opw}$ is independent of $\RZ_V$, as a straightforward consequence of the definition of $\wskc(\RZ_V\|\RZ_{\opw})$, we have $\wskc(\RZ_V\|\RZ_{\opw})=\skc(\RZ_V)$. Therefore, we see that $$\rco(\RZ_V) = H(\RZ_V)-\skc(\RZ_V)=H(\RZ_V|\RZ_{\opw}) - \wskc(\RZ_V\|\RZ_{\opw}) \stackrel{\eqref{eq:RL:lb}}{\leq} \rl(\RZ_V\|\RZ_{\opw}) \stackrel{\eqref{eq:rl_rco}}{\le} \rco(\RZ_V).$$ Thus, the upper bound in \eqref{eq:rl_rco} and the lower bound in \eqref{eq:RL:lb} in fact hold with equality in this case.

\section{Duality between secure omniscience and wiretap secret key agreement: Limited interaction}\label{sec:counterexamaple_duality}
In this section, we address the question of whether there is always a duality between secure omniscience and wiretap secret key agreement for any multiterminal source model with wiretapper. We study this by considering a necessary condition for duality, which is $\wskc > 0$ iff $\rl < H(\RZ_V|\RZ_{\opw})$. One direction, namely, that $\rl < H(\RZ_V|\RZ_{\opw})$ implies $\wskc > 0$ holds for any  source follows from \eqref{eq:RL:lb}. For the other direction, intuitively, if the users can generate a secret key that is independent of the wiretapper's side information, then they can use this advantage to protect some information during an omniscience scheme. However, we will prove that this need not be the case if we limit the number of messages exchanged between the users. 

%Now we will address the question: Is the ability to generate a positive secret key rate equivalent to that the users can achieve omniscience by strictly protecting a part of the source? More precisely, does the statement $\wskc > 0$ iff $\rl < H(\RZ_V|\RZ_{\opw})$ hold? One direction that $\rl < H(\RZ_V|\RZ_{\opw})$ implies $\wskc > 0$ follows from the lower bound \eqref{eq:RL:lb} which uses the idea of privacy amplification of the recovered source.  

To illustrate this result, let us consider a two-user setting ($m=2$) with source distribution $P_{\RZ_1\RZ_2\RZ_{\opw}}$. Let $r$ be the number of messages exchanged between the users, and let $\wskc^{(r)}$ and $\rl^{(r)}$ denote the wiretap secret key capacity and the minimum leakage rate for omniscience, respectively, when we allow at most $r$ messages to be exchanged among the users. Note that we can ensure omniscience only if we allow $r \geq 2$ because omniscience is not guaranteed with one message transmission. Moreover, omniscience can be obtained using a non-interactive communication that involves only $2$ messages.  Here $\rl^{(r)}<H(\RZ_1,\RZ_2|\RZ_{\opw})$ implies $\wskc^{(r)}>0$, because if the users can achieve omniscience using $r$ messages such that $\rl^{(r)}<H(\RZ_1,\RZ_2|\RZ_{\opw})$, then they can apply privacy amplification to recover a key with positive rate implying $\wskc^{(r)} > 0$. For the other direction, we show that $\wskc^{(r)}>0$ does not imply $\rl^{(r)}<H(\RZ_1,\RZ_2|\RZ_{\opw})$ if $r=2$. This is stated in the following proposition.

\begin{proposition} \label{prop:positivity} 

If $r=2$, then for any source $P_{\RZ_1\RZ_2\RZ_{\opw}}$,
    $$\rl^{(r)}<H(\RZ_1,\RZ_2|\RZ_{\opw}) \Longrightarrow \wskc^{(r)}>0.$$
    However, the converse need not hold. %In particular, for the source given in Lemma~\ref{lem:twowaycounter}, $\wskc^{(r)}> 0$ but  $\rl^{(r)}=H(\RZ_1,\RZ_2|\RZ_{\opw})$.

    % \item If $r\geq 3$, then for any source $P_{\RZ_1\RZ_2\RZ_{\opw}}$,
    % $$\rl^{(r)}<H(\RZ_1,\RZ_2|\RZ_{\opw}) \iff \wskc^{(r)}>0.$$

\end{proposition}

%The above proposition hints that this should be the case even with an arbitrary but fixed number of messages and the unlimited number of messages. 

%The following theorem summarizes this result for two user case on the interplay between the positivity of secret key capacity and the non-maximality of minimum leakage rate for omniscience.

For the converse part, we first derive an upper bound on $\rl^{(2)}$ using the results from the one-way communication setting. We then give a source in Lemma~\ref{lem:twowaycounter} that serves as a counterexample to illustrate that the converse does not hold in general. In the rest of this section, we denote $\RZ_1,\RZ_2$ and $\RZ_{\opw}$ by $\RX, \RY$, and $\RZ$, respectively. The random variables $\RX, \RY$, and $\RZ$ take values in  finite sets $\mc{X}$, $\mc{Y}$, and $\mc{Z}$, respectively.

\subsection{One-way communication, i.e., $r=1$}
Before we address the problem completely, first, we consider a model with only one message allowed. Since omniscience requires a minimum of two messages between users, we slightly modify the setup by letting only one of the users recover the other user's observations\textemdash see Fig. \ref{fig:oneway}. We define the minimum leakage rate for recovery of $\RX$ by user $2$ as 
$$\rl^{\text{ow}}:= \inf  \biggl\lbrace \limsup_{n \to \infty} \frac{1}{n}I(\RF_1^{(n)} \wedge \RX^n|\RZ^n) \biggr\rbrace,$$ where the infimum is over all one-way communication schemes that allow user $2$ to recover $\RX$. 
% Furthermore, the definition of one-way wiretap secret key capacity, denoted by $\wskc^{\text{ow}}$,  is the same as \eqref{eq:wskc} with the exception that the supremum is taken over all one-way SKA schemes. 

\begin{figure}[h]
\centering
\resizebox{0.9\width}{!}{\tikzstyle{block}=[rectangle, draw, thick, minimum width=2em, minimum height=2em]

\begin{tikzpicture}[node distance=4cm,auto,>=latex']

    \node (f) {};
    \node [block] (a) [left of = f, node distance=1.5cm] {1};
    \node (x) [above of = a, node distance=1.25cm] {$\RX^n$};
    \node [block] (b) [right of=f, node distance=1.5cm] {2};
    \node (y) [above of = b, node distance=1.25cm] {$\RY^n$};
    \node [block] (w) [below of = f, node distance=1.5cm] {W};
    \node (z) [left of = w, node distance=1.25cm] {$\RZ^n$};
    \node (e1) [left of = a, node distance=1.5cm] {$\RE_1^{(n)}$};
    \node (e2) [right of = b, node distance=1.5cm] {$\RE_2^{(n)}$};

    \draw[->, thick] (x) --  (a);
    \draw[->, thick] (a) --  (e1);
    \draw[->, thick] (y) --  (b);
    \draw[->, thick] (b) --  (e2);
   \draw[->, thick] (a) -- (b) node [midway, above] {$\RF_1$};
%    \draw[->] (f.center) --  (w);
    \draw[->,thick] (z) -- (w);
\end{tikzpicture}}
\caption{Only one message transfer is allowed.  Since omniscience is, in general, not possible within this setup, we only require user $2$ to recover user $1$'s observations, i.e., $\RE_1^{(n)}$ is  constant and $\RE_2^{(n)}= \hat{\RX}^{(n)}$.}
\label{fig:oneway}
 \end{figure}
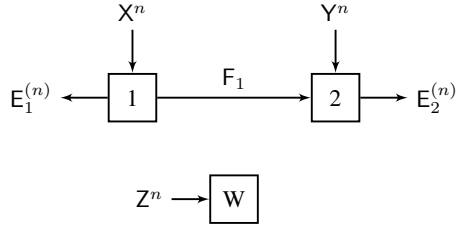
% Ahlswede and Csisz\'ar, in \cite{ahlswedeCRpart1}, studied the one-way wiretap secret key agreement, and gave a single-letter expression \cite[Theorem~1]{ahlswedeCRpart1} for secret key capacity: 
% \begin{align}
%     \wskc^{\text{ow}} = \max_{\RV-\RU-\RX-(\RY, \RZ)}\Big[I(\RU \wedge \RY | \RV)-I(\RU \wedge \RZ | \RV)\Big]. \label{eq:oneway_wskc}
% \end{align}
% In the above optimization, it is enough to consider random variables $\RU$ and $\RV$ (taking values in sets $\mc{U}$ and $\mc{V}$, respectively) such that $|\mcU| \leq |\mcX|^2$ and $|\mcV| \leq |\mcX|$.
On the other hand, the problem of one-way leakage rate was studied in \cite{vinod07}, but with a measure of leakage $\rl^{\text{ow}}+I(\RX \wedge \RZ)$. A single-letter expression obtained, in \cite[Theorem~1]{vinod07}, for $\rl^{\text{ow}}+I(\RX \wedge \RZ)$ is $\min\limits_{\RS-\RX-(\RY, \RZ)}\left[I(\RS,\RZ \wedge \RX)\right.\linebreak\left.+\,H(\RX | \RS, \RY)\right]$.
Therefore, we have
\begin{align}\label{eq:oneway_rl}
    \rl^{\text{ow}} = \min\limits_{\RS-\RX-(\RY, \RZ)}\left[I(\RS \wedge \RX | \RZ)+H(\RX | \RS, \RY)\right],
\end{align}
where the minimization is over random variables $\RS$ taking values in a set $\mc{S}$ such that $|\mcS|\leq |\mcX|$.

\subsection{Two messages are allowed, i.e., $r=2$ }
If we allow the users to exchange two messages interactively (Fig.~\ref{fig:twoway}), then omniscience is possible, as users 1 and 2 can communicate non-interactively at any rate larger than $H(\RX|\RY)+H(\RY|\RX)$ to recover each other's source. Let $\wskc^{(r)}$ and $\rl^{(r)}$ be defined as in \eqref{eq:wskc} and \eqref{eq:rl} but with a restriction to communication schemes involving only  $r=2$ interactive messages. Here we do not impose the condition that a particular user must transmit the first message. So any user can initiate the protocol, but we allow at most two messages to be exchanged. In this case, we can ask the same question: Does $\wskc^{(2)}>0$ imply that $\rl^{(2)}< H(\RX, \RY|\RZ)\/$? 

 \begin{figure}[h]
\centering
\resizebox{0.9\width}{!}{\tikzstyle{block}=[rectangle, draw, thick, minimum width=2em, minimum height=2em]

\begin{tikzpicture}[node distance=4cm,auto,>=latex']

    \node (f) {};
    \node [block] (a) [left of = f, node distance=1.5cm] {1};
    \node (x) [above of = a, node distance=1.25cm] {$\RX^n$};
    \node [block] (b) [right of=f, node distance=1.5cm] {2};
    \node (y) [above of = b, node distance=1.25cm] {$\RY^n$};
    \node [block] (w) [below of = f, node distance=1.5cm] {W};
    \node (z) [left of = w, node distance=1.25cm] {$\RZ^n$};
    \node (e1) [left of = a, node distance=1.5cm] {$\RE_1^{(n)}$};
    \node (e2) [right of = b, node distance=1.5cm] {$\RE_2^{(n)}$};

    \draw[->, thick] (x) --  (a);
    \draw[->, thick] (a) --  (e1);
    \draw[->, thick] (y) --  (b);
    \draw[->, thick] (b) --  (e2);
   \draw[->, thick] ([yshift=0.15 cm] a.east) -- ([yshift=0.15 cm]b.west) node [midway, above] {$\RF$};
   \draw[->, thick] ([yshift=-0.15 cm]b.west) -- ([yshift=-0.15 cm]a.east) ;
  %    \draw[->] (f.center) --  (w);
    \draw[->,thick] (z) -- (w);
\end{tikzpicture}}
\caption{Two messages are allowed. Here omniscience is feasible. If user $1$ initiates the communication, then $\RF=(\RF_1, \RF_2)$ where $\RF_2$, the communication by user $2$, depends on $\RF_1$ . Similarly, if user $2$ starts the communication, then $\RF=(\RF_2, \RF_1)$ and $\RF_1$, the communication made by user $1$, depends on $\RF_1$.}
\label{fig:twoway}
 \end{figure}
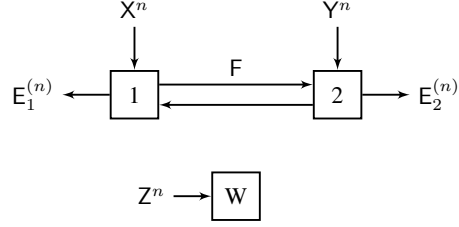
 
It turns out that with two messages, the ability to generate a positive secret key rate does not imply that the minimum leakage rate for omniscience is  strictly  less than $H(\RX, \RY | \RZ)$. To show this, we will use the results from the one-way communication setting. Let $\rl^{\text{ow}}(1 \rightarrow 2)$ (resp. $\rl^{\text{ow}}(2 \rightarrow 1))$  denote the minimum leakage rate for recovery of $\RX$ by user $2$ when user $1$ is the transmitter (resp. recovery of $\RY$ by user $1$ when user $2$ is the transmitter). By \eqref{eq:oneway_rl}, we have 
\begin{align*}
    \rl^{\text{ow}} (1 \rightarrow 2) &= \min\limits_{\RS-\RX-(\RY, \RZ)}\left[I(\RS \wedge \RX | \RZ)+H(\RX | \RS, \RY)\right]
    % \wskc^{\text{ow}}(1 \rightarrow 2) & = \max_{\RV-\RU-\RX-(\RY, \RZ)}\left[I(\RU \wedge \RY | \RV)-I(\RU \wedge \RZ | \RV)\right],
\end{align*}
and 
\begin{align*}
    \rl^{\text{ow}} (2 \rightarrow 1) &= \min\limits_{\RS-\RY-(\RX, \RZ)}\left[I(\RS \wedge \RY | \RZ)+H(\RY | \RS, \RX)\right].
    % \wskc^{\text{ow}}(2 \rightarrow 1) & = \max_{\RV-\RU-\RY-(\RX, \RZ)}\left[I(\RU \wedge \RX | \RV)-I(\RU \wedge \RZ | \RV)\right].
\end{align*}

% Since any one-way SKA scheme is also a valid SKA scheme in the $r=2$ case, 
% \begin{align}
%     \wskc^{(2)}\geq \max \left\lbrace \wskc^{\text{ow}}(1 \rightarrow 2), \wskc^{\text{ow}}(2 \rightarrow 1)\right \rbrace.\label{eq:twowskc}
% \end{align}
We next prove the following lower bound on the minimum leakage rate: 
\begin{align}
         \rl^{(2)}\geq \min & \left\lbrace \rl^{\text{ow}} (1 \rightarrow 2)+ H(\RY | \RZ, \RX),  
         %\right. \notag \\ &\mkern 30mu\left.
         \rl^{\text{ow}} (2 \rightarrow 1)+ H(\RX | \RZ, \RY)\right\rbrace,\label{eq:tworl}
\end{align}
% \begin{align*}
%          \rl^{(2)}\geq \min & \left\lbrace \min_{\RU-\RX-\RY,\RZ} I(\RU \wedge \RX | \RZ)+ H(\RX | \RY, \RU)+ H(\RY | \RZ, \RX), \right.  \notag \\ &\mkern -30mu\left. \min_{\RU-\RY-\RX,\RZ} I(\RU \wedge \RY | \RZ)+ H(\RY | \RX, \RU)+ H(\RX | \RZ, \RY)\right\rbrace,\label{eq:tworl}
% \end{align*}
where each term corresponds to a lower bound on the leakage rate when a particular user transmits first.  This bound may not be tight in general but will be enough for our purpose of constructing a counterexample. To prove \eqref{eq:tworl}, first we will show that $\rl^{(2)}\geq \rl^{\text{ow}} (1 \rightarrow 2)+ H(\RY | \RZ, \RX)$ when user $1$ starts the communication. Note that for any omniscience scheme $(\RF_1^{(n)},\RF_2^{(n)})$, we have $I(\RF_1^{(n)},\RF_2^{(n)} \wedge \RX^n,\RY^n|\RZ^n) \geq I(\RF_1^{(n)} \wedge \RX^n|\RZ^n) + I(\RF_2^{(n)} \wedge \RY^n|\RZ^n,\RX^n) \geq I(\RF_1^{(n)} \wedge \RX^n|\RZ^n) + H(\RY^n|\RZ^n,\RX^n)- n\delta_n$, where the last equality follows from Fano's inequality and the recoverability condition of $\RY^n$ from $\RF_2^{(n)}$ and $\RX^n$. Here, $\delta_n \to 0$ as $n \to \infty$. Therefore, we have 
\begin{align*}
\limsup_{n \to \infty} \frac{1}{n} I(\RF_1^{(n)},\RF_2^{(n)} \wedge \RX^n,\RY^n|\RZ^n) & \geq \limsup_{n \to \infty}\frac{1}{n} I(\RF_1^{(n)} \wedge \RX^n|\RZ^n) + H(\RY|\RZ,\RX)\\
&\geq \rl^{\text{ow}} (1 \rightarrow 2)+ H(\RY | \RZ, \RX).
\end{align*}
Since the above inequality holds for any omniscience scheme where user $1$ initiates the communication, we can conclude that $\rl^{(2)}\geq \rl^{\text{ow}} (1 \rightarrow 2)+ H(\RY | \RZ, \RX)$. Similarly, for  omniscience schemes with user $2$ starting the communication, we have that $\rl^{(2)}\geq \rl^{\text{ow}} (2 \rightarrow 1)+ H(\RX | \RZ, \RY)$. This completes the proof of \eqref{eq:tworl}.

We will make use of the following standard result on broadcast channels to construct a source $P_{\RX\RY\RZ}$ with $\wskc^{(2)}>0$ and $\rl^{(2)}=H(\RX, \RY | \RZ)$. Let $h(q)$ denote the binary entropy function, i.e, $h(q)= -q \log_2q-(1-q) \log_2(1-q)$, for $q \in (0,1)$.
 
% The motivation for considering broadcast channels comes from the structure of the expressions.

\begin{lemma}[{{\cite[p.~121]{elgamalbook}}}] \label{lem:bc_oneway}
    Consider a discrete memoryless broadcast channel $P_{\RY\RZ | \RX}$ with $\mc{X} \in \{0,1\}$, $\mc{Y} \in \{0,1\}$ and $\mc{Z} \in \{0,1,\Delta\}$, where the channel from $\RX$ to $\RY$ is BSC($p$), $p \in (0,\frac{1}{2})$, and the channel from $\RX$ to 
    $\RZ$ is BEC($\epsilon$),  $\epsilon \in (0,1)$.  Then, for $\epsilon \leq h(p)$,  $\RZ$ is more capable than $\RY$, i.e., for every input distribution $P_{\RX}$, $I(\RX\wedge\RZ) \geq I(\RX\wedge\RY).$
\end{lemma}

For a source distribution  $P_{\RX\RY\RZ}=P_{\RX}P_{\RY\RZ|\RX}=P_{\RY}P_{\RX\RZ|\RY} $, if $\RZ$ is more capable than $\RY$ for the channel $P_{\RY\RZ|\RX}$, then $\min\limits_{\RS-\RX-(\RY, \RZ)}\left[I(\RX \wedge \RZ | \RS) - I(\RX \wedge \RY | \RS)\right] =\sum_{s \in \mcS} P_{\RS}(s)\left[I(\RX \wedge \RZ | \RS=s) - I(\RX \wedge \RY | \RS=s)\right]\geq 0$. This is because for an $s \in \mcS$ with $P_{\RS}(s)>0$, the term $I(\RX \wedge \RZ | \RS=s) - I(\RX \wedge \RY | \RS=s)$ is evaluated with respect to $P_{\RX,\RY,\RZ | \RS=s} = P_{\RX | \RS=s}P_{\RY,\RZ | \RX}= P_{\RX | \RS=s}P_{\RY | \RX}P_{\RZ | \RX}$. So this term is equal to $I(\RX_s \wedge \RZ) - I(\RX_s \wedge \RY)$, where $\RX_s \sim P_{\RX | \RS=s}$, and $\RY$ (resp. $\RZ$) is obtained by passing $\RX_s$ through $P_{\RY|\RX}$ (resp. $P_{\RZ|\RX}$). Since $\RZ$ is more capable than $\RY$, $I(\RX_s \wedge \RZ) - I(\RX_s \wedge \RY)\geq 0$ for every $s$. As a result, we have $I(\RX \wedge \RZ | \RS) - I(\RX \wedge \RY | \RS) \geq 0$. Therefore, we have
\begin{align*}
    \rl^{\text{ow}} (1 \rightarrow 2)+ H(\RY | \RZ, \RX) &= \min\limits_{\RS-\RX-(\RY, \RZ)}\left[I(\RS \wedge \RX | \RZ)+H(\RX | \RS, \RY)\right]+ H(\RY | \RZ, \RX)\\
    &= H(\RX,\RY | \RZ) + \min\limits_{\RS-\RX-(\RY, \RZ)}\left[I(\RX \wedge \RZ | \RS) - I(\RX \wedge \RY | \RS)\right]\\
    &\geq H(\RX,\RY | \RZ).
\end{align*}
Similarly, for the channel $P_{\RX\RZ|\RY}$, if $\RZ$ is more capable than $\RX$, then we have $ \rl^{\text{ow}} (2 \rightarrow 1)\geq H(\RX,\RY | \RZ)$. Thus $\rl^{(2)} = H(\RX, \RY | \RZ)$, which follows from \eqref{eq:RL:lb} and \eqref{eq:tworl}. 

A source $(\RX, \RY, \RZ)$ is called a \emph{DSBE$(p,\epsilon)$ source} if $(\RX, \RY)$ is a doubly  symmetric binary source with parameter $p$, and $\RZ\in \{0,1\}^2 \cup \{\Delta\}$ is obtained by passing $(\RX, \RY)$ through an erasure channel with erasure probability $\epsilon$. It means that for a DSBE$(p,\epsilon)$ source $(\RX, \RY, \RZ)$, $\RX \sim \text{Ber}(\frac{1}{2})$, the channel from $\RX$ to $\RY$ is a BSC($p$), and the channel from $(\RX, \RY)$ to $\RZ$ is 
 $$P_{\RZ | \RX,\RY}\left(z|x,y\right)= \left\{\begin{array}{ll} 1-\epsilon, & \text{if } z=(x,y),\\
        \epsilon, & \text{if } z=\Delta,\\
        0, & \text{otherwise},
        \end{array}\right.$$ 
for every $(x,y)\in \{0,1\}^2$.
\begin{lemma}\label{lem:twowaycounter}
For a DSBE$(p,\epsilon)$ source with $p$ and $\epsilon$ chosen so that $ \frac{\min\{p,1-p\}}{\max\{p,1-p\}} < \epsilon \leq h(p)$, we have $\wskc^{(2)}>0$ but $\rl^{(2)}=H(\RX,\RY|\RZ)$.
\end{lemma}
\begin{proof}
 First note that $\wskc^{(2)}>0$ is equivalent to the condition that $\wskc>0$ (with no restriction on the number of communications), as was shown by Orlitsky and Wigderson in \cite{orlitsky1993secrecy} (and reproduced in Theorem~1 of \cite{amin2020}). For a DSBE$(p,\epsilon)$, it follows from Equation~(69) of \cite{amin2020} that $\wskc>0$ if and only if $\epsilon > \frac{\min\{p,1-p\}}{\max\{p,1-p\}}$. As a result, we have $\wskc^{(2)}>0$ for the chosen parameters.

Let us now argue that $\rl^{(2)}=H(\RX,\RY|\RZ)$. Since a DSBE$(p,\epsilon)$ source is symmetrical in $\RX$ and $\RY$, it is enough to show that the more capable condition hold for one user. In other words, it is sufficient to show that for the channel $P_{\RY\RZ | \RX}$, $\RZ$ is more capable than $\RY$. For any binary input distribution $P_{\tRX}=(P_{\tRX}(0), P_{\tRX}(1)):=(q, 1-q), 0 \leq q \leq 1$, to the channel $P_{\RY \RZ | \RX}$,  $I(\tRX\wedge\RY) =  h(p*q) - h(p) $, where $p*q=p(1-q)+(1-p)q$. Let $f(q):= (1-\epsilon)h(q) - h(p*q) +h(p)= I(\tRX\wedge\RZ) - I(\tRX\wedge\RY)$. Note that this difference is the same as that of the source considered in Lemma~\ref{lem:bc_oneway}. The proof of that lemma involves showing that for $\epsilon \leq h(p)$, $f(q)$ is a non-negative function, % and moreover, $f(q)$ is strictly convex around $q=\frac{1}{2}$, 
which is equivalent to the more capable condition. %and not less noisy conditions, respectively. 
Making use of this property of $f(q)$, we can also conclude that  for $\epsilon \leq h(p)$, $\RZ$ is more capable than $\RY$ for $P_{\RY\RZ | \RX}$. 

Thus, the minimum leakage rate $\rl^{(2)}=H(\RX,\RY|\RZ)$ because $\RZ$ is more capable than $\RY$ for the channel $P_{\RY\RZ | \RX}$, and $\RZ$ is more capable than $\RX$ for the channel $P_{\RX\RZ | \RY}$.
% Since $\RZ$ is not less noisy than $\RY$ for $P_{\RY\RZ | \RX}$, $\wskc^{\text{ow}}(1 \rightarrow 2)$ is positive, and hence we have $\wskc^{(2)}>0$ by \eqref{eq:twowskc}. 
% And, the minimum leakage rate $\rl^{(2)}=H(\RX,\RY|\RZ)$ because $\RZ$ is more capable than $\RY$ for the channel $P_{\RY\RZ | \RX}$, and $\RZ$ is more capable than $\RX$ for the channel $P_{\RX\RZ | \RY}$.
\end{proof}

For the source given in the above lemma, no user can gain an advantage in terms of $\rl^{(2)}$ over the other by starting the communication. This completes the proof of Proposition~\ref{prop:positivity}. 

 This result seems to indicate that duality does not always hold. We conjecture that for the DSBE  source considered in the above lemma, $\wskc^{(r)}>0$ need not imply $\rl^{(r)} < H(\RZ_V|\RZ_{\opw})$, $r\geq 2$. We additionally conjecture that, even with no restriction on the number of communications, $\wskc>0$ need not imply $\rl < H(\RZ_V|\RZ_{\opw})$.   

% \subsection{$r\geq 3$ messages are allowed}
% The interplay between  the positivity of $\wskc$ and non-maximality of $\rl$ becomes more evident if there is no restriction on the number of messages. Unlike as in the case of one or two messages, our intuition actually works with multiple messages, in particular, with three messages.

\section{Duality for finite linear source models}\label{sec:duality_fls}
In this section, we  consider  a  broad  class  of  sources,  namely, finite linear sources, for which we believe the duality between secure omniscience and wiretap secret key agreement must hold. 

\begin{definition}[Finite linear source \cite{chan11itw}]
A source $(\RZ_V, \RZ_{\opw})$ is said to be a \textit{finite linear source (FLS)} if we can express $\RZ_V$ and $\RZ_{\opw}$ as
$$\bM \RZ_V & \RZ_{\opw}\eM=\bM \RZ_1 & \cdots& \RZ_m & \RZ_{\opw}\eM=\RX \bM \MM_1\;\cdots\;\MM_m \; \MW \eM,$$
where  $\RX$ is a random row vector of some length $l$ that is uniformly distributed over a field $\Fq^l$, and $\MM_1, \ldots,\MM_m, \MW$ are some matrices over $\Fq$ with dimensions $l \times l_1, \ldots ,l \times l_m, l \times l_w$, respectively. Each terminal observes a collection of linear combinations of the entries in $\RX$. 
\end{definition} 

In the context of FLS models, we say a communication scheme $\RF^{(n)}$ is \emph{linear} if each user's communication is a linear function of its observations and the previous communication on the channel. Without loss of generality \cite[Sec.~II]{chan19oneshot}, linear communication can be assumed to be non-interactive.  In the rest of the paper, we consider only matrices over $\Fq$ unless otherwise specified.

The following notions related to G\'{a}cs-K\"{o}rner common information will play an important role in proving some of our subsequent results. The \emph{ G\'{a}cs-K\"{o}rner common information} of  $\RX$ and $\RY$ with joint distribution $P_{\RX,\RY}$ is defined as
\begin{align}\label{eq:gk}
 J_{\op{GK}}(\RX \wedge \RY) := \max \left\lbrace H(\RG) : H(\RG|\RX)=H(\RG|\RY) =0 \right\rbrace
\end{align}
A $\RG$  that satisfies the constraint in \eqref{eq:gk} is called a common function (c.f.) of $\RX$ and $\RY$. An optimal $\RG$ in \eqref{eq:gk} is called a \emph{maximal common function} (m.c.f.) of $\RX$ and $\RY$, and is denoted by $\op{mcf}(\RX, \RY)$. Similarly, for $m$ random variables, $\RX_1, \RX_2, \ldots, \RX_m$,  we can extend these definitions by replacing the condition in \eqref{eq:gk} with $H(\RG|\RX_1)=H(\RG|\RX_2)=\ldots=H(\RG|\RX_n)=0$. For a two-user FLS $(\RZ_1, \RZ_2)$, i.e., $\RZ_1 = \RX \MM_1$ and $\RZ_2=\RX \MM_2$ for some matrices $\MM_1$ and $\MM_2$ where $\RX$ is a $1 \times l$ row vector   uniformly distributed on $\Fq^l$, it was shown in \cite{chan18zero} that the $\op{mcf}(\RZ_1, \RZ_2)$ is a linear function of each of $\RZ_1$ and $\RZ_2$. This means that there exists some matrices $\MM_{z_1}$ and $\MM_{z_2}$ such that $\op{mcf}(\RZ_1, \RZ_2) = \RZ_1 \MM_{z_1}=\RZ_2 \MM_{z_2}$. One can infer from this relation that if $\RZ_1$ and $\RZ_2$ are independent, then $\op{mcf}(\RZ_1, \RZ_2)$ is identically $0$.

We prove results in this and the next section favoring the following conjecture.
%we address the question: For what sources, secure omniscience achieves wiretap secret key capacity? In other words, does $$\rl = H(\RZ_V|\RZ_{\opw}) - \wskc $$ hold for a large enough class of sources?\\
\begin{Conjecture}\label{conj:duality:fls}
$\rl = H(\RZ_V|\RZ_{\opw}) - \wskc$ holds for finite linear sources.
\end{Conjecture}

The reason to believe Conjecture~\ref{conj:duality:fls} comes from the following two theorems. Since the source is linear, it is reasonable to conjecture that linear schemes are optimal. Theorem~\ref{thm:linearscheme} below states that if a linear \emph{perfect SKA scheme} is optimal in terms of $\wskc$, then secure omniscience achieves wiretap secret key capacity. Here, we call an SKA scheme \emph{perfect} if there exists a sequence of communication-key pairs $(\RF^{(n)}, \RK^{(n)})_{n\geq1}$ such that 
$H(\RK^{(n)}|\RF^{(n)},\RZ_i^n) = 0$ for all users $i \in V$ (perfect key recoverability condition), and $\log |\mc{K}^{(n)}| = H(\RK^{(n)}| \RF^{(n)},\RZ_{\opw}^n)$ (perfect secrecy condition).

\begin{theorem}\label{thm:linearscheme}
  For an FLS $(\RZ_V,\RZ_{\opw})$, if a linear perfect SKA scheme achieves $\wskc$, then we have
  $$\rl = H(\RZ_V|\RZ_{\opw}) - \wskc.$$
\end{theorem}
\begin{proof}
See Appendix~\ref{app:thm:proof_linearscheme}.
\end{proof}

The next theorem shows the duality between secure omniscience and wiretap secret key agreement for two-user FLS without any restriction to linear schemes. It also provides single-letter expressions for $\rl$ and $\wskc$.

\begin{theorem}[Two-user FLS]
  \label{thm:fls}
  For secure omniscience with $V=\Set{1,2}$ and FLS $(\RZ_V,\RZ_{\opw})$, we have
  \begin{align}
    R_{\opL} &= H(\RZ_1,\RZ_2|\RZ_{\opw}) - \wskc,\\
    \wskc& = I(\RZ_1\wedge \RZ_2|\RG),\label{eq:fls}
  \end{align}
  where $\RG$ can be chosen to be $\RG_1$, $\RG_2$, or $(\RG_1,\RG_2)$, with $\RG_i$ being the solution to 
  \begin{align}
    J_{\op{GK}}(\RZ_{\opw}\wedge \RZ_i) := \max_{\RG_i: H(\RG_i|\RZ_{\opw})=H(\RG_i|\RZ_i)=0} H(\RG_i) \label{eq:JGK}
  \end{align}
  for $i\in V$.
\end{theorem}
\begin{proof}
See Appendix~\ref{app:thm:fls}.
\end{proof}

The following example compares the result of Theorem~\ref{thm:fls} with that of Csisz\'ar and Narayan~\cite{csiszar04} for the case of  two-user FLSs. Recall from Section~\ref{sec:wska:def} that for general sources, $\skc(\RZ_V) = H(\RZ_V) - \rco(\RZ_V)$ and $\pkc(\RZ_V | \RZ_{\opw}) = H(\RZ_V|\RZ_{\opw}) - \rco(\RZ_V|\RZ_{\opw})$; it also holds that $\wskc(\RZ_V \| \RZ_{\opw}) \leq \min\{\skc(\RZ_V),\pkc(\RZ_V | \RZ_{\opw})\}$, and $\rco(\RZ_V|\RZ_{\opw})\leq \rl(\RZ_V \| \RZ_{\opw}) \leq \rco(\RZ_V)$. For the source considered in the next example, we see that all these inequalities are strict.
\begin{example}
Let $V=\{1, 2\}$, and let $\RX= (\RX_{a}, \RX_{b}, \RX_{c}, \RX_{d})$ be a random vector uniformly distributed over $\mathbb{F}_2^4$. Consider the two-user FLS  $(\RZ_V, \RZ_{\opw})$ with
\begin{gather*}
   \RZ_1= (\RX_a, \RX_b, \RX_c)  \quad \RZ_2= (\RX_b, \RX_c, \RX_d) \quad \RZ_{\opw}= (\RX_b+\RX_c, \RX_a+\RX_d).
\end{gather*}
For this source, $\RG_1=\RG_2 = \RG=\RX_b+\RX_c$.
It follows from Theorem~\ref{thm:fls} that $\wskc(\RZ_V \| \RZ_{\opw})= I(\RZ_1 \wedge \RZ_2|\RG)= H(\RZ_1|\RG)-H(\RZ_1|\RG, \RZ_2)= H(\RX_a, \RX_b, \RX_c|\RX_b+\RX_c)-H(\RX_a, \RX_b, \RX_c|\RX_b+\RX_c, \RX_b, \RX_c, \RX_d)= 2-1= 1 \text{ bit}$. On the other hand, the secret key capacity of this source is $\skc(\RZ_V)=I(\RZ_1 \wedge \RZ_2)=H(\RZ_1)-H(\RZ_1|\RZ_2)=3-1=2 \text{ bits}$; and the private key capacity is $\pkc(\RZ_V|\RZ_{\opw})=I(\RZ_1 \wedge \RZ_2|\RZ_{\opw})=H(\RZ_1|\RZ_{\opw})-H(\RZ_1|\RZ_2,\RZ_{\opw})=2-0=2 \text{ bits}$. Observe that $$1=\wskc(\RZ_V \| \RZ_{\opw})< \min \{\skc(\RZ_V), \pkc(\RZ_V|\RZ_{\opw})\}=2.$$

Similarly, using the duality in Theorem~\ref{thm:fls} and the results of \cite{csiszar04}, we get $\rl (\RZ_V\|\RZ_{\opw})= H(\RZ_1,\RZ_2|\RZ_{\opw}) - \wskc(\RZ_V\|\RZ_{\opw})=2-1=1\text{ bit}$, $\rco(\RZ_V) = H(\RZ_1,\RZ_2) - \skc(\RZ_V)=4-2=2\text{ bits}$, and $\rco(\RZ_V|\RZ_{\opw}) = H(\RZ_1,\RZ_2|\RZ_{\opw}) - \pkc(\RZ_V|\RZ_{\opw}) = 2-2 =0 $. Thus, we have $\rco(\RZ_V|\RZ_{\opw})< \rl(\RZ_V\|\RZ_{\opw}) < \rco(\RZ_V)$.

\end{example}

In the next section, we prove the duality between secure omniscience and wiretap secret key agreement for tree-PIN sources with linear wiretapper, a sub-class of FLSs. We also give single-letter expressions for $\rl$ and $\wskc$ for this model.

\section{Tree-PIN source with linear wiretapper}\label{sec:treepin}
 A source $\RZ_V$ is said to be \emph{tree-PIN} if there exists a tree $T=(V,E,\xi)$ and for each edge $e \in E$, there is a non-negative integer $n_e$ and a random vector $\RY_e = \left( \RX_{e,1}, \ldots, \RX_{e,n_e} \right)$. We assume that the collection of random variables $\RX :=(\RX_{e,k}: e\in E, k \in [n_e])$ are i.i.d. and each component is  uniformly distributed over a finite field, say $\Fq$. For $i \in V$,
 \begin{align*}
  \RZ_i = \left( \RY_e : i \in \xi (e) \right).
 \end{align*} 

 The linear wiretapper's side information $\RZ_{\opw}$ is defined as 
\begin{align*}
 \RZ_{\opw} = \RX \MW,
\end{align*}
where $\RX$ is a $1 \times (\sum_{e \in E}n_e)$ vector and $\MW$ is a $(\sum_{e \in E}n_e) \times n_w$ full column-rank matrix over $\Fq$. We sometimes refer to $\RX$ as the base vector. We refer to the pair $(\RZ_V, \RZ_{\opw})$ defined as above as a \emph{tree-PIN source with linear wiretapper}. This is a special case of an FLS. 

\begin{example}
Consider the tree $T$ in Fig.~\ref{fig:tree_pin_def} defined on  $V=\{1, \ldots,5\}$ with $E=\{a,b,c,d\}$. Let $\RY_a = (\RX_{a1}, \RX_{a2})$, $\RY_b = \RX_{b1}$, $\RY_c =\RX_{c1}$, and $\RY_d =\RX_{d1}$, where the base vector $\RX= (\RX_{a1}, \RX_{a2}, \RX_{b1}, \RX_{c1}, \RX_{d1})$ is uniformly distributed over $\mathbb{F}_2^5$. The corresponding source $\RZ_V$ is given by
\begin{gather*}
   \RZ_1= (\RX_{a1}, \RX_{a2}), \quad \RZ_2= (\RX_{a1}, \RX_{a2}, \RX_{b1}, \RX_{c1}), \quad \RZ_3=\RX_{c1}\\
    \RZ_4= (\RX_{b1}, \RX_{d1}), \quad \RZ_5= \RX_{d1}.
\end{gather*}

\begin{figure}[h]
    \centering
        \resizebox{0.85\width}{!}{\begin{tikzpicture}[-,>=stealth, line width=5pt,thick, auto]
\tikzstyle{vertex}=[circle, fill=black,inner sep=0pt, minimum size=5pt];

\node[vertex]      (2)        [label= above:{$2$}] {};
\node[vertex]      (1)        [above left = 3 em and 4 em of 2,label= above:{$1$}]  {};
\node[vertex]      (3)       [below left = 3 em and 4 em of 2,label= below:{$3$}] {};
\node[vertex]      (4)      [right = 4 em of 2,label= above:{$4$}]  {};
\node[vertex]      (5)      [right = 4 em of 4,label= above:{$5$}]  {};

\draw (1) -- (2) node (a) [midway, above] {$a$} ;
\draw (2) -- (4) node (b) [midway, above] {$b$} ;
\draw (2) -- (3) node (c) [midway, below] {$c$} ;
\draw (4) -- (5) node (d) [midway, above] {$d$} ;

\end{tikzpicture}}
        \caption{Tree $T$ corresponding to a tree-PIN model}\label{fig:tree_pin_def}
\end{figure}
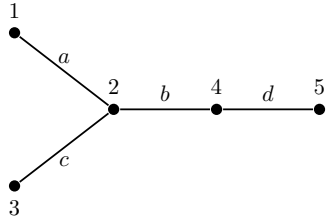
The observations of a wiretapper are $\RZ_{\opw}=(\RX_{a1}+\RX_{a2}, \RX_{a2}+\RX_{b1}+\RX_{d1})$.
This source $(\RZ_V, \RZ_{\opw})$ is an example of a tree-PIN source with a linear wiretapper. This is a special case of an FLS, as we can express it as $\bM \RZ_1 & \cdots& \RZ_m & \RZ_{\opw}\eM=\RX \bM \MM_1\;\cdots\;\MM_m \; \MW \eM$ for some matrices $\MM_1,\dots, \MM_m, \MW$. For example, we can write  $$\RZ_2= (\RX_{a1}, \RX_{a2}, \RX_{b1}, \RX_{c1})=\RX \underbrace{\bM 1&0&0&0\\0 & 1&0&0\\ 0 & 0&1&0\\0&0&0&1\\ 0 & 0&0&0 \eM}_{\MM_2},$$ 
$$\RZ_{\opw}=(\RX_{a1}+\RX_{a2}, \RX_{a2}+\RX_{b1}+\RX_{d1})=\RX \underbrace{\bM 1&0\\1 & 1\\ 0 & 1\\0&0\\ 0 & 1 \eM}_{\MW}.$$
\end{example}

\subsection{Motivating example} The following example of a tree-PIN source with linear wiretapper appeared in our earlier work \cite{chan20secure}, where we constructed an optimal secure omniscience scheme. Let $V=\{1,2,3,4\}$ and  
     \begin{align}
        \RZ_{\opw} &= \RX_a+\RX_b+\RX_c, \\
        \RZ_1 &= \RX_a, \; \RZ_2 = (\RX_a, \RX_b),\; \RZ_3 = (\RX_b, \RX_c),\; \RZ_4 =  \RX_c,
      \end{align}
  where $\RX_a$, $\RX_b$ and $\RX_c$ are uniformly random and independent bits. The tree here is a path of  length $3$ (Fig.~\ref{fig:exampletree}) and the wiretapper observes a linear combination of all the edge random variables. For secure omniscience, terminals 2 and 3, using $n=2$ i.i.d. realizations of the source, communicate linear combinations of their observations. The communication is of the form $\RF^{(2)} =(\tRF_2^{(2)},\tRF_3^{(2)})$, where  $\tRF_2^{(2)} =\RX^2_a+\MM \RX^2_b$ and $\tRF_3^{(2)}=(\MM + \MI) \RX_b^2 +\RX_c^2$ with $\MM:=\bM 1 & 1\\ 1 & 0\eM$.  Since the matrices $\MM$ and $\MM+\MI$ are invertible, all the terminals can recover $\RZ_V^2$ using this communication. For example, user 1 can first  recover $\RX_b^2$ from $(\RX_a^2, \tRF_2^{(2)})$ as $\RX_b^2 = (\MM+\MI)(\RX_a^2+ \tRF_2^{(2)})$, then $\RX_b^2$ can be used along with $\tRF_3^{(2)}$ to recover $\RX_c^2$ as $\RX_c^2 = (\MM+\MI)\RX_b^2+ \tRF_3^{(2)}$.  More interestingly, this communication is ``aligned" with the eavesdropper's observations, since $\RZ^2_{\opw} = \tRF_2^{(2)}+\tRF_3^{(2)}$. This scheme achieves $R_L$, which is 1 bit. 
  
  For minimizing leakage, this kind of alignment must happen. For example, if $\RZ^2_{\opw}$ were not contained in the span of $\tRF_2^{(2)}$ and $\tRF_3^{(2)}$, then the wiretapper could infer a lot more from the communication.  Ideally, if one wants zero leakage, then $\RF^{(n)}$ must be within the span of $\RZ^n_{\opw}$, which is not feasible in many cases because, with that condition, the communication might not achieve omniscience in the first place. Therefore keeping this in mind, it is reasonable to assume that there can be components of $\RF^{(n)}$ outside the span of $\RZ^n_{\opw}$. But we look for communication schemes that span as much of $\RZ_{\opw}$ as possible. Such an alignment condition is used to control the leakage. In this particular example, it turned out that an omniscience communication that achieves $\rco$ can be made to align with the wiretapper side information completely, i.e., $H(\RZ^n_{\opw}|\RF^{(n)})=0$.  Motivated by this example, we show that it is always possible for some omniscience communication to achieve complete alignment with the wiretapper's observations within the class of tree-PIN sources with linear wiretapper.

\begin{theorem}\label{thm:cwsk:red} 
For a tree-PIN source $\RZ_V$ with linear wiretapper observing $\RZ_{\opw}$,
\begin{align*}
\wskc &= \min_{e \in E} H(\RY_e|\op{mcf}(\RY_e,\RZ_{\opw})),  \\ 
\rl &=\left(\sum_{e \in E}n_e -n_w\right)\log_2q -\wskc \text{ bits}.
\end{align*}
In fact, a linear non-interactive scheme  is  sufficient to achieve both $\wskc$ and $\rl$ simultaneously.
\end{theorem}

% The above theorem shows that the intrinsic upper bound on $\wskc$ holds with equality. In the multiterminal setting, the intrinsic bound that follows from \cite[Theorem 4]{csiszar04} is given by 
% \begin{align*}
%  \wskc(\RZ_V \|\RZ_{\opw}) \leq \min_{\RJ-\RZ_{\opw}-\RZ_V}\pkc(\RZ_V|\RJ).
% \end{align*}
%  This is analogous to the intrinsic bound for the two-terminal case \cite{maurer99intrinsic}. 
% For the class of tree-PIN sources with linear wiretapper,  when $\RJ^*= \left( \op{mcf}(\RY_e,\RZ_{\opw}) \right)_{e \in E}$, it can be shown that  $\pkc(\RZ_V|\RJ^*)= \min_{e \in E} H(\RY_e|\op{mcf}(\RY_e,\RZ_{\opw})) $. This can be derived using the characterization  in \cite{csiszar04} of the conditional minimum rate of communication for omniscience, $\rco(\RZ_V|\RJ^*)$. In fact, the same derivation can also be found in \cite{alireza19} for a  $\RJ$ that is obtained by passing the edge random variables through independent channels. In particular, $\RJ^{*}$ is a function of edge random variables $(\RY_e)_{e \in E}$ because $\op{mcf}(\RY_e, \RZ_{\opw})$ is a function of $\RY_e$. Therefore, we can see that $\pkc(\RZ_V|\RJ^*)$, which is an upper bound on $ \min_{\RJ-\RZ_{\opw}-\RZ_V}\pkc(\RZ_V|\RJ)$, matches with the $\wskc$ obtained from Theorem~\ref{thm:cwsk:red}.

 The theorem guarantees that we can achieve the wiretap secret key capacity in the tree-PIN case with linear wiretapper through a linear secure omniscience scheme, which establishes the duality between the two problems. This illustrates that omniscience can be helpful even beyond the case when there is no wiretapper side information.

 Our proof of Theorem~\ref{thm:cwsk:red} is through a reduction to the particular subclass of \emph{irreducible} sources, which we defined next.
 
\begin{definition} \label{def:irreducible}
 A tree-PIN source with linear wiretapper is said to be \emph{irreducible} if $\op{mcf}(\RY_e, \RZ_{\opw}) $ is a constant function for every edge $e \in E$ .
\end{definition}

Whenever there is an edge $e$ such that $\RG_e:=\op{mcf}(\RY_e, \RZ_{\opw}) $ is a non-constant function, the user corresponding to a vertex incident on $e$ can reveal $\RG_e$ to the other users. This communication does not leak any additional information to the wiretapper because $\RG_e$ is a function of $\RZ_{\opw}$. Intuitively, for further communication, $\RG_e$ is not useful and hence can be removed from the source. After the reduction, the m.c.f. corresponding to $e$ becomes a constant function. In fact, we can carry out the reduction until the source becomes irreducible. This idea of reduction is illustrated in the following example.

\begin{Example}
Let us consider a source $\RZ_V$ defined on a path of length 3, which is shown in Fig.~\ref{fig:exampletree}. Let $\RY_a = (\RX_{a1}, \RX_{a2})$, $\RY_b = \RX_{b1}$ and $\RY_c =\RX_{c1}$, where $\RX_{a1}$, $\RX_{a2}$, $\RX_{b1}$ and $\RX_{c1}$ are uniformly random and independent bits. 
\begin{figure}[h]
\centering
\resizebox{\width}{1cm}{\begin{tikzpicture}[-,>=stealth,thick, auto]
\tikzstyle{vertex}=[circle,fill=black,inner sep=0pt,minimum size=5pt];

\node[vertex]      (1)        [label= below:{$1$}] {};
\node[vertex]      (2)  [right  = 4 em of 1,label= below:{$2$}] {};
\node[vertex]      (3)  [right  = 4 em of 2,label= below:{$3$}] {};
\node[vertex]      (4)  [right  = 4 em of 3,label= below:{$4$}] {};

\draw (1) -- (2) node [midway, above] {$a$};
\draw (2) -- (3) node [midway, above] {$b$};
\draw (3) -- (4) node [midway, above] {$c$};
\end{tikzpicture}}
\caption{A path of length 3}
\label{fig:exampletree}
 \end{figure}
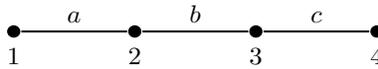
 If $\RZ_{\opw}=\RX_{b1}+\RX_{c1}$, then the source is irreducible because $\op{mcf}(\RY_e, \RZ_{\opw})$ is a constant function for all $e \in \{a,b,c\}$. 
 
 However if  $\RZ_{\opw}=(\RX_{a1}+\RX_{a2},  \RX_{b1}+\RX_{c1})$, then the source is not irreducible, as $\op{mcf}(\RY_a, \RZ_{\opw}) =\RX_{a1}+\RX_{a2}$, which is a non-constant function. An equivalent representation of the source is
 $\RY_a = (\RX_{a1}, \RG_{a})$, $\RY_b = \RX_{b1}$, $\RY_c =\RX_{c1}$ and $\RZ_{\opw}=(\RG_{a}, \RX_{b1}+\RX_{c1})$, where $\RG_{a}=\RX_{a1}+\RX_{a2}$, which is also a uniform bit independent of $(\RX_{a1}, \RX_{b1}, \RX_{c1})$. So, for omniscience, user 2 initially can reveal $\RG_{a}$ without affecting the information leakage as it is completely aligned to $\RZ_{\opw}$. Since everyone has $\RG_a$, users can just communicate according to the omniscience scheme corresponding to the source without $\RG_a$. Note that this new source is irreducible.
\end{Example}

The next lemma shows that the kind of reduction to an irreducible source used in the above example is indeed optimal in terms of $R_L$ and $\wskc$ for all tree-PIN sources with linear wiretapper.
\begin{lemma}\label{lem:irred} 
 If a tree-PIN source with linear wiretapper $(\RZ_V,\RZ_{\opw})$ is not irreducible then there exists an irreducible source $(\tRZ_V, \tRZ_{\opw})$ such that 
 \begin{align*}
\wskc(\RZ_V\| \RZ_{\opw}) = \wskc(\tRZ_V\|\tRZ_{\opw}),\\ \rl(\RZ_V\|\RZ_{\opw}) = \rl(\tRZ_V\|\tRZ_{\opw}),\\
H(\RY_e|\op{mcf}(\RY_e,\RZ_{\opw})) = H(\tRY_e)
\end{align*}
for all $e \in E$.
\end{lemma}

\begin{proof}
See Appendix~\ref{lem:proof:irred}.
\end{proof}
Note that, in the above lemma, the scheme that achieves $\rl(\RZ_V\|\RZ_{\opw})$  involves revealing the reduced m.c.f. components first and then communicating according to the scheme that achieves $\rl(\tRZ_V\|\tRZ_{\opw})$. As a consequence of Lemma~\ref{lem:irred}, to prove Theorem~\ref{thm:cwsk:red}, it suffices to consider only irreducible sources. For ease of reference, we re-state the theorem for irreducible sources below.

\begin{theorem}\label{thm:cwsk:irred} 
If a tree-PIN source $\RZ_V$ with linear wiretapper $\RZ_{\opw}$ is irreducible then 
\begin{align*}
\wskc &= \min_{e \in E} H(\RY_e)=\skc, \\
\rl &=\left(\sum_{e \in E}n_e -n_w\right)\log_2q- \wskc\text{ bits},
\end{align*}
where $\skc$ is the secret key capacity of Tree-PIN source without the wiretapper side information~\cite{csiszar04}.
\end{theorem}

\begin{IEEEproof}[Proof outline]
The main component of the proof is the achievability part, which involves a construction of an omniscience scheme with leakage rate $ \left(\sum_{e \in E}n_e -n_w\right)\log_2q - \min_{e \in E} H(\RY_e)$. Since $\wskc$ is upper bounded by $\skc = \min_{e \in E} H(\RY_e)$, we have $ \rl \leq \left(\sum_{e \in E}n_e -n_w\right)\log_2q - \min_{e \in E} H(\RY_e) \leq \left(\sum_{e \in E}n_e -n_w\right)\log_2q- \wskc$.  This upper bound, together with the $\rl$ lower bound in Theorem~\ref{thm:RL:lb}, proves the theorem.

In the construction of this scheme, we start by assuming a certain linear structure for the communication $\RF^{(n)}$, with the coefficients of the linear combinations involved in the communication taken to be variables (termed ``communication coefficients'') whose values need to be determined. We then argue that the desired leakage rate can be achieved if we additionally impose two conditions on $\RF^{(n)}$, namely, perfect omniscience and perfect alignment. These conditions translate to certain linear-algebraic constraints on the communication coefficients in $\RF^{(n)}$. So it is enough to find an $\RF^{(n)}$ with its communication coefficients satisfying these constraints. To that end, we use the probabilistic method to argue that if the blocklength $n$ of the source is large enough, then there indeed exists such a realization of $\RF^{(n)}$. 

For clarity of exposition, we first execute this proof in the case of tree-PIN sources in which the underlying tree is a simple path, before generalizing it to the case of trees. The details are given in Appendix~\ref{thm:proof:cwsk:irred}.
\end{IEEEproof}

% \begin{proof}
% See Appendix~\ref{thm:proof:cwsk:irred}.
% \end{proof}

Theorem~\ref{thm:cwsk:irred} shows that, for irreducible sources, even when the wiretapper has side information, the users can still extract a key at rate $\skc$. In terms of secret key generation, the users are not really at a disadvantage if the wiretapper has linear observations.

\subsection{Constrained wiretap secret key capacity of tree-PIN source with linear wiretapper}
Secure omniscience in fact plays a role even in achieving the constrained wiretap secret key capacity of tree-PIN source with linear wiretapper. The constrained wiretap secret key capacity, denoted by $\wskc(R)$, is defined as in \eqref{eq:wskc} but with the supremum over all SKA schemes with $\limsup\frac{1}{n}\log|\mc{F}^{(n)}|<R$ where $\mc{F}^{(n)}$ is the alphabet of $\RF^{(n)}$. The following theorem gives a single-letter expression for the constrained wiretap secret key capacity, whose form is reminiscent of the constrained secret key capacity, \cite[Theorem~4.2]{chan19}. The proof involves a construction of a secure omniscience communication scheme for a part of the source.

\begin{theorem}\label{thm:rateconstrained_treepin}
Given a tree-PIN source $\RZ_V$ with a linear wiretapper $\RZ_{\opw}$, we have
\begin{align*}
\wskc(R) = \min \left\{\frac{R}{|E|-1}, \wskc\right\},
\end{align*}
where $R$ is the total discussion rate and $\wskc =\min_{e \in E}H(Y_e|\op{mcf}(Y_e,Z_{\opw}))$, which is the unconstrained wiretap secret key capacity.
\end{theorem}

\begin{IEEEproof}[Proof outline]
As in Theorem~\ref{thm:cwsk:red}, this proof is also divided into two parts. In the first part, which is 
Theorem~\ref{lem:irred:rate}, we argue that removing an edge random variable that is also available with the wiretapper does not affect the constrained wiretap secret key capacity $\wskc(R)$. The argument involves showing that any valid SKA scheme on the original model can be converted into a valid SKA scheme on the reduced model, and vice versa.

In the second part, which is Theorem~\ref{thm:rate:irreducible}, we prove the statement of Theorem~\ref{thm:rateconstrained_treepin} for irreducible sources. The converse part follows from the constrained secret key capacity result in the case of tree-PIN sources without wiretapper side information, \cite[Theorem~4.2]{chan19}. The achievability part pivots on the argument that the rate pair $((|E|-1)\wskc,\wskc)$, where $\wskc=\min_{e \in E}H(Y_e)$, is achievable, as the rest of the curve follows from a time sharing argument.
% Note that, in the proof of Theorem~\ref{thm:cwsk:irred}, we showed that the key rate $\min_{e \in E}H(Y_e)$ is achievable using a  communication of rate $\left(\sum_{e \in E}n_e\right) \log_2q  - \min_{e \in E} H(\RY_e) \geq (|E|-1) \min_{e \in E} H(\RY_e)$. 
To show the achievability of this rate pair, we ignore some of the edge components of the source $(\RZ_V, \RZ_{\opw})$ to obtain a new source $(\RZ'_V, \RZ'_{\opw})$, which will be used for the purpose of secret key generation. In fact, we retain only $s=\min_{e \in E}n_e$ components for each edge. We then argue that for the new source $(\RZ'_V, \RZ'_{\opw})$, we can generate a key of rate $\wskc$ with a communication of rate $(|E|-1)\wskc$, if we use the key generation scheme used in the proof of Theorem~\ref{thm:cwsk:irred}.

Theorem~\ref{thm:rateconstrained_treepin} follows from combining both these parts. The complete details are given in Appendix~\ref{app:thm:proof_rateconstrained_treepin}.
\end{IEEEproof}

%  \begin{proof}
% See Appendix~\ref{app:thm:proof_rateconstrained_treepin}.
% \end{proof}

% \section{Positivity of $\wskc$} \label{sec:positivity}
% \input{positivity}
{\section{Secure function computation\label{sec:funcomp}}}
{We can also consider the problem of secure function computation in the context of the multiterminal source model. Given a function of the users' observations, the goal of the users is to interactively communicate over a public channel to compute this function without revealing much information about the computed value to the wiretapper. 

Formally, let $g(z_1,\ldots,z_m)$ be a function over $\mc{Z}_1\times \cdots \times \mc{Z}_m$. For $i \in V$, user $i$ observes $\RZ_i^n:=(\RZ_{i1},\ldots,\RZ_{in})$, $n$ i.i.d. realizations of the component $\RZ_i$ of the source, and the wiretapper observes $\RZ_{\opw}^n$. Using their observations, the users interactively communicate $\RF^{(n)}$ over a public channel. The wiretapper only listens to this public communication.  Finally, user $i$ computes $\RG_i^{(n)}$ for $i \in V$ using the public communication $\RF^{(n)}$, $\RZ_i^n$, and possible local randomness. We say that the users can compute $\RG= g(\RZ_1,\ldots,\RZ_m)$ if there exists an interactive communication $\RF^{(n)}$ satisfying
    \begin{align}
\liminf_{n \to \infty} \Pr(\RG_1^{(n)} = \ldots =\RG_m^{(n)} = \RG^n) = 1 \label{eq:func:recoverability},
\end{align}
where $\RG^n:=(g(\RZ_{11},\ldots,\RZ_{m1}),\ldots, g(\RZ_{1n},\ldots,\RZ_{mn}))$.% which is $n$ copies of the function $g$ evaluated for each independent realization of the source. 

In this set-up, the minimum leakage rate $\rl^{\RG}$ achievable for computation of $\RG$ is defined as

\begin{align}
\begin{split}
 \rl^{\RG}&:= \inf  \biggl\lbrace \limsup_{n \to \infty} \frac{1}{n}I(\RF^{(n)} \wedge \RG^n|\RZ_{\opw}^n) \biggr\rbrace, \label{eq:rl:func}
 \end{split}
\end{align}
where the infimum is over all communications $\RF^{(n)}$ that allow users to compute $\RG$. When $\RG=\RZ_V$, the secure function computation problem becomes the secure omniscience problem (Sec.~\ref{subsec:omniscience}).

We say that a function $\RG$ is \emph{securely computable} if $\rl^{\RG}=0$. This means that $\RG$ can be computed without revealing any information to the wiretapper other than what is already known through the side information. The communication that achieves $\rl^{\RG}=0$ allows users to compute $\RG$ securely. Henceforth, with an abuse of notation, we write $\rl$ instead $\rl^{\RG}$ when the computed function is clear from the context. The work of \cite{tyagi11} considered the case where wiretapper has no side information with an aim of characterizing the securely computable functions, i.e., $\rl=0$. Remarkably, \cite{tyagi11} gave this characterization in terms of the wiretap secret key capacity: If $\RG$ is securely computable then $H(\RG)\leq \skc$ and conversely, if $\RG$ satisfies $H(\RG)< \skc$, then $\RG$ is securely computable.

\begin{example}
Let $\RX_i$,  $i \in \{1,2,3,4\}$, be a collection of i.i.d. Ber($\frac{1}{2}$) random variables. A two user source $(m=2)$ is defined as follows.  
\begin{align*}
    \RZ_1&:=(\RX_1,\RX_2,\RX_3)\\
    \RZ_2&:=(\RX_2,\RX_3,\RX_4)
\end{align*}
% User $1$'s observations are $\RZ_1:=(\RX_1,\RX_2,\RX_3)$ and User $2$'s observations are $\RZ_2:=(\RX_3,\RX_4,\RX_5)$. The wiretapper's observations are $\RZ_{\opw}:=(\RX_1, \RX_5)$.
The wiretap secret key capacity is given by $   \skc=I(\RZ_1 \wedge \RZ_2) = H(\RX_2,\RX_3)=2$.
\begin{enumerate}[label=\alph*)]
    \item Let $g_1(x_1,\ldots,x_4):=x_1+x_4$, where $+$ is over $\mathbb{F}_2$. The function $\RG_1=g_1(\RX_1,\ldots,\RX_4)=\RX_1+\RX_4$ is securely computable by the users because $H(\RG_1)=1 < 2=\skc$.
    \item Let $g_2(x_1,\ldots,x_4):=(x_2, x_3, x_4)$. The function $\RG_2=g_2(\RX_1,\ldots,\RX_4)=(\RX_2, \RX_3, \RX_4)$ is not securely computable because $\skc=2<3=H(\RG_2)$.
\end{enumerate}
\end{example}

The following lemma addresses the secure function computation problem when the wiretapper side information is non-trivial.

\begin{lemma}\label{thm:func:necessary}
For a general source $(\RZ_V,\RZ_{\opw})$, any computable function $\RG$ must satisfy 
\begin{align}\label{eq:func_comp:lb}
    H(\RG|\RZ_{\opw})-\rl \leq \wskc.
\end{align}    
\end{lemma}
\begin{proof} We can  prove this result along the lines of the proof of Theorem~\ref{thm:RL:lb}. The idea is that given a discussion scheme that computes the function $\RG$ with leakage rate $\rl$, one can apply privacy amplification to extract a secret key of rate $H(\RG|\RZ_{\opw})-\rl$ from the recovered source.
\end{proof}

 It is clear from \eqref{eq:func_comp:lb} that if $\RG$ is securely computable (i.e., $\rl=0$), then  $H(\RG|\RZ_{\opw})\leq \wskc$. 
 \begin{example}
Let $\RX_i$,  $i \in \{1,2,\ldots,5\}$, be a collection of i.i.d. Ber($\frac{1}{2}$) random variables. A two user source $(m=2)$ is defined as follows.  
\begin{align*}
    \RZ_1&:=(\RX_1,\RX_2,\RX_3)\\
    \RZ_2&:=(\RX_3,\RX_4,\RX_5)\\
    \RZ_{\opw}&:=(\RX_1, \RX_5)
\end{align*}
% User $1$'s observations are $\RZ_1:=(\RX_1,\RX_2,\RX_3)$ and User $2$'s observations are $\RZ_2:=(\RX_3,\RX_4,\RX_5)$. The wiretapper's observations are $\RZ_{\opw}:=(\RX_1, \RX_5)$.
It follows from Theorem~\ref{thm:fls} that the wiretap secret key capacity is
\begin{align*}
    \wskc&=I(\RZ_1 \wedge \RZ_2 | \RG) \\
&= I(\RX_1,\RX_2,\RX_3 \wedge \RX_3,\RX_4,\RX_5 | \RX_1)\\
&= H(\RX_1,\RX_2,\RX_3 | \RX_1)-H(\RX_1,\RX_2,\RX_3 | \RX_3,\RX_4,\RX_5, \RX_1)\\
&= 2-1=1 \text{ bit },
\end{align*}
where $\RG=\RX_1$, which is the maximal common function of $\RZ_1$ and $\RZ_{\opw}$. Let $g_1(x_1,\ldots,x_5):=(x_2, x_3, x_4)$. The function $\RG_1=g_1(\RX_1,\ldots,\RX_5)=(\RX_2, \RX_3, \RX_4)$ is not securely computable because $\wskc=1<3=H(\RG_1|\RZ_{\opw})$ violates the condition in Lemma~\ref{thm:func:necessary}.

\end{example}

 Establishing the reverse direction, namely, that $H(\RG|\RZ_{\opw})< \wskc$ implies secure computability of $\RG$, is an interesting open problem. The following lemma addresses the minimum leakage rate of those source models for which this reverse direction holds. 

\begin{lemma}\label{lemma:sc}
If every function $\RS$ with $H(\RS|\RZ_{\opw}) < \wskc$ is securely computable via omniscience, then the minimum leakage rate for computing a given function $\RG$ is
     \begin{align}\label{eq:func:11}
         \rl = |H(\RG|\RZ_{\opw}) - \wskc|^{+}.
     \end{align} 
Here $|u|^{+}\triangleq \max\{u,0\}$.
\end{lemma}
\begin{proof}
    Assume first that $H(\RG|\RZ_{\opw}) < \wskc$. By the hypothesis of the lemma, $\RG$ is securely computable, so that the minimum leakage rate for computing $\RG$ is $\rl=0$. So, it suffices to consider the case when $H(\RG|\RZ_{\opw})\geq \wskc$. By Theorem~\ref{thm:func:necessary}, we have $\rl \geq H(\RG|\RZ_{\opw})- \wskc$. We will now show the reverse inequality  $\rl \leq H(\RG|\RZ_{\opw})- \wskc$, which will imply that $\rl=H(\RG|\RZ_{\opw})- \wskc$.

    Fix $\epsilon > 0$. For $m$ large enough, consider the source $(\RZ_V^m, \RZ_{\opw}^m)$, which corresponds to $m$ i.i.d. realizations of the source $(\RZ_V, \RZ_{\opw})$. By the random binning argument, there exists a function $\tRG$ of $\RG^m$ such that 
    \begin{align}\label{eq:func:1}
\wskc(\RZ_V^m\|\RZ_{\opw}^m) \geq H(\tRG|\RZ_{\opw}^m)> \wskc(\RZ_V^m\|\RZ_{\opw}^m) -\epsilon,
    \end{align}
    where $\wskc(\RZ_V^m\|\RZ_{\opw}^m)=m\wskc(\RZ_V\|\RZ_{\opw})$.

Note that the hypothesis also holds for the $m$-letter source. With $\RZ_V, \RZ_{\opw}$, and $\RG$ replaced by $\RZ_V^m, \RZ^m_{\opw}$ , and $\tRG$, respectively, there exists an omniscience scheme $\tRF$ such that

\begin{align}
 &\frac{1}{n}I(\tRG^n \wedge \tRF|\RZ_{\opw}^{mn}) \leq \delta_n^{(m)}\label{eq:func:2}\\
 &\liminf_{n \to \infty} \Pr(\RE_1^{(mn)} = \ldots =\RE_m^{(mn)} = \RZ_V^{mn}) = 1 
\end{align}
for some $\delta_n^{(m)} \to 0$ as $n \to \infty$.
Observe that  $\RG^m$ can be computed by all the users through $\tRF$ because $\tRF$ is an omniscience scheme. Furthermore, the leakage rate for computing $\RG^m$ is 
  \begin{align*}
        \frac{1}{mn} I(\RG^{mn}\wedge \tRF | \RZ^{mn}_{\opw}) & \stackrel{(a)}{=} \frac{1}{mn} I(\tRG^n, \RG^{mn}\wedge \tRF | \RZ^{mn}_{\opw})\\
        &= \frac{1}{mn} I(\tRG^n\wedge \tRF | \RZ^{mn}_{\opw}) + \frac{1}{mn} I(\RG^{mn}\wedge \tRF | \RZ^{mn}_{\opw}, \tRG^n)\\
        & \stackrel{(b)}{\leq}  \frac{1}{m}\delta_n^{(m)} + \frac{1}{mn} I(\RG^{mn}\wedge \tRF | \RZ^{mn}_{\opw}, \tRG^n)\\
        & \leq  \frac{1}{m}\delta_n^{(m)} + \frac{1}{mn} H(\RG^{mn} | \RZ^{mn}_{\opw}, \tRG^n)\\
        & \stackrel{(c)}{\leq}  \frac{1}{m}\delta_n^{(m)} + \frac{1}{mn} \left[ H(\RG^{mn} | \RZ^{mn}_{\opw}) -  H(\tRG^{n} | \RZ^{mn}_{\opw})\right]\\
        & =  \frac{1}{m}\delta_n^{(m)} +  H(\RG | \RZ_{\opw}) -  \frac{1}{m} H(\tRG | \RZ^{m}_{\opw})\\
        & \stackrel{(d)}{\leq} \frac{1}{m}\delta_n^{(m)} +  H(\RG | \RZ_{\opw}) -  mn \wskc +\epsilon,
    \end{align*}
where $(a)$ and $(c)$ are due to the fact the $\tRG$ is a function of $\RG^m$, $(b)$ and $(d)$ follow respectively from \eqref{eq:func:2} and \eqref{eq:func:1}. By taking limit in the above chain of inequalities, we get 
$$\rl \leq \limsup_{n \to \infty} \frac{1}{mn} I(\RG^{mn}\wedge \tRF | \RZ^{mn}_{\opw}) = H(\RG | \RZ_{\opw}) - \wskc +\epsilon.$$
As $\epsilon$ is arbitrary, we get $\rl \leq H(\RG | \RZ_{\opw}) - \wskc$, completing the proof.
\end{proof}

Next we will consider the secure function computation problem for finite linear sources.

\begin{theorem}\label{thm:func:linearscheme}
  For a finite linear source $(\RZ_V, \RZ_{\opw})$, any leakage rate in computing a linear function $\RG$ via a perfect linear communication scheme is also achievable by a linear omniscience scheme.
\end{theorem}
\begin{proof} The proof is similar to the proof of Theorem~\ref{thm:linearscheme}. Let $\RF^{(n)}$ be a perfect linear communication scheme that computes $\RG$, but $\RF^{(n)}$ need not achieve omniscience. By  \cite[Theorem~1]{chan19oneshot}, we can assume that $\RF^{(n)}$ is a linear function of $\RZ_V^n$ alone (additional randomization by any user is not needed) and $\RG^n$ is also a linear function of $\RZ_V^n$. 

If $\RF^{(n)}$ already attains omniscience, then we are done. If not, for some $i,j \in V$, $i \ne j$, we have a component $\RX \in \Fq$ of random vector $\RZ_i^{n}$ such that
\begin{align*}
    H(\RX \mid \RF^{(n)}, \RZ_j^n) &\neq 0.
\end{align*}
We will show that there exists an additional discussion $\RF'^{(n)}$ such that
\begin{align}
    H(\RX \mid \RF^{(n)},\RF'^{(n)}, \RZ_j^n) &= 0 \label{eq:func:additional_recov}
\end{align}  and 
\begin{align}
 I(\RG^{(n)} \wedge \RF^{(n)},\RF'^{(n)}| \RZ_{\opw}^n) = I(\RG^{(n)} \wedge \RF^{(n)}| \RZ_{\opw}^n).\label{eq:func:additional_secrecy}
\end{align}
If $(\RF^{(n)},\RF'^{(n)})$ achieves omniscience, we are done; else, we repeat the construction in our argument till we obtain the desired omniscience-achieving communication.

So, consider the non-trivial case where $H(\RX\mid \RF^{(n)}, \RZ_j^n) \neq 0$ and $I(\RG^n \wedge \RF^{(n)}, \RX| \RZ_{\opw}^n) \neq 0$. (If $I(\RG^n \wedge \RF^{(n)}, \RX| \RZ_{\opw}^n)=0$, then user $i$ transmits $\RF'^{(n)}:= \RX$ which satisfies \eqref{eq:func:additional_recov} and \eqref{eq:func:additional_secrecy}.)
%(Since $\RG^n$ must be recoverable from $(\RF^{(n)}, \RZ_i^n)$, the alternative assumption $I(\RG^n \wedge \RF^{(n)}, \RZ_i^{n}, \RZ_{\opw}^n) = 0$ yields $H(\RG^n) = 0.$) 
Let $\RL^{(n)}$ be a common linear function, not identically $0$, of $\RG^n$ and $(\RF^{(n)}, \RX, \RZ_{\opw}^n))$ taking values in $\Fq$. Such a function exists since $I(\RG^n \wedge \RF^{(n)}, \RX, \RZ_{\opw}^n) \neq 0$. So, we can write \begin{align}
    \RL^{(n)}=\RG^n\MM_K= a\RX + \RF^{(n)}\MM_{F}+\RZ_{\opw}^n \MM_{\opw} \label{eq:func:mcf_lin_perfect}
\end{align}
for some non-zero element $a \in \Fq$, and some  column vectors $\MM_K \neq \M0, \MM_{F},$ and $\MM_{\opw}$ over $\Fq$. (Here, $\RL^{(n)},\RG^n,\RF^{(n)}$ and $\RZ_{\opw}^n$ are the random row vectors with entries uniformly distributed over $\Fq$.) Note if the coefficient $a$ in the above linear combination is zero element in $\Fq$, then \eqref{eq:func:additional_secrecy} is satisfied for any $\RF'^{(n)}$. In particular, we can choose $\RF'^{(n)}$ to be an omniscience scheme, which satisfies \eqref{eq:func:additional_recov}.

Now consider the case of non-zero $a$. Define $\RF'^{(n)}:= \RG^n\MM_K - a\RX$. User $i$ can compute $\RF'^{(n)}$, as it is a function of  $\RG^n$ and $\RZ_i^n$, and transmit it publicly. Let us verify that $\RF'^{(n)}$ satisfies \eqref{eq:func:additional_recov} and \eqref{eq:func:additional_secrecy}. For \eqref{eq:func:additional_recov}, observe that  $H(\RX\mid \RF^{(n)},\RF'^{(n)}, \RZ_j^n) \leq H(\RX\mid \RF'^{(n)}, \RG^n)=0$, the inequality following from $H(\RG^n\mid\RF^{(n)},\RZ_j^n)=0$, and the  equality from the fact that $\RX$ is recoverable from $(\RF'^{(n)}, \RG^n)$. For \eqref{eq:func:additional_secrecy}, $I(\RG^{(n)} \wedge \RF^{(n)},\RF'^{(n)}| \RZ_{\opw}^n) - I(\RG^{(n)} \wedge \RF^{(n)}| \RZ_{\opw}^n)=I(\RG^n\wedge \RF'^{(n)}|\RF^{(n)}, \RZ_{\opw}^n)\leq H(\RF'^{(n)}|\RF^{(n)}, \RZ_{\opw}^n) = 0$,  which follows from \eqref{eq:func:mcf_lin_perfect} and the definition of $\RF'^{(n)}$. This completes the proof.
\end{proof}

}

\section{Discussion}\label{sec:discussion}
In this paper, we have explored the possibility of a duality between the wiretap secret key agreement problem and the secure omniscience problem. Though the problem of characterizing the class of sources for which these two problems are dual to each other is far from being solved completely, we made some progress in the case of limited interaction (with at most two communications allowed), and for the class of finite linear sources. Furthermore, we have made use of \eqref{basic_ineq} to identify several equivalent conditions for the positivity of $\wskc$ in the multi-user case, which is an extension of a recent two-user result of \cite{amin2020}.

By limiting the number of messages to two, we showed that for the source in Lemma~\ref{lem:twowaycounter}, the duality does not hold. This result seems to indicate that the duality does not always hold. 
In particular, we believe that for the DSBE source  considered in Lemma~\ref{lem:twowaycounter}, the duality does not hold even if we relax the restriction on the number of messages (Conjecture~\ref{conj:duality:msg_converse}). To prove this result, we actually  need a single-letter lower bound on $\rl$ that strictly improves our current bound $H(\RZ_V|\RZ_{\opw}) - \wskc$.  However, it has turned out to be challenging to find a better lower bound on $\rl$. 

\newtheorem*{C1}{Conjecture~\ref{conj:duality:fls}}
\begin{C1}
$\rl = H(\RZ_V|\RZ_{\opw}) - \wskc$ holds for finite linear sources.
\end{C1}

\begin{Conjecture}\label{conj:duality:msg_converse}
For $r\geq m$, $\wskc^{(r)}>0$ need not imply $\rl^{(r)} < H(\RZ_V|\RZ_{\opw})$. Moreover, with no restriction on the number of messages, $\wskc>0$ need not imply $\rl < H(\RZ_V|\RZ_{\opw})$.
\end{Conjecture}

In our attempt to resolve the duality for finite linear sources (Conjecture~\ref{conj:duality:fls}), we were able to prove it in the case of two-user FLS models and in the case of tree-PIN models. The proof construction mainly relies on the idea of aligning the communication with the wiretapper side information. Specifically, in the case of tree-PIN models, we used a reduction to obtain an irreducible source on which we constructed an $\rco$-achieving omniscience scheme that aligns perfectly with the wiretapper side information. In fact, we have shown that this construction is $\rl$-achieving. 

However, for more general PIN sources, this proof strategy fails. The notion of irreducibility in Definition~\ref{def:irreducible} can certainly be extended to general PIN sources. However, it turns out that this definition of irreducibility is not good enough. There are irreducible PIN sources on graphs with cycles whose $\rl$ is not achieved by an omniscience protocol of rate $\rco$ that is perfectly aligned with the wiretapper side information. So, proving the duality conjecture for sources beyond the tree-PIN model could be interesting as it will require new communication strategies other than the ones we used in the proof of the tree-PIN model with a linear wiretapper.  

\appendices

% \section{Proof of Lemma~\ref{lem:bc_oneway}}
% \input{proof_lessnoisy}\label{app:proof_lessnoisy}
 \section{Proof of the Lower Bound in Theorem~\ref{thm:RL:lb}}
\label{app:thm:proof_rl_bound}
Fix $\epsilon>0$. Let $\left(\RF^{(N)}\right)_{N\geq 1}$ be a communication scheme allowing users to achieve omniscience such that
\begin{align}
    \limsup_{N \to \infty} \frac{1}{N}I\left(\RF^{(N)} \wedge \RZ_V^N \big| \RZ_{\opw}^N\right) < \rl + \epsilon/2
\end{align}
and 
\begin{align}
\liminf_{N \to \infty} \Pr(\RE_1^{(N)} = \dots =\RE_m^{(N)} = \RZ_V^N) >1-\epsilon/2.
\end{align}
Then, for a large enough integer $N_0$,  $\frac{1}{N}I\left(\RF^{(N)} \wedge \RZ_V^N \big| \RZ_{\opw}^N\right) < \rl + \epsilon$ and $\Pr(\RE_1^{(N)} = \dots =\RE_m^{(N)} = \RZ_V^N) >1-\epsilon$  for all  $N \geq N_0$. We will use the specific communication scheme $\RF^{(N_0)}$ to construct a SKA scheme that achieves a secret key rate close to $H(\RZ_V|\RZ_{\opw})-\rl$. It suffices to construct this protocol only for $n$'s that are integer multiples of $N_0$ (cf. Theorem 3.2 of \cite{csiszar08}).

% Let $(\tRF^{(n)})_{n\geq 1}$ and $(\RK^{(n)})_{n\geq 1}$ be the  communications and keys, respectively, associated to the resulting WSKA scheme. In fact, it is enough to construct this protocol only for $n$'s that are integer multiples of $N_0$ with the desired properties; because for $kN_0 \leq n < (k+1)N_0$, we can use only the first $kN_0$ realizations of the source in the key generation protocol, i.e., set  $\tRF^{(n)}:= \tRF^{(kN_0)}$ and $\RK^{(n)}:=\RK^{(kN_0)}$. Under this construction, it is easy to show that the desired properties of the subsequence in $k$ will carry over to the whole sequence in $n$.

The constructed omniscience communication $\tRF^{(kN_0)}$ comprises two parts. The first part is of the form $(\RF^{(N_0)})^k$, which is obtained by applying the function corresponding to $\RF^{(N_0)}$ individually on the $k$ sub-blocks of size $N_0$ of the $kN_0$ source realizations. The second part contains an extra Slepian-Wolf communication $f_k(\RZ_V^{kN_0})$ (\cite[Lemma~3.1]{csiszar08}) that ensures the recoverability of $\RZ_V^{kN_0}$ at each user with high probability. The form of  $f_k(\RZ_V^{kN_0})$ is specified next.

Let  $(\RS_i)_{i \in V}$ be the private randomness that can potentially be used by the terminals, which satisfies $P_{\RZ^{N_0}_V\RZ_{\opw}^{N_0}\RS_V}=P_{\RZ^{N_0}_V\RZ_{\opw}^{N_0}}\prod_{i \in V}P_{\RS_i}$. Fix a user $j \in  V$, and let $i \neq j$ be another user. By applying Lemma~3.1 of \cite{csiszar08} with $\RU=\RZ_i^{N_0}$ and $\RV=({\RF}^{(N_0)}, \RZ_j^{N_0}, \RS_j)$, we can conclude that there exists a function $g_{ij}$ of $\RZ_i^{kN_0}$ such that $\RZ_i^{kN_0}$ is recoverable from $\tilde{\RF}^{(kN_0)}$, $\RZ_j^{kN}$, $\RS^k_j$ and $g_{ij}(\RZ_i^{kN_0})$ with probability going to one as $k \to \infty$, and $H(g_{ij}(\RZ_i^{kN_0}))< k(\epsilon N_0 \log|\mc{Z}_i|+ h_2(\epsilon))$. This means that user $i$ can transmit $g_{ij}(\RZ_i^{kN_0})$ to user $j$ so that $\RZ_i^{kN_0}$ is recoverable with high probability. Similarly, every user can communicate in this way to each of the other users. By using the union bound, we can conclude that this communication allows every user to recover the complete source $\RZ_V^{kN_0}$ with probability going to one as $k \to \infty$. The extra communication $f_k(\RZ_V^{kN_0})$ is of the form $\{g_{ij}(\RZ_i^{kN_0}):i \in V, j \in  V\setminus \{i\}\}$, for which $R_f:=\frac{1}{k}H(f_k(\RZ_V^{kN_0})) <(|V|-1)\left[\epsilon N_0 \log|\mc{Z}_V|+h_2(\epsilon)\right]$.  The users now can apply a function on the recovered source to extract a key $\RK^{(kN_0)}$, which satisfies the recoverability condition \eqref{eq:sk:recoverability}.

To obtain a function that makes this key satisfy the secrecy condition \eqref{eq:sk:secrecy}, we rely on the 
balanced coloring lemma~\cite[Lemma~B.3]{csiszar04}. We will apply this lemma  with $\RU=\RZ_V^{N_0}$, $\RV=(\RF^{(N_0)}, \RZ^{N_0}_{\opw})$ and $f(\RU^k)= f_k(\RZ_V^{kN_0})$, which satisfy the inequality 
\begin{align*}
    H(\RU|\RV)-R_f&=H(\RZ_V^{N_0}|\RF^{(N_0)}, \RZ^{N_0}_{\opw})-R_f \\
&> H(\RZ_V^{N_0}|\RF^{(N_0)}, \RZ^{N_0}_{\opw})-(|V|-1)\left[\epsilon N_0 \log|\mc{Z}_V|+h_2(\epsilon)\right].
\end{align*}
The balanced coloring lemma guarantees the existence of a key function $g:\mc{Z}_V^{kN_0} \to \{1, \ldots, 2^{kR_g}\}$ with $$R_g= H(\RZ_V^{N_0}|\RF^{(N_0)}, \RZ^{N_0}_{\opw})-(|V|-1)\left[\epsilon N_0 \log|\mc{Z}_V|+h_2(\epsilon)\right]$$ such that  $kR_g-H(g(\RU^k))+I(g(\RU^k) \wedge f(\RU^k), \RV^k) \longrightarrow 0$ exponentially quickly in $k$. Then, $\RK^{(kN_0)}=g(\RZ_V^{kN_0})$ has rate  $\frac{1}{kN_0}\log|\mc{K}^{(kN_0)}|=\frac{1}{kN_0}kR_g= \frac{1}{N_0} H(\RZ_V^{N_0}|\RF^{(N_0)}, \RZ^{N_0}_{\opw})-(|V|-1)\left[\epsilon \log|\mc{Z}_V|+\frac{h_2(\epsilon)}{N_0}\right]$, and by virtue of the balanced coloring lemma, it satisfies the secrecy condition. Therefore,
\begin{align*}
    \wskc &\geq \frac{1}{N_0} H(\RZ_V^{N_0}|\RF^{(N_0)}, \RZ^{N_0}_{\opw})-(|V|-1)\left[\epsilon \log|\mc{Z}_V|+\frac{h_2(\epsilon)}{N_0}\right]\\
    & = H(\RZ_V| \RZ_{\opw})- \frac{1}{N_0}I(\RZ_V^{N_0}\wedge \RF^{(N_0)}| \RZ^{N_0}_{\opw})-(|V|-1)\left[\epsilon \log|\mc{Z}_V|+\frac{h_2(\epsilon)}{N_0}\right]\\
    & >H(\RZ_V| \RZ_{\opw})- \rl -\epsilon-(|V|-1)\left[\epsilon \log|\mc{Z}_V|+\frac{h_2(\epsilon)}{N_0}\right].
\end{align*}
Since, for a fixed $\epsilon$, $N_0$ can be made arbitrarily large, and $\epsilon$ is also arbitrary, we have $\wskc \geq H(\RZ_V| \RZ_{\opw})- \rl$.

\section{Proof of Theorem~\ref{thm:linearscheme}}
\label{app:thm:proof_linearscheme}
It suffices to show that $\wskc$ can be achieved through omniscience because then
\begin{align*}
    nH(\RZ_V|\RZ_{\opw})&\geq I(\RK^{(n)},\RF^{(n)} \wedge \RZ_V^n \mid \RZ_{\opw}^n)\\
                        & = I(\RF^{(n)} \wedge \RZ_V^n \mid \RZ_{\opw}^n)+I(\RK^{(n)} \wedge \RZ_V^n \mid \RZ_{\opw}^n,\RF^{(n)})\\
                        & \geq I(\RF^{(n)} \wedge \RZ_V^n \mid \RZ_{\opw}^n) + n(\wskc-\delta_n) 
\end{align*}
for some $\delta_n \to 0$, where the last inequality follows from the fact an optimal key is recoverable from $\RZ_V^n$. By taking limsup on both side of the above inequality after normalizing by $n$, we get $ H(\RZ_V|\RZ_{\opw})\geq \limsup_{n \to \infty}\frac{1}{n}I(\RF^{(n)} \wedge \RZ_V^n \mid \RZ_{\opw}^n)  - \wskc \geq \rl- \wskc$. Therefore, $\rl \leq H(\RZ_V|\RZ_{\opw}) -\wskc$.

Let $(\RF^{(n)}, \RK^{(n)})$ be a communication-key pair of a linear perfect SKA scheme that achieves $\wskc$, but $\RF^{(n)}$ need not achieve omniscience. By  \cite[Theorem~1]{chan19oneshot}, we can assume that $\RF^{(n)}$ is a linear function of $\RZ_V^n$ alone (additional randomization by any user is not needed) and the key is also a linear function of $\RZ_V^n$. 

If $\RF^{(n)}$ already attains omniscience, then we are done. If not, for some $i,j \in V$, $i \ne j$, we have a component $\RX \in \Fq$ of random vector $\RZ_i^{n}$ such that
\begin{align*}
    H(\RX \mid \RF^{(n)}, \RZ_j^n) &\neq 0.
\end{align*}
We will show that there exists an additional discussion $\RF'^{(n)}$ such that
\begin{align}
    H(\RX \mid \RF^{(n)},\RF'^{(n)}, \RZ_j^n) &= 0 \label{eq:additional_recov}
\end{align}  and 
\begin{align}
 I(\RK^{(n)} \wedge \RF^{(n)},\RF'^{(n)}, \RZ_{\opw}^n) = 0.\label{eq:additional_secrecy}
\end{align}
If $(\RF^{(n)},\RF'^{(n)})$ achieves omniscience, we are done; else, we repeat the construction in our argument till we obtain the desired omniscience-achieving communication.

So, consider the non-trivial case where $H(\RX\mid \RF^{(n)}, \RZ_j^n) \neq 0$ and $I(\RK^{(n)} \wedge \RF^{(n)}, \RX, \RZ_{\opw}^n) \neq 0$. (If $I(\RK^{(n)} \wedge \RF^{(n)}, \RX, \RZ_{\opw}^n)=0$, then user $i$ transmits $\RF'^{(n)}:= \RX$ which satisfies \eqref{eq:additional_recov} and \eqref{eq:additional_secrecy}.)
%(Since $\RK^{(n)}$ must be recoverable from $(\RF^{(n)}, \RZ_i^n)$, the alternative assumption $I(\RK^{(n)} \wedge \RF^{(n)}, \RZ_i^{n}, \RZ_{\opw}^n) = 0$ yields $H(\RK^{(n)}) = 0.$) 
Let $\RL^{(n)}$ be a common linear function, not identically $0$, of $\RK^{(n)}$ and $(\RF^{(n)}, \RX, \RZ_{\opw}^n))$ taking values in $\Fq$. Such a function exists since $I(\RK^{(n)} \wedge \RF^{(n)}, \RX, \RZ_{\opw}^n) \neq 0$. So, we can write \begin{align}
    \RL^{(n)}=\RK^{(n)}\MM_K= a\RX + \RF^{(n)}\MM_{F}+\RZ_{\opw}^n \MM_{\opw} \label{eq:mcf_lin_perfect}
\end{align}
for some non-zero element $a \in \Fq$, and some  column vectors $\MM_K \neq \M0, \MM_{F},$ and $\MM_{\opw}$ over $\Fq$. (Here, $\RL^{(n)},\RK^{(n)},\RF^{(n)}$ and $\RZ_{\opw}^n$ are the random row vectors with entries uniformly distributed over $\Fq$.) Note  the coefficient $a$ in the above linear combination must be a non-zero element in $\Fq$. If not, then $\RL^{(n)}(=\RK^{(n)}\MM_K= \RF^{(n)}\MM_{F}+\RZ_{\opw}^n \MM_{\opw})$ is a non-constant common function of $\RK^{(n)}$ and $(\RF^{(n)}, \RZ_{\opw}^n)$. This contradicts the secrecy condition $I(\RK^{(n)} \wedge \RF^{(n)},  \RZ_{\opw}^n) = 0$.

Define $\RF'^{(n)}:= \RK^{(n)}\MM_K - a\RX$. User $i$ can compute $\RF'^{(n)}$, as it is a function of  $\RK^{(n)}$ and $\RZ_i^n$, and transmit it publicly. Let us verify that $\RF'^{(n)}$ satisfies \eqref{eq:additional_recov} and \eqref{eq:additional_secrecy}. For \eqref{eq:additional_recov}, observe that  $H(\RX\mid \RF^{(n)},\RF'^{(n)}, \RZ_j^n) \leq H(\RX\mid \RF'^{(n)}, \RK^{(n)})=0$, the inequality following from $H(\RK^{(n)}\mid\RF^{(n)},\RZ_j^n)=0$, and the  equality from the fact that $\RX$ is recoverable from $(\RF'^{(n)}, \RK^{(n)})$. For \eqref{eq:additional_secrecy}, $I(\RK^{(n)}\wedge \RF^{(n)},\RF'^{(n)}, \RZ_{\opw}^n)= I(\RK^{(n)}\wedge \RF^{(n)}, \RZ_{\opw}^n) = 0$,  the first equality being a consequence of $\RF'^{(n)}$ also being expressible as $\RF^{(n)}\MM_{F}+\RZ_{\opw}^n \MM_{\opw}$, and the last equality from the secrecy condition of the key, i.e., $I(\RK^{(n)} \wedge \RF^{(n)},  \RZ_{\opw}^n) = 0$. This completes the proof.

\section{Proof of Theorem~\ref{thm:fls}} \label{app:thm:fls}

Before presenting the proof, we introduce some notation that will be used in the achievability part of our argument. It is known, from the proof of Lemma~5.2 of \cite{chan18zero}, that a finite linear source $(\RX, \RY)$ can be decomposed as
\begin{align}
        \RX &= (\RX',\RC), \label{eq:X_X'C}\\
        \RY &= (\RY',\RC), \label{eq:Y_Y'C}
\end{align}
where  $\RX'$ (resp. $\RY'$) is a linear function of $\RX$ (resp. $\RY$) and $\RC=\op{mcf}(\RX, \RY)$ is a linear function of each of $\RX$ and $\RY$; altogether, they satisfy the independence relation
\begin{align}
    H(\RX',\RC,\RY') = H(\RX')+H(\RC)+H(\RY'). \label{eq:XCY_independence}
\end{align}
We use the notation $\RX \setminus \RY$ (resp. $\RY \setminus \RX$) to denote $\RX'$ (resp. $\RY'$).

\begin{proof}
\emph{Converse part.} Note that $\RG$ satisfies the Markov condition $\RG-\RZ_{\opw}-\RZ_V$ because $\RG$ is a function of $\RZ_{\opw}$ whether it is chosen to be $\RG_1$, $\RG_2$ or both. By \eqref{eq:RL:lb}, we have
\begin{align*}
\rl &\geq H(\RZ_V|\RZ_{\opw}) - \wskc(\RZ_V\|\RZ_{\opw})\\
&\utag{a}\geq H(\RZ_V|\RZ_{\opw}) - \pkc(\RZ_V|\RG)\\
&\utag{b}=H(\RZ_V|\RZ_{\opw}) - I(\RZ_1\wedge \RZ_2 | \RG)
\end{align*}
where \uref{a} is because for $\RW-\RZ_{\opw}-\RZ_V$, $\wskc(\RZ_V\|\RZ_{\opw}) \leq \pkc(\RZ_V|\RW)$ \cite[Theorem~4]{csiszar04} and $\RG$ forms the Markov condition $\RG-\RZ_{\opw}-\RZ_V$, and we have used the fact that $\pkc(\RZ_V|\RG)=I(\RZ_1\wedge \RZ_2 | \RG)$ \cite[Theorem~2]{csiszar04} in \uref{b} .\\

\emph{Achievability part.}
It suffices to prove the reverse inequality for $\RG=\RG_1$, i.e.,
\begin{align}
    \rl \leq H(\RZ_V|\RZ_{\opw}) - \underbrace{I(\RZ_1\wedge \RZ_2|\RG_1)}_{`(1)} \label{eq:fls:ub}
\end{align}
because then the reverse inequality will also hold for $\RG=\RG_2$ by symmetry, and for $\RG=(\RG_1,\RG_2)$ since 
\[I(\RZ_1\wedge \RZ_2|\RG_1,\RG_2)\leq I(\RZ_1\wedge \RZ_2,\RG_2 \mid \RG_1) = I(\RZ_1\wedge \RZ_2|\RG_1)\]
by the assumption that $\RG_2$ is a function of $\RZ_2$.

The desired reverse inequality~\eqref{eq:fls:ub} will follow from the following upper bound with an appropriate choice of public discussion $\RF'$ of block length $1$, i.e.,
\begin{align*}
    \rl &\leq \underbrace{\rco(\RZ_V|\RF')}_{=H(\RZ_V|\RF')-I(\RZ_1\wedge \RZ_2|\RF')} + \underbrace{I(\RZ_V\wedge \RF'|\RZ_{\opw})}_{=H(\RZ_V|\RZ_{\opw})-H(\RZ_V|\RZ_{\opw},\RF')}\\
    &= H(\RZ_V|\RZ_{\opw})+ \underbrace{I(\RZ_V\wedge \RZ_{\opw}|\RF')}_{`(2)} -\underbrace{I(\RZ_1\wedge \RZ_2|\RF')}_{`(3)}.
\end{align*}
The idea behind this upper bound involves splitting of the leakage rate into two components after a discussion $\RF'$: one component is the leakage rate due to  $\RF'$, and the other one is the residual leakage rate for subsequent omniscience, which is upper bounded by $\rco(\RZ_V|\RF')$. Note that the equality $\rco(\RZ_V|\RF')=H(\RZ_V|\RF')-I(\RZ_1\wedge \RZ_2|\RF')$ follows from the characterization of the conditional $\rco$ given in \cite[Proposition~1]{csiszar04}.
It suffices to give a feasible $\RF'$ with $`(2)=0$ and $`(3)=`(1)$. We will construct this $\RF'$ by decomposing the source $(\RZ_V, \RZ_{\opw})$.

For the source $(\RX, \RY)$ with $\RX=(\RZ_1, \RG_1)$ and $\RY=(\RZ_2, \RG_1)$, the decomposition is as follows: 
\begin{align}
        (\RZ_1, \RG_1) &:= (\RX_a ,\RX_c), \label{eq:Z'1}\\
        (\RZ_2, \RG_1) &:= (\RX_b ,\RX_c),\label{eq:Z'2}
\end{align}
where $\RX_a$, $\RX_b$, and $\RX_c = \op{mcf}((\RZ_1, \RG_1), (\RZ_2, \RG_1))$ are uniformly random row vectors over some finite field, say $\bbF_q$, satisfying the independence relation
\begin{align}
      H(\RX_a,\RX_b,\RX_c) &= H(\RX_a)+H(\RX_b)+H(\RX_c)\label{eq:XG}.
\end{align}

Observe that $\RG_1$ is a linear common function of $(\RZ_1, \RG_1)$ and $(\RZ_2, \RG_1)$. Using the decomposition \eqref{eq:Z'1} and \eqref{eq:Z'2}, we can write $\RG_1= \RX_a\MM_a+\RX_c\MM_c =\RX_b\MM_b+\RX_c\tMM_c$ for some matrices $\MM_a,\MM_b,\MM_c$ and $\tMM_c$. Therefore, we have $\RX_a\MM_a-\RX_b\MM_b +\RX_c(\MM_c -\tMM_c)=0$. But, $\RX_a$, $\RX_b$, and $\RX_c$ are mutually independent, which implies (for finite linear sources) that $\MM_a=\MM_b=\MM_c -\tMM_c=0$ and $\RG_1= \RX_c\MM_c$. This shows that $\RG_1$ is a linear function of $\RX_c$.
% Since any linear common function is a linear function of the m.c.f. of an FLS, we have that $\RG_1$ is a linear function of $\RX_c$.
% \begin{align}
%      H(\RG_1|\RX_c)=0. \label{eq:g1_xc}
% \end{align}
Let $\RX'_c := \RX_c \setminus \RG_1$. So, we can write $\RX_c=(\RX'_c, \RG_1)$, where $\RX'_c$ is independent of $\RG_1$, and both are linear functions of $\RX_c$. Therefore, we can further decompose the source in \eqref{eq:Z'1} and \eqref{eq:Z'2} as
\begin{align}
        (\RZ_1, \RG_1) &= (\RX_a ,\RX'_c,\RG_1), \label{eq:Z'1_decomp}\\
        (\RZ_2, \RG_1) &= (\RX_b ,\RX'_c, \RG_1),\label{eq:Z'2_decomp}
\end{align}
where $\RX_a$, $\RX_b$, and $\RX'_c$ are uniformly random row vectors such that 
\begin{align}
      H(\RX_a,\RX_b,\RX'_c,\RG_1) &= H(\RX_a)+H(\RX_b)+H(\RX'_c)+H(\RG_1)\label{eq:XG_decomp}.
\end{align}
Note that $(\RX_b ,\RX'_c)=\RZ_2 \setminus \RG_1$ which is a linear function of $\RZ_2$.

Now consider the decomposition of the form \eqref{eq:X_X'C} and \eqref{eq:Y_Y'C} for the source  $(\RZ_V,\RZ_{\opw})$: 
\begin{align}
        \RZ_V &:= (\RZ'_V, \RG_{\opw}), \label{eq:ZV_decomp}\\
        \RZ_{\opw} &:= (\RZ'_{\opw} ,\RG_{\opw}),\label{eq:ZW_decomp}
\end{align}
where $\RG_{\opw}$ is the m.c.f. of $\RZ_V$ and $\RZ_{\opw}$. As the components $(\RZ'_V, \RG_{\opw}, \RZ'_{\opw})$ are mutually independent by \eqref{eq:XCY_independence}, we have 
\begin{align}
    I(\RZ_V \wedge \RZ_{\opw}|\RG_{\opw}) =0.\label{eq:Z'Zw'}
\end{align}
Moreover, using the fact that the m.c.f. $\RG_{\opw}$ is a linear function of $\RZ_V$, and $(\RX_a,\RX_b,\RX'_c,\RG_1)$ is an invertible linear transformation of $\RZ_V$ (by \eqref{eq:Z'1_decomp} and \eqref{eq:Z'2_decomp}),  we can  write $\RG_{\opw}$ as
\begin{align}
    \RG_{\opw} = \RX_a \tMA + \RX_b \tMB + \RX'_c \tMC +\RG_1\tMD \label{eq:Gw}
\end{align}
for some deterministic matrices $\tMA$, $\tMB$, $\tMC$ and $\tMD$ over $\bbF_q$ such that $[\tMA^T \; \tMB^T \; \tMC^T \; \tMD^T]^T$ is a full column-rank matrix.  Since $\RG_1$ is a m.c.f. of $\RZ_1$ and $\RZ_{\opw}$, it is a linear function of $\RG_{\opw}$, which can also be argued along the same lines as the proof of $\RG_1$ is a linear function of $\RX_c$. So we can write \eqref{eq:Gw} as 
\begin{align}
    \RG_{\opw} = (\RX_a \MA + \RX_b \MB + \RX'_c \MC, \RG_1) \label{eq:Gw_decompose}
\end{align}
for some deterministic matrices $\MA$, $\MB$, and $\MC$ over $\bbF_q$ such that $\RX_a \MA + \RX_b \MB + \RX'_c \MC = \RG_{\opw} \setminus \RG_1$. 

Finally, by \eqref{eq:Z'1_decomp}, \eqref{eq:Z'2_decomp}, \eqref{eq:ZW_decomp}, \eqref{eq:Z'Zw'} and \eqref{eq:Gw_decompose}, we can write the decomposition of the source $(\RZ_V, \RZ_{\opw})$ as
\begin{align}
    \RZ_1 &= (\RX_a, \RX'_c,\RG_1),\label{eq:final_z1}\\
    (\RZ_2, \RG_1) &= (\RX_b ,\RX'_c, \RG_1)\label{eq:final_z2}\\
    \RZ_{\opw} &= (\RZ'_{\opw} ,\RX_a \MA + \RX_b \MB + \RX'_c \MC, \RG_1),\label{eq:final_zw}
\end{align}
where the components $\RX_a,\RX_b,\RX'_c,\RG_1,\RZ'_{\opw}$ and $\RX_a \MA + \RX_b \MB + \RX'_c \MC)$ satisfy the following independence relations:
\begin{enumerate}
    \item \eqref{eq:XG_decomp} holds, i.e., $\RX_a,\RX_b,\RX'_c$ and $\RG_1$ are mutually independent;
    \item $\RX_a,\RX'_c,\RG_1,\RZ'_{\opw}$ and $\RX_a \MA + \RX_b \MB + \RX'_c \MC$ are mutually independent.
\end{enumerate}
To verify the second independence relation above, it is enough to show that $I(\RX_a,\RX'_c\wedge \RG_1, \RZ'_{\opw}, \RX_a \MA + \RX_b \MB + \RX'_c \MC)=0$ because of \eqref{eq:XG_decomp},\eqref{eq:ZW_decomp}, and \eqref{eq:Gw_decompose}, which is equivalent to showing  $I(\RX_a,\RX'_c\wedge  \RZ'_{\opw}, \RX_a \MA + \RX_b \MB + \RX'_c \MC \mid \RG_1)=0$ by \eqref{eq:XG_decomp}. Note that  by \eqref{eq:Z'Zw'},  $0=I(\RZ_1\wedge \RZ_{\opw}|\RG_1)=I(\RZ_1, \RG_1\wedge \RZ_{\opw}|\RG_1)= I(\RX_a,\RX'_c,\RG_1\wedge \RG_1, \RZ'_{\opw}, \RX_a \MA + \RX_b \MB + \RX'_c \MC \mid \RG_1) = I(\RX_a,\RX'_c \wedge \RZ'_{\opw}, \RX_a \MA + \RX_b \MB + \RX'_c \MC \mid \RG_1)$.

%  Therefore, we have 
% \begin{align}
%     I(\RX_a, \RX'_c \wedge \RZ_{\opw})=0.
% \end{align}

Let us construct a linear communication using the components from the above decomposition. User $1$ transmits $\RF'_1 := (\RX_a \MA, \RG_1)$ using his source $\RZ_1= (\RX_a, \RX'_c,\RG_1)$. User $2$ communicates $\RF'_2 := \RX_b \MB + \RX'_c \MC$ using the source $(\RX_b ,\RX'_c)$ which is a function of $\RZ_2$. Define $\RF':=(\RF'_1, \RF'_2)$,  a valid  discussion of block length $n=1$.

By \eqref{eq:Z'Zw'}, we have  $0= I(\RZ_V \wedge \RZ_{\opw}|\RG_{\opw}) = I(\RZ_V, \RF' \wedge \RZ_{\opw}|\RG_{\opw})= I(\RF' \wedge \RZ_{\opw}|\RG_{\opw})+I(\RZ_V \wedge \RZ_{\opw}|\RG_{\opw}, \RF')=I(\RZ_V \wedge \RZ_{\opw}| \RF')$, where the last equality follows from $I(\RF' \wedge \RZ_{\opw}|\RG_{\opw})\leq I(\RZ_V \wedge \RZ_{\opw}|\RG_{\opw}) \stackrel{\eqref{eq:Z'Zw'}}{=}0$, and $H(\RG_{\opw}|\RF')=0$.  Hence we conclude that 
\begin{align}
    `(2) = I(\RZ_V \wedge \RZ_{\opw} \mid \RF')=0 \label{eq:Z'wF'}
\end{align}

Let us show the remaining inequality $`(3)=`(1)$. By the independence relation 1), we evidently have
\begin{align}
I(\RX_a\MA \wedge \RX_b, \RX'_c \mid \RG_1) &= 0 \label{eq:aux1}
    % &I(\RX_a\MA \wedge \RX_b, \RX'_c|\RG_1) \notag\\
    % &\mkern 30mu= H(\RX_a\MA|\RG_1)- H(\RX_a\MA|\RG_1, \RX_b, \RX'_c) \notag \\ &\mkern 30mu= H(\RX_a\MA)-H(\RX_a\MA)=0, \label{eq:aux1}
\end{align}
Using the independence condition 2), we also obtain
\begin{align}
    &I(\RX_a, \RX'_c \wedge \RX_b \MB + \RX'_c \MC \mid \RG_1, \RX_a\MA) \notag\\ &=I(\RX_a, \RX'_c \wedge \RX_a\MA+\RX_b \MB + \RX'_c \MC \mid \RG_1, \RX_a\MA)\notag\\
    &=H(\RX_a\MA+\RX_b \MB + \RX'_c \MC \mid \RG_1, \RX_a\MA)%\notag\\ &\mkern 30mu 
    - H(\RX_a\MA+\RX_b \MB + \RX'_c \MC \mid \RG_1, \RX_a\MA,\RX_a, \RX'_c)\notag\\
    &=H(\RX_a\MA+\RX_b \MB + \RX'_c \MC \mid \RG_1, \RX_a\MA)%\notag\\ &\mkern 30mu  
    - H(\RX_a\MA+\RX_b \MB + \RX'_c \MC \mid \RG_1,\RX_a, \RX'_c)\notag\\
    &=H(\RX_a\MA+\RX_b \MB + \RX'_c \MC)%\notag\\ &\mkern 30mu 
    \ \textcolor{blue}{ - }\ H(\RX_a\MA+\RX_b \MB + \RX'_c \MC)\notag\\
    &=0.\label{eq:aux2}
\end{align}
It follows from \eqref{eq:final_z1} and \eqref{eq:final_z2} that
\begin{align*}
    `(1)&=I(\RZ_1\wedge \RZ_2|\RG_1)\\
        &=I(\RZ_1 \wedge \RZ_2, \RG_1 \mid \RG_1)\\
        &=I(\RX_a, \RX'_c, \RG_1 \wedge \RX_b, \RX'_c, \RG_1 \mid \RG_1)\\
        &=I(\RX_a, \RX'_c \wedge \RX_b, \RX'_c \mid \RG_1)\\
        &=I(\RX_a, \RX'_c, \RX_a\MA \wedge \RX_b, \RX'_c \mid \RG_1)\\
        &=I(\RX_a\MA \wedge \RX_b, \RX'_c \mid \RG_1)+I(\RX_a, \RX'_c \wedge \RX_b, \RX'_c \mid \RG_1, \RX_a\MA)\\
        &\stackrel{\eqref{eq:aux1}}{=}I(\RX_a, \RX'_c \wedge \RX_b, \RX'_c \mid \RG_1, \RX_a\MA)\\
        &=I(\RX_a, \RX'_c \wedge \RX_b, \RX'_c,\RX_b \MB + \RX'_c \MC \mid \RG_1, \RX_a\MA)\\
        &=I(\RX_a, \RX'_c \wedge \RX_b \MB + \RX'_c \MC \mid \RG_1, \RX_a\MA)\\
        &\mkern 30mu+I(\RX_a, \RX'_c \wedge \RX_b, \RX'_c\mid \RG_1, \RX_a\MA, \RX_b \MB + \RX'_c \MC)\\
        &\stackrel{\eqref{eq:aux2}}{=}I(\RX_a, \RX'_c \wedge \RX_b, \RX'_c \mid \RG_1, \RX_a\MA, \RX_b \MB + \RX'_c \MC)\\
        &=I(\RZ_1\wedge \RZ_2 \mid \RF')=`(3)
\end{align*}
 This completes the proof.

\end{proof}

\section{Proofs from Section~\ref{sec:treepin}} \label{app:thm:cwsk:irred}
\subsection{Proof of Lemma~\ref{lem:irred}}\label{lem:proof:irred}
In this proof, we first identify an edge whose m.c.f. with the wiretapper's observations is a non-constant function. Then, by appropriately transforming the source, we separate out the m.c.f. from the random variables corresponding to the edge and the wiretapper. Later we argue that the source can be reduced by removing the m.c.f. component entirely without affecting  $\wskc$ and $\rl$. And we repeat this process until the source becomes irreducible. At each stage, to show that the reduction indeed leaves the m.c.f. related to the other edges  unchanged and makes the m.c.f. of the reduced edge a constant function, we  use  the following lemma which is proved in Appendix~\ref{lem:mcf}.
\begin{lemma}\label{lem:indgk}
 If $(\RX,\RY)$ is independent of $\RZ$, then  $\op{mcf}(\RX, (\RY,\RZ)) = \op{mcf}(\RX,\RY)$ and $\op{mcf}((\RX,\RZ), (\RY,\RZ))  = \linebreak(\op{mcf}(\RX,\RY),\RZ)$.
\end{lemma}

 Since $(\RZ_V, \RZ_{\opw})$ is not irreducible, there exists an edge $e \in E$ such that $\RG_e := \op{mcf}(\RY_e, \RZ_{\opw})$ is a non-constant function. By using the result that the m.c.f. of a finite linear source is a linear function~\cite{chan18zero}, we can write $\RG_e =\RY_e \MM_e  =\RZ_{\opw} \MM_{\opw}$ for some full column-rank matrices, $\MM_e$ and $\MM_{\opw}$ over $\Fq$. 

We will appropriately transform the random vector $\RY_e$. Let $\MN_e$ be any matrix with full column-rank such that $\bM \MM_e \mid  \MN_e \eM$ is invertible. Define $\tRY_e := \RY_e \MN_e$, then
\begin{align*}
  \bM \RX_{e,1},\ldots,\RX_{e,n_e}\eM \bM \MM_e \mid  \MN_e \eM &= \RY_e \bM \MM_e \mid  \MN_e \eM \\
  &=\bM \RG_e, \tRY_e \eM\\
  &= \bM \RG_{e,1},\ldots,\RG_{e,\ell}, \tRX_{e,1},\ldots,\tRX_{e,\tilde{n}_e} \eM
\end{align*}
where $\tRY_e = [\tRX_{e,1},\ldots,\tRX_{e,\tilde{n}_e}]$, $\RG_e = [\RG_{e,1},\ldots, \RG_{e,\ell}]$, $\ell$ is the length of the vector $\RG_e$ and $\tilde{n}_e = n_e -\ell$. Therefore, we can obtain $(\RG_e, \tRY_e)$ by an invertible linear transformation of $\RY_e$. Note that the components $ \RG_{e,1},\ldots,\RG_{e,\ell}, \tRX_{e,1},\ldots,\linebreak \tRX_{e,\tilde{n}_e}$ are also  i.i.d. random variables that are uniformly distributed over $\Fq$, and they are independent of $ \RY_{E \setminus \{e\}}:=(\RY_b: b \in E \setminus \{e\}))$. Hence $\RG_e$ is independent of $\tRY_e$ and $\RY_{E \setminus \{e\}}$.

Now we will express $\RZ_{\opw}$ in terms of  $\RG_e$ and $\tRY_e$.
\begin{align*}
 \RZ_{\opw} &= \RX \MW\\
 & =\RY_e\MW_e + \RY_{E \setminus \{e\}} \MW_{E \setminus \{e\}}\\
 &= \bM \RG_e & \tRY_e\eM\bM \MM_e & \MN_e \eM^{-1}\MW_e + \RY_{E \setminus \{e\}} \MW_{E \setminus \{e\}}\\
 &= \RG_e \MW^{'}_e +\tRY_e \MW^{''}_e + \RY_{E \setminus \{e\}} \MW_{E \setminus \{e\}}
\end{align*}
where the  matrices $\MW_e$ and $\MW_{E \setminus \{e\}}$ are sub-matrices  of $\MW$ formed by rows corresponding to $e$ and $E \setminus \{e\}$ respectively. Also, the matrices $\MW^{'}_e$ and $\MW^{''}_e$ are sub-matrices  of $\bM \MM_e &  \MN_e \eM^{-1}\MW_e$ formed  by first $\ell$ rows and last $\tilde{n}_e$ rows respectively. Define $\tRZ_{\opw}:=\tRY_e \MW^{''}_e + \RY_{E \setminus \{e\}} \MW_{E \setminus \{e\}}$. Since $\RZ_{\opw}= \bM\RG_e & \tRZ_{\opw}\eM\bM  \MW^{'}_e \\ \MI \eM$ and $\bM\RG_e & \tRZ_{\opw}\eM = \RZ_{\opw}\bM \MM_{\opw} & \MI -\MM_{\opw}\MW^{'}_e \eM$, $\bM \RG_e & \tRZ_{\opw} \eM$ can be obtained by an invertible linear transformation of $\RZ_{\opw}$.

Since the transformations are invertible, $\RY_e$ and $\RZ_{\opw}$ can equivalently be written as $(\RG_e, \tRY_e)$ and $(\RG_e, \tRZ_{\opw} )$ respectively. We will see that $\RG_e$ can be removed from the source without affecting $\wskc$ and $\rl$.  Let us consider a  new tree-PIN  source $\tRZ_V$, which is the same as $\RZ_V$ except that  $\tRY_e$ and $\tilde{n}_e$ are associated to the edge $e$, and the wiretapper side information is  $\tRZ_{\opw}$. Note that $(\tRZ_V, \tRZ_{\opw})$ is also a tree-PIN source with linear wiretapper, and $\RG_e$ is independent of $(\tRZ_V, \tRZ_{\opw})$.

 For the edge $e$, $\op{mcf}(\tRY_e, \tRZ_{\opw})$ is a  constant function. Suppose if it were a non-constant function $\tRG_e$ w.p. 1, which  is indeed independent of $\RG_e$, then $\op{mcf}(\RY_e, \RZ_{\opw}) = \op{mcf}((\RG_e, \tRY_e), (\RG_e,\tRZ_{\opw}))= (\RG_e, \tRG_e)$. The last equality uses Lemma~\ref{lem:indgk}. Therefore, $H(\RG_e) =H(\op{mcf}(\RY_e, \RZ_{\opw}))=H(\RG_e, \tRG_e) >  H(\RG_e)$, which is a contradiction.  Moreover $H(\RY_e|\op{mcf}(\RY_e, \RZ_{\opw}))= H(\RY_e|\RG_e) =H(\tRY_e, \RG_e|\RG_e) = H(\tRY_e)$. For the other edges $b \neq e$, $\tRY_b = \RY_b$ and $\op{mcf}(\tRY_b,\tRZ_{\opw})= \op{mcf}(\RY_b,\tRZ_{\opw})= \op{mcf}(\RY_b, (\RG_e,\tRZ_{\opw})) = \op{mcf}(\RY_b, \RZ_{\opw})$, which follows from Lemma~\ref{lem:indgk}.
 
Now we will verify that $\wskc$ and $\rl$ do not change.  First let us show that $\rl(\RZ_V\|\RZ_{\opw}) \leq \rl(\tRZ_V\|\tRZ_{\opw})$ and $\wskc(\RZ_V\| \RZ_{\opw}) \geq \wskc(\tRZ_V\|\tRZ_{\opw})$.
Let $\tRF^{(n)}$  be  an optimal communication  for $\rl(\tRZ_V\|\tRZ_{\opw})$. We can make use of $\tRF^{(n)}$ to construct an omniscience communication for the source $(\RZ_V,\RZ_{\opw})$. Set  $\RF^{(n)}= (\RG_e^n, \tRF^{(n)})$. This communication is made as follows. Both the terminals incident on the edge $e$  have $\RY_e^n$ or equivalently $(\RG_e^n, \tRY_e^n)$. One of them  communicates $\RG_e^n$. In addition, all the terminals communicate according to $\tRF^{(n)}$ because for every user $i$,  $\tRZ_i^n$ is recoverable from $\RZ_i^n$. It is easy to verify that this is an omniscience communication for $(\RZ_V,\RZ_{\opw})$.
The minimum rate of leakage for omniscience 
\begin{align*}
\rl(\RZ_V\|\RZ_{\opw})&\leq \frac{1}{n}I(\RZ_V^n\wedge  \RF^{(n)}|\RZ_{\opw}^n)\\ &= \frac{1}{n}I(\RZ_V^n\wedge  \RG_e^n, \tRF^{(n)}|\RZ_{\opw}^n)\\
&\utag{a}= \frac{1}{n}I(\tRZ_V^n,\RG_e^n\wedge  \RG_e^n, \tRF^{(n)}|\tRZ_{\opw}^n, \RG_e^n) \\ &= \frac{1}{n}I(\tRZ_V^n\wedge \tRF^{(n)}|\tRZ_{\opw}^n, \RG_e^n) \\
&\utag{b}= \frac{1}{n}I(\tRZ_V^n\wedge \tRF^{(n)}|\tRZ_{\opw}^n) \utag{c}\leq \rl(\tRZ_V\|\tRZ_{\opw}) + \delta_n,
\end{align*}
for some $\delta_n \to 0$. Here, \uref{a} is due to the fact that $(\RG_e, \tRZ_{\opw})$ is obtained by a linear invertible transformation of $\RZ_{\opw}$, \uref{b} follows from the independence of $\RG_e$ and $(\tRZ_V, \tRZ_{\opw})$, and (c) uses the fact that $\tRF^{(n)}$ is an $\rl(\tRZ_V\|\tRZ_{\opw})-$achieving communication. It shows that  $\rl(\RZ_V\|\RZ_{\opw}) \leq \rl(\tRZ_V\|\tRZ_{\opw})$. Similarly, let $(\tRF^{(n)},\tRK^{(n)})$ be a communication and key pair  which is optimal  for  $\wskc(\tRZ_V\|\tRZ_{\opw})$. By letting $(\RF^{(n)},\RK^{(n)})=( \tRF^{(n)}, \tRK^{(n)})$ for the source $(\RZ_V, \RZ_{\opw})$, we can see that the key recoverability condition is satisfied. Thus $(\RF^{(n)},\RK^{(n)})$ constitute a valid SKA scheme for $(\RZ_V, \RZ_{\opw})$ which implies that $\wskc(\RZ_V\| \RZ_{\opw}) \geq \wskc(\tRZ_V\|\tRZ_{\opw})$. 

To prove the  reverse inequalities, $\rl(\RZ_V\|\RZ_{\opw}) \geq \rl(\tRZ_V\|\tRZ_{\opw})$ and $\wskc(\RZ_V\| \RZ_{\opw}) \leq \wskc(\tRZ_V\|\tRZ_{\opw})$, we use the idea of simulating  source $(\RZ_V, \RZ_{\opw})$ from  $(\tRZ_V,\tRZ_{\opw})$. Consider the source $(\tRZ_V,\tRZ_{\opw})$ in which one of the terminals $i$ incident on the edge $e$ generates an independent randomness $\tRG_e$ that has the same distribution as $\RG_e$. Then, terminal $i$ reveals $\tRG_e$ in public, from which the other terminal $j$ incident on $e$ and the wiretapper gain access to $\tRG_e$. The two terminals $i$ and $j$ simulate $\RY_e$ from $\tRY_e$ and $\tRG_e$, whereas the other terminals' observations, besides $\tRG_e$, are the same as those of $\RZ_V$.  Hence  they can communicate according to $\RF^{(n)}$ on the simulated source $\RZ_V$.  If  $\RF^{(n)}$  achieves omniscience for $\RZ_V^n$ then so does $\tRF^{(n)}=(\tRG_e^n, \RF^{(n)})$ for $\tRZ_V^n$ . Therefore the omniscience recoverability condition is satisfied. Furthermore, if we choose $\RF^{(n)}$ to be an $\rl(\RZ_V\|\RZ_{\opw})$-achieving communication, then the minimum  rate of leakage for omniscience,
\begin{align*}
\rl(\tRZ_V\|\tRZ_{\opw})&\leq \frac{1}{n}I(\tRZ_V^n\wedge  \tRF^{(n)}|\tRZ_{\opw}^n)\\
&= \frac{1}{n}I(\tRZ_V^n\wedge  \tRG_e^n, \RF^{(n)}|\tRZ_{\opw}^n)\\
&= \frac{1}{n}I(\tRZ_V^n\wedge  \tRG_e^n|\tRZ_{\opw}^n)+\frac{1}{n}I(\tRZ_V^n\wedge  \RF^{(n)}|\tRZ_{\opw}^n,\tRG_e^n)\\
&\utag{a}= \frac{1}{n}I(\tRZ_V^n,\tRG_e^n\wedge  \RF^{(n)}|\tRZ_{\opw}^n,\tRG_e^n) \\
&\utag{b}= \frac{1}{n}I(\RZ_V^n\wedge  \RF^{(n)}|\RZ_{\opw}^n) \\
&\utag{c}\leq \rl(\RZ_V\|\RZ_{\opw})+\delta_n,
\end{align*}
for some $\delta_n \to 0$. Here, \uref{a} follows from the independence of $\tRG_e$ and $(\tRZ_V, \tRZ_{\opw})$, \uref{b} is because $(\tRG_e, \tRZ_{\opw})$ can be obtained by a linear invertible transformation of $\RZ_{\opw}$, and (c) uses the fact that $\RF^{(n)}$ is an $\rl(\RZ_V\|\RZ_{\opw})$-achieving communication.
This shows that $\rl(\RZ_V\|\RZ_{\opw}) \geq \rl(\tRZ_V\|\tRZ_{\opw})$. Similarly, if  $(\RF^{(n)}, \RK^{(n)})$ is a communication and key pair for $(\RZ_V, \RZ_{\opw})$ then terminals can communicate according to  $\tRF^{(n)}= (\tRG_e^n, \RF^{(n)})$ and agree upon the key $\tRK^{(n)}= \RK^{(n)}$, which is possible due to simulation. Hence the key recoverability is immediate. The secrecy condition is also satisfied because $ I(\tRK^{(n)}\wedge  \tRF^{(n)}, \tRZ_{\opw}^n) = I(\RK^{(n)}\wedge  \RF^{(n)}, \tRG_e^n, \tRZ_{\opw}^n) = I(\RK^{(n)}\wedge  \RF^{(n)}, \RZ_{\opw}^n) $. Hence $(\tRF^{(n)},\tRK^{(n)})$ forms a valid WSKA scheme for $(\tRZ_V, \tRZ_{\opw})$ which implies that $\wskc(\RZ_V\| \RZ_{\opw}) \geq \wskc(\tRZ_V\|\tRZ_{\opw})$.

We have shown that $\rl(\RZ_V\|\RZ_{\opw}) = \rl(\tRZ_V\|\tRZ_{\opw})$,  $\wskc(\RZ_V\| \RZ_{\opw}) = \wskc(\tRZ_V\|\tRZ_{\opw})$  and  for the edge $e$, $\op{mcf}(\tRY_e, \tRZ_{\opw})$ is a constant function and $H(\RY_e|\op{mcf}(\RY_e, \RZ_{\opw}))= H(\tRY_e)$. Furthermore, we have shown that this  reduction does not change the m.c.f. of  $\RY_b$  and  $\tRZ_{\opw}$, when $b \neq e$. If the source $(\tRZ_V, \tRZ_{\opw})$ is not irreducible, then we can apply the above reduction again on $(\tRZ_V, \tRZ_{\opw})$ without affecting $\wskc$ and $\rl$. Note that the cardinality of the set of all edges $b$ such that $\op{mcf}(\RY_b, \RZ_{\opw})$ is a non-constant function reduces by one after each reduction step. So, this process terminates after a finite number of steps at an irreducible source, which completes the proof.

\subsection{Proof of Theorem~\ref{thm:cwsk:irred}}\label{thm:proof:cwsk:irred}
\emph{Converse part.} An  upper bound on  $\wskc$  is $\skc$ because the key generation ability of the users can only increase if the wiretapper has no side information. It was shown in \cite[Example 5]{csiszar04} that if the random variables of a source form a Markov chain on a tree, then $\skc = \min_{(i,j) : \{i,j\} = \xi(e) } I(\RZ_i \wedge \RZ_j)$. In the tree-PIN case, which satisfies the Markov property, this turns out to be $\skc=\min_{e \in E} H(\RY_e)$. As a consequence, we have $\wskc \leq \min_{e \in E} H(\RY_e)$ and 
\begin{align}\label{eq:rl:conv}
\begin{split}
 \rl&\utag{a}\geq H(\RZ_V|\RZ_{\opw}) -\wskc \\ 
 &\utag{b}= \left(\sum_{e \in E}n_e -n_w\right)\log_2q -\wskc \\
 &\geq \left(\sum_{e \in E}n_e -n_w\right)\log_2q  - \min_{e \in E} H(\RY_e)
 \end{split}
\end{align}\\
where \uref{a} follows from Theorem~\ref{thm:RL:lb} and \uref{b} is due to the full column-rank assumption on $\MW$.

\emph{Achievability part.}  In this section, we will show the existence of an omniscience scheme with leakage rate $\left(\sum_{e \in E}n_e -n_w\right)\log_2q  - \min_{e \in E} H(\RY_e)$. Hence $\rl \leq \left(\sum_{e \in E}n_e -n_w\right)\log_2q  - \min_{e \in E} H(\RY_e)$,  which together with the chain of inequalities~\eqref{eq:rl:conv} imply that $\wskc = \min_{e \in E} H(\RY_e)=\skc$  and $\rl =\left(\sum_{e \in E}n_e -n_w\right)\log_2q- \skc$. In particular, for achieving a secret key of rate $\wskc = \min_{e \in E} H(\RY_e)$, the terminals use privacy amplification on the recovered source.

In fact, the existence of an omniscience scheme is shown by first constructing a template for the communication with desired properties and then showing the existence of an instance of it by a probabilistic argument. The  following are the key components involved in this construction.
\begin{enumerate}
\item \emph{Deterministic scheme:} A scheme is said to be \emph{deterministic} if  terminals are  not allowed to use any locally generated private randomness. 
 \item \emph{Perfect omniscience~\cite{sirinperfect}:} For a fixed $n \in \bb{N}$, $\RF^{(n)}$ is said to achieve \emph{perfect omniscience} if  terminals can recover the source $\RZ_V^n$ perfectly, i.e., $H(\RZ_V^n|\RF^{(n)}, \RZ_i^n) =0$ for all $i \in V$. If we do not allow any private randomness, then $H( \RF^{(n)} |  \RZ_V^n) = 0$, which implies
  \begin{align*}\label{eq:perfectomni}
  \begin{split}
   \frac{1}{n}  I(\RZ_V^n\wedge \RF^{(n)} | \RZ_{\opw}^n) &= \frac{1}{n}\left [H( \RF^{(n)} | \RZ_{\opw}^n) - H( \RF^{(n)} | \RZ_{\opw}^n, \RZ_V^n) \right] \\&= \frac{1}{n}H( \RF^{(n)} | \RZ_{\opw}^n).
  \end{split}
\end{align*} 
\item \emph{Perfect alignment:} For an $n \in \bb{N}$, we say that $\RF^{(n)}$ \emph{perfectly aligns} with $\RZ_{\opw}^n$ if $H( \RZ_{\opw}^n|\RF^{(n)} ) = 0$. Note that $\RZ_{\opw}^n$ is  recoverable from $\RF^{(n)}$ but not the other way around. In this case, $H( \RF^{(n)} | \RZ_{\opw}^n) =H( \RF^{(n)}) - H( \RZ_{\opw}^n)$. In an FLS, the wiretapper side information is $\RZ_{\opw}^n = \RX^n \MW^{(n)}$ where $\RX$ is the base vector. Suppose the communication is of the form $\RF^{(n)} = \RX^n \MF^{(n)}$, for some matrix $\MF^{(n)}$, then the condition of  perfect alignment is equivalent to the condition that the column space of $\MF^{(n)}$ contains the column space of $\MW^{(n)}$. This is in turn equivalent to the condition that the left nullspace of $\MW^{(n)}$ contains the left nullspace of $\MF^{(n)}$, i.e., if $\Ry \MF^{(n)}=\R0$ for some vector $\Ry$ then $\Ry \MW^{(n)}=\R0$.
\end{enumerate}
So we will construct a (deterministic) linear communication scheme, for some fixed $n$, achieving both perfect omniscience and perfect alignment. As a consequence,  the leakage rate for omniscience is equal to $\frac{1}{n}  I(\RZ_V^n\wedge \RF^{(n)} | \RZ_{\opw}^n) = \frac{1}{n}H( \RF^{(n)} | \RZ_{\opw}^n) = \frac{1}{n}[H( \RF^{(n)}) - H( \RZ_{\opw}^n)] = \frac{1}{n}H( \RF^{(n)}) - n_w\log_2q$. To show the desired rate, it is enough to have $\frac{1}{n}H( \RF^{(n)}) = \left(\sum_{e \in E}n_e\right) \log_2q  - \min_{e \in E} H(\RY_e) $.

We describe our construction first for the case of a PIN model on a path of length $L$, and $n_e = s$ for all edges $e \in E$. The essential ideas in this construction will serve as a road map for other, more general, cases. The construction is extended to the case of tree-PIN models, again with $n_e = s$ for all edges $e$, using the the fact that there exists a unique path from any vertex to a particular vertex designated as the root of the tree. Finally, for tree-PIN models in which $n_e$ can be different for distinct edges $e$, we give only a sketch of the proof; the technical details required to fill in the sketch can be found in \cite{treepin21arxiv}. 

% This construction has been separated into multiple cases for the ease of understanding:
% \begin{enumerate}
%    \item Path with length $L$ and $n_e=s$ for all $e \in E$,
%    \item Tree with $L$ edges and $n_e=s$ for all $e \in E$,
%    \item Path and tree with $L$ edges  and arbitrary $n_e$.
% \end{enumerate}
% We give detailed constructions only for the case $n_e=s$ for all $e \in E$. First we consider a PIN model defined on a path  graph and prove the result.  We then extend it to the tree-PIN case by using the fact that there exists a unique path from any vertex to the root of the tree. For the arbitrary $n_e$ case, we can extend the proof ideas for a slightly different communication. The technical details for these two cases can be found in \cite{treepin21arxiv}.

\subsubsection{Path of length $L$ and $n_e=s$ for all $e \in E$}
 Let $V= \{0,1,\ldots,L\}$ be the set of vertices and $E=\{1,\ldots,L\}$ be the edge set such that edge $i$ is incident on  vertices $i-1$ and $i$ (Fig.~\ref{fig:path_model}). Since $n_e =s$, $\min_{e \in E} H(\RY_e)=s \log_2q$. Fix a positive integer $n$,  such that $n > \log_q(sL)$. With $n$ i.i.d. realizations of the source, the vector corresponding to edge $i$ can be expressed as $\RY_i^{n} =[ \RX^n_{i,1} \ldots \RX^n_{i,s}]$ where $\RX^n_{i,j}$'s  can be viewed as elements in $\bb{F}_{q^n}$. Hence $\RY_i^{n} \in (\bb{F}_{q^n})^s$.  The goal is to construct a linear communication scheme $\RF^{(n)}$ that simultaneously achieves perfect omniscience and perfect alignment, such that $H( \RF^{(n)}) =n \left[ \left(\sum_{e \in E}n_e\right) \log_2q  - \min_{e \in E} H(\RY_e)\right] = n  \left(sL - s\right) \log_2q$.

 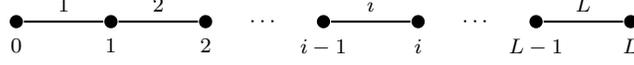
\begin{figure}[h]
\centering
\begin{tikzpicture}[-,>=stealth,thick, auto]
\tikzstyle{vertex}=[circle,fill=black,inner sep=0pt,minimum size=5pt];
\tikzstyle{every node}=[font= \fontsize{8pt}{10pt}\selectfont]
\node[vertex]      (0)        [label= below:{$0$}] {};
\node[vertex]      (1)  [right  = 3 em of 0,label= below:{$1$}] {};
\node[vertex]      (2)  [right  = 3 em of 1,label= below:{$2$}] {};
\node     (dots1)  [right  = 1 em of 2] {$\cdots$};
\node[vertex]      (i-1)  [right  = 1 em of dots1,label= below:{$i-1$}] {};
\node[vertex]      (i)  [right  = 3 em of i-1,label= below:{$i$}] {};
\node      (dots2)  [right  = 1 em of i] {$\cdots$};
\node[vertex]      (L-1)  [right  = 1 em of dots2,label= below:{$L-1$}] {};
\node[vertex]      (L)  [right  = 3 em of L-1,label= below:{$L$}] {};
% \node      (y1)  [right  = 1.3 em of 0,label= below:{$\RY_1$}] {};
% \node      (y2)  [right  = 1.3 em of 1,label= below:{$\RY_2$}] {};
% \node      (yi)  [right  = 1.3 em of i-1] {};
% \node      (x)  [below of  = yi] {$\RY_i=[ \RX_{i,1} \ldots \RX_{i,s}]$};
% \node      (yL)  [right  = 1.3 em of L-1,label= below:{$\RY_L$}] {};

\draw (0) -- (1) node [midway, above] {$1$};
\draw (1) -- (2) node [midway, above] {$2$};
\draw (i-1) -- (i) node [midway, above] {$i$};
\draw (L-1) -- (L) node [midway, above] {$L$};
\end{tikzpicture}
%\captionsetup{justification=centering}
\caption{Path of length $L$.}
\label{fig:path_model}
 \end{figure}
 
 Now we will construct the  communication as follows. Leaf nodes $0$ and $L$ do not communicate. The internal node $i$ communicates $\tRF_i^{(n)} = \RY^n_{i} + \RY^n_{i+1}\MA_{i}$, where $\MA_{i}$ is an $s \times s$ matrix with elements from $\bb{F}_{q^n}$. This  communication is of the form
\begin{align*}
\RF^{(n)} & = \begin{bmatrix}
\tRF_1^{(n)} \cdots \tRF_{L-1}^{(n)}
\end{bmatrix} = \begin{bmatrix}
\RY_1^{n}\cdots \RY_L^{n}
\end{bmatrix} \underbrace{\begin{bmatrix}
\MI & \M0& \cdots  &\M0&\M0\\
\MA_1&\MI &  \cdots &\M0&\M0\\
\M0&{\MA_2} & \cdots &\M0&\M0\\
\vdots&\vdots&\ddots&\vdots&\vdots\\
 \M0 &\M0&\cdots&\MA_{L-2}&\MI  \\
\M0 &\M0&\cdots& \M0&\MA_{L-1} \\
\end{bmatrix}}_{:=\MF^{(n)}}
\end{align*}
Here $\MF^{(n)}$ is an $sL \times s(L-1) $ matrix over $\bb{F}_{q^n}$. Observe that $\rank_{\bb{F}_{q^n}}(\MF^{(n)})= s(L-1)$, which implies that $H( \RF^{(n)}) =\left(sL - s\right) \log_2q^n$ and  the dimension of the left nullspace of $\MF^{(n)}$ is $s$. Now the communication coefficients, $(\MA_i : 1\leq i\leq L-1)$, have to be chosen such that $\RF^{(n)}$ achieves both perfect omniscience and perfect alignment. Let us derive some conditions on these matrices.

For  perfect  omniscience,  it  is  sufficient  for  the $\MA_i$'s  to  be invertible. 
%Perfect omniscience  is equivalent to the condition that the $\MA_i$'s  are invertible. The necessity of the invertibility condition is immediate since if $\MA_{L-1}$ were not invertible, then vector $\RY_L^n$ is not completely recoverable from the communication by some users, for instance, user $0$. Sufficiency 
This follows from the observation that for any $i \in V$, $[\MF^{(n)} \mid \MH_i]$  is  full rank,  where $\MH_i$ is a block-column vector with an $s \times s$ identity matrix at block-index $i$ and all-zero $s \times s$ matrix at the rest of the block-indices. In other words, $(\RY_1^{n}\cdots \RY_L^{n})$ is recoverable from $(\RF^{(n)},  \RY_i^n)$ for any $i \in E$, hence achieving omniscience. 

 For perfect alignment, we require that the left nullspace of $\MF^{(n)}$ is contained in  the left nullspace of $\MW^{(n)}$, which is the wiretapper matrix corresponding to $n$ i.i.d. realizations. Note that $\MW^{(n)}$ is a $\left(\sum_{e \in E} n_e\right) \times n_w$ matrix over $\bb{F}_{q^n}$ with entries $\MW^{(n)}(k,l) = \MW(k,l) \in  \bb{F}_{q}$; since $\bb{F}_{q} \subseteq \bb{F}_{q^n}$, $\MW^{(n)}(k,l) \in \bb{F}_{q^n}$. As pointed out before, the dimension of the  left nullspace of $\MF^{(n)}$ is $s$ whereas the dimension of the left nullspace of $\MW^{(n)}$ is $sL-n_w$. Since the source is irreducible, it follows from Lemma~\ref{lem:upbdirred} in Appendix~\ref{subsec:lemmas:irreducible} that $s \leq sL-n_w$. Since the dimensions are appropriate, the left nullspace inclusion condition is not impossible. Set $\MS := [\MS_1 \; \MS_2 \; \cdots \; \MS_L]$, where  $\MS_1$  is some invertible matrix (over $\bb{F}_{q^n}$) and $\MS_{i+1} :=(-1)^{i}\MS_1\MA_{1}^{-1}\cdots \MA_{i}^{-1}$ for $1 \leq i \leq L-1$. Observe that $\MS \MF^{(n)}=\M0$. Note that the $\MS_i$'s are also invertible, and $\MA_i= -\MS_{i+1}^{-1}\MS_i$ for $1 \leq i \leq L-1$.
%  Observe that 
% \begin{align*}
% \underbrace{\begin{bmatrix}
% \MS_1 & -\MS_1\MA_1^{-1} &
% \cdots&
% (-1)^{L-1}\MS_1\MA_{1}^{-1}\ldots \MA_{L-1}^{-1}
% \end{bmatrix}}_{:=\MS} \MF^{(n)}=\M0.
% \end{align*}
% where  $\MS_1$  is some invertible matrix. We write $\MS = [\MS_1 \ldots \MS_L]$ with $\MS_{i+1} :=(-1)^{i}\MS_1\MA_{1}^{-1}\ldots \MA_{i}^{-1}$ for $1 \leq i \leq L-1$. Notice that the $\MS_i$'s are invertible. We can also express the $\MA_i$'s in terms of the $\MS_i$'s as $\MA_i= -\MS_{i+1}^{-1}\MS_i$ for $1 \leq i \leq L-1$. 
The dimension of the left nullspace of $\MF^{(n)}$ is $s$, and all the $s$ rows of $\MS$ are independent,  so these rows span the left nullspace of $\MF^{(n)}$. Therefore for the inclusion, we must have $\MS\MW^{(n)} =\M0.$

Thus, proving the existence of communication coefficients $\MA_i$'s  that achieve perfect omniscience and perfect alignment is equivalent to proving the existence of  $\MS_i$'s  that are invertible and satisfy $[\MS_1\:\cdots\: \MS_L]\MW^{(n)} =\M0$. To do this, we use the probabilistic method.  Consider the system of equations $[\Ry_1\:\cdots\:\Ry_{sL}]\MW^{(n)} =\M0$ in $sL$ variables. Since the matrix $\MW^{(n)}$ has full column rank, the solutions can be described in terms of  $m:=sL- n_w$ free variables. As a result, any $\MS$ that satisfies $\MS\MW^{(n)}=\M0$ can be parametrized by $ms$ variables. Without loss of generality, we assume that the submatrix of $\MS$ formed by the first $m$ columns has these independent variables, $(\Rs_{i,j}: 1\leq i \leq s, 1 \leq j \leq m)$. Knowing these entries will determine the rest of the entries of $\MS$.  So we choose $\Rs_{i,j}$'s independently and uniformly from $\bb{F}_{q^n}$.  We would like to know if there is any realization such that all the $\MS_i$'s are invertible, which is equivalent to the condition $\prod_{i=1}^{L} \det(\MS_i)\neq 0$. Note that $\prod_{i=1}^{L} \det(\MS_i)$ is a multivariate polynomial in the variables $\Rs_{i,j}$, $1\leq i \leq s, 1 \leq j \leq m$, with degree at most $sL$. Furthermore the polynomial is not identically zero, which follows from the irreducibility of $\MW^{(n)}$. A proof of this fact is given in  Lemma~\ref{lem:nonzeropoly} in  Appendix~\ref{subsec:lemmas:irreducible}. Therefore, applying the  Schwartz-Zippel lemma (Lemma~\ref{lem:sz} in Appendix~\ref{subsec:lemmas:irreducible}), we have
\begin{align*}
 \Pr\left\lbrace \prod_{i=1}^{L} \det(\MS_i)\neq 0\right\rbrace \geq 1- \frac{sL}{q^n} \stackrel{(a)}{>} 0 \\
\end{align*}  
where $(a)$ follows from the choice $n > \log_q(sL)$.  Since the probability is strictly positive, there exists a realization of $\MS= [\MS_1 \;  \cdots \; \MS_L]$ such that $\MS \MW^{(n)} = 0$ and the $\MS_i$'s are invertible, which in turn shows the existence of a desired $\MF^{(n)}$.

\subsubsection{Tree with $L$ edges and $n_e=s$ for all $e \in E$}
For the tree-PIN model, we essentially use the same kind of communication construction as that of the path model.  Consider a PIN model on a tree with $L+1$ nodes and $L$ edges. To describe the linear communication, fix some leaf node as the root, $\rho$, of the tree. For any internal node $i$ of the tree, let $E_i$ denote the edges incident with $i$, and in particular, let $e^*(i)\in E_i$ denote the edge incident with $i$ that is on the unique path between $i$ and $\rho$ --- see Fig.~\ref{fig:uniquepath}. Fix a positive integer $n$,  such that $n > \log_q(sL)$. The communication from an internal node $i$ is  $(  \RY^n_{e^*(i)} + \RY^n_{e}\MA_{i,e}: e \in E_i \setminus \{e^*(i)\})$, where $\MA_{i,e}$ is an $s \times s$ matrix.  Each internal node communicates $s(d_i - 1)$ symbols from $\bb{F}_{q^n}$, where $d_i$ is the degree of the node $i$. Leaf nodes do not communicate. The total number of $\bb{F}_{q^n}$-symbols communicated is $\sum_i s(d_i-1)$, where the sum is over all nodes, including leaf nodes. The contribution to the sum from leaf nodes is in fact $0$, but including all nodes in the sum allows us to evaluate the sum as $s[2 \times(\text{number of edges}) - (\text{number of nodes})] = s(L-1)$.
 Thus, we have the overall communication of the form 
\begin{align*}
 \RF^{(n)} = \RY^n \MF^{(n)} 
\end{align*}
 where $\MF^{(n)} $ is a $sL \times s(L-1)$ matrix over $\bb{F}_{q^n}$ and $\RY^n = (\RY^n_e)$. The rows of $\MF^{(n)}$  correspond to the edges  of the tree.  The aim is to  choose the matrices $\MA_{i,e}$ in a way that simultaneously achieves perfect omniscience and perfect alignment.
%  such that $H( \RF^{(n)}) =n \left[ \left(\sum_{e \in E}n_e\right) \log_2q  - \min_{e \in E} H(\RY_e)\right] = n  \left(sL - s\right) \log_2q$.  
 
For perfect omniscience,  it is sufficient for the $\MA_{i,e}$'s to be     invertible. First observe that all the leaf nodes are connected to  the root node $\rho$  via paths. On each of these paths the communication has exactly the same form as that of the path model considered before. So when the $\MA_{i,e}$'s are invertible, the root node can recover the entire source using  $\RY_{e_\rho}^n$, where $e_\rho$ is the edge incident on $\rho$. Now fix a node $i \neq \rho$.  It follows from a property of trees that there is a unique path from $\rho$ to $i$. Again the form of the communication restricted to this path is the same as that of the path model.
Therefore, if the $\MA_{i,e}$'s are invertible, then node $i$ can recover the observation, $\RY_{e_\rho}^n$, of node $\rho$, using the communication along the unique path. Since the root node is able to recover the entire source using $\RY_{e_\rho}^n$ and the overall communication, node $i$ can also recover the entire source using the recovered observation $\RY_{e_\rho}^n$ and the overall communication.

Because $\RY^n$ is recoverable from $(\RF^{(n)},  \RY_e^n)$ for any $e \in E$, $[\MF^{(n)} \mid \MH_e]$ is an invertible $sL \times sL$ matrix, where $\MH_e$ is a block-column vector with $\MI$ at the location corresponding to edge e and zero matrices in the rest of the locations. Therefore $\MF^{(n)}$ is a full column-rank matrix, i.e., $\rank_{\bb{F}_{q^n}}(\MF^{(n)})= s(L-1)$, which implies that $H( \RF^{(n)}) =\left(sL - s\right) \log_2q^n$ and  the dimension of the left nullspace of $\MF^{(n)}$ is $s$.

For perfect alignment, we require that the left nullspace of $\MF^{(n)}$ is contained in  the left nullspace of $\MW^{(n)}$. So, let us construct an $\MS = (\MS_e)$ such that  $ \MS\MF^{(n)}=\M0$ as follows. Let $\MS_1$ be an invertible matrix. Each edge $e$ has two nodes incident with it; let $i^*(e)$ denote the node that is closer to the root $\rho$. There is a unique path $i^*(e) = i_1 \longrightarrow i_2 \longrightarrow  \cdots \longrightarrow i_{\ell} = \rho$ that connects $i^*(e)$ to $\rho$ and let the edges along the path in this order be $(e=e_1, e_2,\ldots, e_{\ell})$ --- see Fig.~\ref{fig:uniquepath}.
 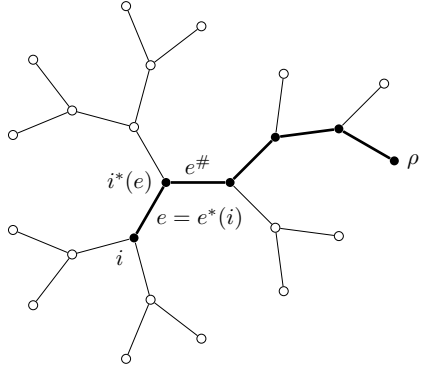
\begin{figure}[h]
\centering
\resizebox{\totalheight}{!}{% \begin{tikzpicture}[-,>=stealth,thick, auto]
% \tikzstyle{vertex}=[circle,fill=black,inner sep=0pt,minimum size=5pt];
% \tikzstyle{every node}=[font= \fontsize{8pt}{10pt}\selectfont]
% \node[vertex]      (0)        [label= below:{$\rho=i_l$}] {};
% \node[vertex]      (1)  [right  = 3 em of 0,label= below:{$i_{l-1}$}] {};
% \node[vertex]      (2)  [right  = 3 em of 1,label= below:{$i_{l-2}$}] {};
% \node     (dots1)  [right  = 1 em of 2] {$\cdots$};
% \node[vertex]      (3)  [right  = 1 em of dots1,label= below:{$i_2$}] {};
% \node[vertex]      (4)  [right  = 3 em of 3,label= below:{\begin{tabular}{l}
%     $i^*(e)$\\
%     $=i_1$ 
% \end{tabular}}] {};
% \node[vertex]      (5)  [right  = 3 em of 4] {};
% \node      (dots2)  [right  = 1 em of 5] {$\cdots$};

% % \node      (y1)  [right  = 1.3 em of 0,label= below:{$\RY_1$}] {};
% % \node      (y2)  [right  = 1.3 em of 1,label= below:{$\RY_2$}] {};
% % \node      (yi)  [right  = 1.3 em of i-1] {};
% % \node      (x)  [below of  = yi] {$\RY_i=[ \RX_{i,1} \ldots \RX_{i,s}]$};
% % \node      (yL)  [right  = 1.3 em of L-1,label= below:{$\RY_L$}] {};

% \draw (0) -- (1) node [midway, above] {$e_l$};
% \draw (1) -- (2) node [midway, above] {$e_{l-1}$};
% \draw (3) -- (4) node [midway, above] {$e_2$};
% \draw (4) -- (5) node [midway, above] {$e_1=e$};
% % \draw[thin] (1) -- ($ (1) + (0.5,-0.5) $);
% % \draw[thin] (2) -- ($ (2) + (0.5,-0.5) $);
% % \draw[thin] (2) -- ($ (2) + (0.5,0.5) $);
% % \draw[thin] (5) -- ($ (5) + (0.5,-0.5) $);
% % \draw[thin] (5) -- ($ (5) + (0.5,0.5) $);
% \end{tikzpicture}

\begin{tikzpicture}
[   level distance=1 cm,
    level 1/.style = {sibling angle =120},
    level 2/.style = {sibling angle = 90},
    level 3/.style = {sibling angle = 75},
    level 4/.style = {sibling angle = 75},
    grow cyclic
]
\tikzstyle{vbold}=[circle,fill=black,inner sep=0pt,minimum size=4pt];
\tikzstyle{v}=[draw, circle,inner sep=0pt,minimum size=4pt];
%\tikzstyle{every node}=[font=\fontsize{8pt}{10pt}\selectfont]

\node [vbold] [label= left:{$i^*(e)$}] (i*) {}
    child{node[vbold] [label=-120:{$i$}](i){} child{node[v]{} child{node[v]{}} child{node[v]{}}} child{node[v]{} child{node[v]{}} child{node[v]{}}}}   
    child{node [vbold](i3) {}child{node [v] {} child{node [v] {}} child{node [v] {}}} child{node [vbold](i4) {} child{node [vbold](i5) {}child{node[vbold][label= right:{$\rho$}](i6){}} child{node[v] {}}} child{node [v] {}}}}    
    child{node [v] {} child{node [v] {} child{node [v] {} } child{node [v] {}}} child{node [v] {} child{node [v] {}} child{node [v] {}}}};
    \draw[very thick] (i*) -- (i) node [midway, right, yshift=-0.35 em, xshift=-0.1 em] {$e=e^*(i)$};
    \draw[very thick] (i*) -- (i3) node [midway, above] {$e^\#$};
    \draw[very thick] (i3) -- (i4)--(i5)--(i6);

 \end{tikzpicture}}
\caption{Unique path between an internal node $i$ and the root $\rho$}
\label{fig:uniquepath}
 \end{figure}
We set $\MS_e := (-1)^{\ell-1} \MS_1 \MA^{-1}_{i_{\ell -1}, e_{\ell-1}} \cdots \MA^{-1}_{i_{1}, e_{1}} $ for all edges $e$ except for the edge incident with $\rho$, to which we associate $\MS_1$. Note that the $\MS_e$'s are invertible and $\MS_e= - \MS_{e^{\#}}\MA^{-1}_{i^*(e), e}$, where  $e^{\#}$ is the edge adjacent to $e$ on the unique path from $i^*(e)$ to $\rho$. Let us now verify that $\MS \MF^{(n)} = \M0$. The  component corresponding to the internal node $i$ in $\MS \MF^{(n)}$ is of the form $(\MS_{e^*(i)} + \MS_{e}\MA_{i,e}: e \in E_i \setminus \{e^*(i)\})$. But for an  $e \in E_i \setminus \{e^*(i)\}$, $i^{*}(e) = i$ and $e^{\#} = e^*(i)$, thus $\MS_{e}\MA_{i,e} = - \MS_{e^{\#}}\MA^{-1}_{i^*(e), e}\MA_{i,e}= - \MS_{e^*(i)}\MA^{-1}_{i, e}\MA_{i,e} =- \MS_{e^*(i)}$. Hence we have  $\MS_{e^*(i)} + \MS_{e}\MA_{i,e}=\M0$ which implies $\MS \MF^{(n)} = \M0$.
The dimension of the left nullspace of $\MF^{(n)}$ is $s$ and all the $s$ rows of $\MS$ are independent,  so these rows span the left nullspace of $\MF^{(n)}$. Therefore, for the inclusion of one nullspace within the other, we must have $\MS\MW^{(n)} =\M0$.

Finally, we can prove the existence of an $\MS=(\MS_e)$ such that $\MS\MW^{(n)} =\M0$ and the $\MS_e$'s are invertible, using the probabilistic method exactly as before. The details are omitted.  This shows the existence of a desired $\MF^{(n)}$.

\subsubsection{Path and tree with $L$ edges  and arbitrary $n_e$}
In this case, we define $s := \min\{n_e: e \in E\}$. We consider a communication $\RF^{(n)}$ that consists of two parts. One part involves the communication that is similar to that of the $n_e =s$ case, where we use the first $s$ random variables associated to each edge $e$. And the other part involves revealing the rest of the random variables on each edge, but this is done by linearly combining them with the first $s$ rvs.

For this kind of a communication structure, we can in fact show, in a similar way as in the $n_e =s$ case, the existence of an $\MF^{(n)}$ with the desired properties. The technical details are omitted but they can be found in \cite{treepin21arxiv}.

\subsection{Proof of Lemma~\ref{lem:indgk}}\label{lem:mcf}
Recall that we assume that $\RZ$ is independent of  $(\RX,\RY)$.  Any common function (c.f.) of $\RX$ and $\RY$ is also a common function of $\RX$ and $(\RY,\RZ)$. Let $\RF$ be a c.f. of $\RX$ and $(\RY,\RZ)$ which means that $H(\RF|\RX)=0=H(\RF|\RY,\RZ)$. Note that $H(\RF|\RY)=H(\RZ|\RY)+H(\RF|\RY,\RZ)-H(\RZ|\RF,\RY)=H(\RZ)-H(\RZ|\RF,\RY)$. Also we have $H(\RZ|\RF,\RY) \geq H(\RZ|\RX,\RY)$ which follows from the  fact that $\RF$ is a function of $\RX$. Both these inequalities together imply that $0 \le H(\RF|\RY) \leq H(\RZ)-H(\RZ|\RX,\RY) =0$. So any c.f. of $\RX$ and $(\RY,\RZ)$ is also a c.f. of $\RX$ and $\RY$.  Therefore $\op{mcf}(\RX, (\RY,\RZ)) = \op{mcf}(\RX,\RY)$. 
 
 We can see that $ (\op{mcf}(\RX,\RY) ,\RZ)$ is a c.f. of $(\RX,\RZ)$ and $(\RY,\RZ)$. To show that $\op{mcf}((\RX,\RZ), (\RY,\RZ)) = (\op{mcf}(\RX,\RY), \RZ)$, it is enough to show that $H(\op{mcf}(\RX,\RY) ,\RZ) \geq H(\RG)$ for any $\RG$ satisfying $H(\RG|\RX,\RZ)=0=H(\RG|\RY,\RZ)$. Since $\sum_{\Rz \in \mc{Z}}P_{\RZ}(\Rz) H(\RG|\RX,\RZ=\Rz)=H(\RG|\RX,\RZ)=0$, for a $\Rz \in \op{supp}(P_{\RZ})$,  we have $H(\RG|\RX,\RZ=\Rz)=0$. Similarly, $H(\RG|\RY,\RZ=\Rz)=0$. Thus, for a fixed $\RZ =\Rz$, $\RG$ is a c.f.\ of rvs $\RX$ and $\RY$ jointly distributed according to $P_{\RX, \RY \mid \RZ=\Rz}$. In this case, let $\op{mcf}(\RX,\RY)_{\RZ=\Rz}$ denote the m.c.f. which indeed depends on the conditional distribution.  However, because of the independence $P_{\RX, \RY \mid \RZ=\Rz} =P_{\RX, \RY}$,  the $\op{mcf}(\RX,\RY)_{\RZ=\Rz}$ remains the same across all $\Rz$, and is equal to  $\op{mcf}(\RX,\RY)$. Therefore, from the optimality of m.c.f., we have $H(\RG |\RZ=\Rz) \leq  H(\op{mcf}(\RX,\RY)_{\RZ=\Rz} |\RZ=\Rz)=H(\op{mcf}(\RX,\RY) |\RZ=\Rz)= H(\op{mcf}(\RX,\RY))$, where the last equality follows from the independence of $\RZ$ and $(\RX,\RY)$. As a consequence, we have $H(\RG |\RZ) =\sum_{\Rz \in \mc{Z}}P_{\RZ}(\Rz) H(\RG|\RZ=\Rz)\leq H(\op{mcf}(\RX,\RY))$. The desired inequality follows from $H(\RG) \leq H(\RG,\RZ) =H(\RG |\RZ) + H(\RZ) \leq  H(\op{mcf}(\RX,\RY)) + H(\RZ) =H(\op{mcf}(\RX,\RY),\RZ)$.  This proves that $\op{mcf}((\RX,\RZ), (\RY,\RZ)) = (\op{mcf}(\RX,\RY), \RZ)$.

 \subsection{Useful Lemmas related to the proof of Theorem~\ref{thm:cwsk:irred}} \label{subsec:lemmas:irreducible}
\begin{lemma}[Schwartz-Zippel lemma]\label{lem:sz}
Let $\op{P}(\RX_1,\ldots,\RX_n)$ be a non-zero polynomial in $n$ variables with degree $d$ and coefficients from a finite field $\Fq$. Given a non-empty set $S \subseteq \Fq$, if we choose the $n$-tuple $(\RMx_1, \ldots, \RMx_n)$ uniformly from $S^n$, then
\begin{align*}
\Pr \{(\RMx_1, \ldots, \RMx_n)\in S^n: \op{P}(\RMx_1, \ldots, \RMx_n) = 0\} \leq \frac{d}{|S|}.
\end{align*}
\end{lemma}

Fix two positive integers $m$ and $s$ such that $s\leq m$. Consider the integral domain $\Fq\left[\RX_{11}, \ldots ,\RX_{1m},\ldots, \RX_{s1}, \ldots ,\RX_{sm}\right]$, which is the set of all multivariate polynomials in indeterminates $ \RX_{11}, \ldots ,\RX_{1m},\ldots, \RX_{s1}, \ldots ,\RX_{sm}$ with coefficients from a finite field $\Fq$. Let us consider a matrix of the form
\begin{align}
\MM=\begin{bmatrix}
\op{L}_1(\RY_1)&\op{L}_2(\RY_1)&\cdots &\op{L}_s(\RY_1) \\
\op{L}_1(\RY_2)&\op{L}_2(\RY_2)&\cdots &\op{L}_s(\RY_2) \\
\vdots & \vdots & \ddots & \vdots\\
\op{L}_1(\RY_s)&\op{L}_2(\RY_s)&\cdots &\op{L}_s(\RY_s)
\end{bmatrix}_{s \times s}, \label{eqn:detmatrix}
\end{align}
where $\RY_k:=[\RX_{k1}, \ldots ,\RX_{km}]$ for $1 \leq k \leq s$ and $\op{L}_{i}(\RY_k)$  denotes a linear combination over $\Fq $ of the indeterminates $ \RX_{k1}, \ldots ,\RX_{km}$. Note that row $k$ depends only on $\RY_k$.  Let  $\RX := [\RY^T_1, \ldots, \RY^T_s]^T$, and let $\op{P}(\RX)$ denote a polynomial in the indeterminates $ \RX_{11}, \ldots ,\RX_{1m},\ldots, \RX_{s1}, \ldots ,\RX_{sm}$, with coefficients from $\Fq$. 
 
It is a fact \cite[p.~528]{bourbaki1989algebra} that  for a general matrix $\MM$ with entries from $\Fq\left[\RX\right]$, $\det(\MM)=0$ if and only if  there exist polynomials $\op{P}_k (\RX)$, $1 \leq k \leq s$, not all zero,  such that
\begin{align*}
\MM \begin{bmatrix} \op{P}_1(\RX)  , \ldots ,  {\op{P}_s} (\RX) \end{bmatrix}^T= \M0.
\end{align*}
But this does not guarantee a non-zero $\lambda = [\lambda_1, \ldots, \lambda_s]  \in \Fq^s$ such that $ \MM \lambda^T= 0$.  However, the following lemma shows that if the matrix is of the form  (\ref{eqn:detmatrix}), then this is the case.

\begin{lemma} \label{lem:det}
Let $\MM$ be a matrix of the form (\ref{eqn:detmatrix}). Then  $\det(\MM)=0$ iff  there exists a  non-zero $\lambda = [\lambda_1, \ldots, \lambda_s]  \in \Fq^s$ such that $\MM \lambda^T= 0$.
\end{lemma}
\begin{proof}
The ``if" part holds for any matrix $\MM$ by the fact stated above. 
%If $\MM \lambda^T= 0$ for some non-zero $\lambda = [\lambda_1, \ldots, \lambda_s]  \in \Fq^n$, then columns are linearly dependent over $\Fq\left[X\right]$ which implies that $\det(\MM)=0$.
For the ``only if" part, suppose that $\det(\MM)=0$.  We can write $\MM$ as follows
\[\MM=\underbrace{\begin{bmatrix}
\RX_{11}&\RX_{12}&\cdots &\RX_{1m} \\
\RX_{21}&\RX_{22}&\cdots &\RX_{2m} \\
\vdots & \vdots & \ddots & \vdots\\
\RX_{s1}&\RX_{s2}&\cdots &\RX_{sm} 
\end{bmatrix}}_{=\MX}\underbrace{\begin{bmatrix}
a_{11}&a_{21}&\cdots &a_{s1} \\
a_{12}&a_{22}&\cdots &a_{s2} \\
a_{13}&a_{23}&\cdots &a_{s3} \\
\vdots & \vdots & \ddots & \vdots\\
a_{1m}&a_{2m}&\cdots &a_{sm} 
\end{bmatrix}}_{:=\MA}.\]
 for some $\MA \in \Fq^{m \times s}$.  
Now consider the determinant of the matrix $\MM$,
\begin{align*}
\det(\MM) &= \sum_{\sigma \in S_s} \sgn(\sigma )\op{L}\nolimits_{\sigma(1)}(\RY_1)\cdots \op{L}\nolimits_{\sigma(s)}(\RY_s)\\
&=\sum_{\sigma \in S_s}\sgn(\sigma ) \left( \sum_{j_1=1}^{m}a_{\sigma(1)j_1} \RX_{1j_1}\right)\cdots \left( \sum_{j_s=1}^{m}a_{\sigma(s)j_s} \RX_{sj_s}\right)\\
&= \sum_{\sigma \in S_s}\sgn(\sigma )\sum_{j_1,\dots,j_s \in [m]^s}\left(a_{\sigma(1)j_1}\cdots a_{\sigma(s)j_s}\right)\RX_{1j_1}\cdots \RX_{sj_s}\\
& = \sum_{j_1,\ldots,j_s\in [m]^s}\left(\sum_{\sigma \in S_s}\sgn(\sigma )a_{\sigma(1)j_1}\cdots a_{\sigma(s)j_s}\right)\RX_{1j_1}\cdots \RX_{sj_s}\\
&= \sum_{j_1,\ldots,j_s\in [m]^s} \det(A_{j_1\ldots j_s})\RX_{1j_1}\cdots\RX_{sj_s} 
\end{align*}
where $\MA_{j_1j_2\ldots j_s}$ is the $s\times s$ submatrix of $\MA$ formed by the rows $j_1, j_2, \dots ,j_s$. 
%$(a)$ follows from the fact that the monomials $\RX_{1j_1}\RX_{2j_2}\ldots \RX_{sj_s}$, for $j_1,j_2,\ldots,j_s\in [m]^s$, are distinct. $(b)$ holds because the inner sum is just the determinant of $\MA_{j_1j_2\ldots j_s}$. 
Since $\det(\MM)=0$, $\det(\MA_{j_1j_2\ldots j_s})= 0$ for every collection of distinct indices $j_1,j_2,\dots,j_s$, which implies that  any $s$ rows of $\MA$ are linearly dependent over $\Fq$. This shows that the rank$_{ \Fq}(\MA) < s$, therefore the columns of $\MA$ are linearly dependent over $\Fq$. Hence there exists a  non-zero $\lambda = [\lambda_1, \ldots, \lambda_s]  \in \Fq^n$ such that $ \MA\lambda^T= 0 \Rightarrow \MM\lambda ^T= 0$.
\end{proof}

\begin{definition}
 Let $\MW$ be a row-partitioned matrix of the form  
 \begin{align} \label{blockcolumn}
 \renewcommand{\arraystretch}{1.5}
 \begin{bmatrix}
 \begin{array}{c}
  \MW_1\\ \hline
  \MW_2\\ \hline
  \vdots\\ \hline
  \MW_{|E|}
  \end{array}
  \end{bmatrix}
 \end{align}
where $\MW_i$ is an $n_i \times n_w$ matrix over $\Fq$. We say that the matrix $\MW$ is \emph{reducible} if there exist an index $i$ and a non-zero row vector $r_i$ in $\Fq^{n_i}$ such that the column span of $\MW$ contains the column vector $[-0- \mid  \cdots \mid - r_i-\mid \cdots \mid -0-]^T$. If the matrix $\MW$ is not reducible then, we say it is \emph{irreducible}.
\end{definition}
A tree-PIN source with linear wiretapper is irreducible iff  the wiretapper matrix $\MW$ is irreducible.
\begin{lemma} \label{lem:upbdirred}
 Let $\MW$ be a $(\sum_{e \in E} n_e) \times  n_w$ wiretapper  matrix  in the row-partitioned form \eqref{blockcolumn}. If the matrix $\MW$ is irreducible then $n_w  \leq (\sum_{e \in E}n_e)-s$ where $s=\min\{n_e: e \in E\}$. 
\end{lemma}
\begin{proof}
 By  elementary column operations and block-row swapping, we can reduce $\MW$ into the following form
 \begin{align*}
 \renewcommand{\arraystretch}{1.5}
 \left[\begin{array}{cccc}
 \MW_{11}&\M0& \cdots&\M0\\ \hline
 \MW_{21}&\MW_{22}& \cdots &\M0\\ \hline
 \vdots&\vdots&\ddots&\vdots\\ \hline
  \MW_{k1}&\MW_{k2}& \cdots &\MW_{kk}\\ \hline
  \vdots&\vdots&\ddots&\vdots\\ \hline
  \MW_{|E|1}&\MW_{|E|2}& \cdots &\MW_{|E|k}\\
 \end{array}\right]
\end{align*}
 where the diagonal matrices $\MW_{jj}$ are full column-rank matrices. Since $\MW$ is an irreducible matrix, $k \leq (|E|-1)$. An upper bound on  the number of columns of $\MW_{jj}$ is  $n_{e_j}$, where $e_j$ is the edge corresponding to the row $j$ (after block-row swapping). So, 
 \begin{align*}
  n_w & \leq \max\biggl\{\sum_{j\in K} n_{e_j}: K \subseteq [|E|], |K| \leq (|E|-1)\biggr\} \\
  &\leq \max\biggl\{\sum_{j\in K} n_{e_j}:  |K| = (|E|-1)\biggr\}\\
  &= \max\biggl\{\sum_{e\in E} n_e - n_{e'}:  e' \in E\biggr\}\\
  &=\sum_{e\in E} n_e - s.
 \end{align*}
 This completes the proof.
\end{proof}

The next lemma is about matrices over $\Fq\left[\RX\right]$ of the form
\begin{align}
\begin{bmatrix}
\RX_{11}&\cdots &\RX_{1m} & \op{L}_1(\RY_1)&\cdots &\op{L}_{l}(\RY_1) \\
\RX_{21}&\cdots &\RX_{2m} & \op{L}_1(\RY_2)&\cdots &\op{L}_{l}(\RY_2) \\
\vdots &  \ddots & \vdots &\vdots &  \ddots & \vdots\\
\RX_{s1}&\cdots &\RX_{sm} & \op{L}_1(\RY_s)&\cdots &\op{L}_{l}(\RY_s) 
\end{bmatrix}_{s \times m+l} \label{eqn:lemma_matrix}
\end{align}
where $\op{L}_{i}(\RY_k)$  denotes a linear combination over $\Fq$ of entries of $\RY_k=[\RX_{k1}, \ldots ,\RX_{km}]$. Let us denote a matrix whose entries are the zero polynomials by $\M0$.

\begin{lemma}\label{lem:nonzeropoly}
Let $\MW$ be a $(\sum_{e \in E} n_e) \times  n_w$ wiretapper  matrix over $\Fq$ with full column-rank such that $n_w  \leq (\sum_{e \in E}n_e)-s$, where $s=\min\{n_e: e \in E\}$. Let $m := \sum_{e \in E}n_e - n_w $. Consider a  matrix $\MS:=(\MS_e, \MT_e)_{e \in E}$ over $\Fq\left[\RX\right]$ of the form \eqref{eqn:lemma_matrix}, where $\MS_e$ is an $s \times s$ matrix and $\MT_e$ is an $s \times (n_e -s)$ matrix. Furthermore, assume that $\MS$ satisfies $\MS\MW=\M0$ . If $\MW$ is an irreducible matrix, then $\prod_{e\in E} \det (\MS_e)$ is a non-zero polynomial.
   %(Polynomial in terms of the inderterminates corresponding to the free variables of $\MS$  corresponding to $\MS\MW=0$). 
\end{lemma}
\begin{proof}
 Suppose  $\prod_{e\in E} \det (\MS_e)$ is the zero polynomial; then $ \det (\MS_{e^*}) \equiv 0$ for some $i\in E$. There are $sm$ indeterminates in $\MS$, where $s\leq m$. Note that $\MS_{e^*}$ is of the form $\eqref{eqn:detmatrix}$ for some linear functions. By Lemma~\ref{lem:det}, $ \det (\MS_{e^*}) \equiv 0$ implies that there exists a non-zero $\lambda = [\lambda_1, \ldots, \lambda_s]  \in \Fq^s$ such that $ \MS_{e^*} \lambda^T= 0$.  Consider the row vector $\MR=(\MR_e)_{e \in E}$ with $\MR_{e^*} = [\lambda_1, \ldots, \lambda_s, 0,\ldots, 0]\in \Fq^{n_{e^*}}$ and  $\MR_j =[- 0 -]\in \Fq^{n_{e'}}$ for every $e' \in E \setminus \{e^*\}$. Then $ \MS \MR^T= 0$. 
 
 Moreover, it is given that $\MS$ satisfies $\MS\MW=\M0$. Now, let the $m$ indeterminates in the first row of $\MS$ take values in $\Fq$ so that we get $m$ linearly independent vectors in the left nullspace of $\MW$. These vectors are also in the left nullspace of $\MR^T$ because $\MS \MR^T= 0$. Since $\MW$ has full column-rank, this is possible only if $\MR^T$ is in the column span of $\MW$, which implies that $\MW$ is reducible.

% Consider the matrix $\tMW = [\MW \mid \MR^T]$ which also satisfies $\MS \tMW =\M0$. One can see that $\ker(\tMW^T) \subseteq \ker(\MW^T) $. For the other direction, note that any vector in the $\ker(\MW^T)$ also belongs to $\ker(\MR)$. As a consequence $\ker(\tMW^T) = \ker(\MW^T) $, then the dimension of the column space of $\tMW$ is $\sum_{e\in E} n_e  - \dim(\ker(\tMW^T)) =\sum_{e\in E} n_e  - \dim(\ker(\MW^T)) = n_w$. Hence $\MR^T$ is in the column span of $\MW$ which implies that $\MW$ is reducible.  
\end{proof}

 %Since $\MS$ satisfies $\MS\MW = 0$, in each row of $\MS$ there are $m$ independent variables, which are indeterminates, and every other element in the row is expressed  as  a linear combination of these indeterminates. So, in total there are $sm$ indeterminates in $\MS$; without loss of generality, assume them to be in the first $m$ columns of $\MS$. 

\section{Proof of Theorem~\ref{thm:rateconstrained_treepin}} \label{app:thm:proof_rateconstrained_treepin}
Similar to the unconstrained case, we first prove the result for irreducible sources and then argue that the rate region of a general source is the same as that of an irreducible source that is obtained through reduction.
\begin{figure}[h]
\centering
\begin{tikzpicture}
\tikzstyle{every node}=[font= \fontsize{8pt}{10pt}\selectfont]
    \draw [thin,  ->] (0,-0.5) -- (0,2)      % draw y-axis line
        node [above, black] {$\wskc(R)$};              % add label for y-axis
    
    \draw [thin,  ->] (-0.5,0) -- (3,0)      % draw x-axis line
        node [right, black] {$R$};              % add label for x-axis
    
    \draw [draw=black, thick] (0,0) -- (1.5,1);% draw the graph
    \draw [draw=black, thick] (1.5,1) -- (3,1);
    \draw [thin,  dashed] (0,1) -- (1.5,1);
    \draw [thin,  dashed] (1.5,1) -- (1.5,0);
    
    \node [left] at (0,1) {$\wskc$};                % label y-intercept
    \node [below] at (1.5,0) {$(|E|-1)\wskc$};               % label x-intercept
\end{tikzpicture}
\caption{$\wskc(R)$ curve denoting the wiretap secret key capacity at a given rate $R$ }
\label{fig:rateregion}
 \end{figure}
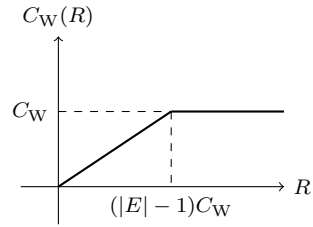
\begin{theorem}\label{thm:rate:irreducible}
Given an irreducible tree-PIN source $\RZ_V$ with a linear wiretapper $\RZ_w$, we have 
\begin{align*}
\wskc(R) = \min \left\{\frac{R}{|E|-1}, \wskc\right\} 
\end{align*}
where $R$ is the total discussion rate and $\wskc =\min_{e \in E}H(Y_e)$, which is the unconstrained wiretap secret key capacity.
\end{theorem}
\begin{proof}
Since the wiretapper side information can only reduce the secret key rate, $\wskc(R) \leq \skc(R)$. It follows from \cite[Theorem~4.2]{chan19} that $\skc(R) =  \min \left\{\frac{R}{|E|-1}, \skc\right\}$. Therefore, we have $\wskc(R) \leq \min \left\{\frac{R}{|E|-1}, \wskc\right\}$ because $\skc =\wskc$ for an irreducible tree-PIN source with a linear wiretapper, which was shown in Theorem~\ref{thm:cwsk:irred}.

For the achievability part, it is enough to show that the point $((|E|-1)\wskc,\wskc)$ is achievable because the rest of the curve follows from the time sharing argument between $((|E|-1)\wskc,\wskc)$ and $(0,0)$ ---  see Fig.~\ref{fig:rateregion}.

Let $s:=\wskc=\min_{e \in E}H(Y_e)=\min_{e \in E}n_e$, which is an integer. We will construct our achievable scheme on a sub-source  $\RZ'_V$ of the tree-PIN source $\RZ_V$ by ignoring some edge random variables. More precisely, $\RZ'_V$ is defined on the same tree $T$ with $\RY'_e := \left( \RX_{e,1}, \ldots, \RX_{e,s} \right)$ for each edge $e \in E$, and $\RZ'_i = \left( \RY'_e : i \in \xi (e) \right)$ for $i \in V$. Note that all the edge random vectors $\RY'_e$ have $s$ components. On the other hand, the  wiretapper side information $\RZ_{\opw}$ is the same as that of the original source. 

Let $\RX' :=(\RX_{e,k}: e\in E, 1 \leq k \leq s)$ and  $\RX'' :=(\RX_{e,k}: e\in E, s < k \leq n_e)$, which is a partition of the underlying components $\RX$ of the original source. This gives rise to a partition of the observations of the wiretapper into two parts: the first part contains observations involving only linear combinations of $\RX'$, and the second part contains linear observations with at least one component from $\RX''$. This means that $\RZ_{\opw}$, after applying some suitable invertible linear transformation, can be written as 
$$\RZ_{\opw} = \bM \RX' & \RX''\eM \bM \MA & \MB \\ \M{0} & \MC \eM ,$$ for some matrices $\MA$, $\MB$, and a full column-rank matrix $\MC$. With $\RZ'_{\opw}=\RX'\MA$ and $\RZ''_{\opw}=\RX'\MB + \RX''\MC$, $\RZ_{\opw} = \bM \RZ'_{\opw} & \RZ''_{\opw} \eM$.

For a large $n$,  users execute a linear secure omniscience communication scheme $\RF^{(n)}$ on the sub-source $\RZ'^n_V$ with respect to the wiretapper side information $\RZ^n_{\opw}$. Moreover, $\RF^{(n)}$ has the following properties: it achieves perfect omniscience at rate 
$$\frac{1}{n}H(\RF^{(n)})=H(\RZ'_V)-s=H(\RX')-s,$$ which is the minimum rate of omniscience $\rco(\RZ'_V)$, and it perfectly aligns with $\RZ'^n_{\opw}$, i.e., $H(\RZ'^n_{\opw}|\RF^{(n)})=0$. The existence of such a communication scheme is guaranteed from the proof of Theorem~\ref{thm:cwsk:irred}. After every user recovers the source $\RZ'^n_V$ using $\RF^{(n)}$, they agree on the key $\RK^{(n)}:=\RY'^n_{e_0}$ where $e_0 \in E$ is an edge incident on a leaf node. It is clear that $\RK^{(n)}$ satisfies the key recoverability condition because it is a function of the recovered source $\RZ'^n_V$. It remains to show that $\RK^{(n)}$ satisfies the secrecy condition.

Since $\RK^{(n)},\RZ'^n_{\opw}$ and $\RF^{(n)}$ are linear functions of $\RX'$, we have  $(\RK^{(n)}, \RF^{(n)}, \RZ'^n_{\opw}) - \RX'^n-\RZ''^n_{\opw}$. Note that $\RZ''^n_{\opw}$ is independent of $\RX'$ because $\MC$ is a full column-rank matrix. As a consequence, $\RZ''^n_{\opw}$ is independent of $(\RK^{(n)}, \RF^{(n)}, \RZ'^n_{\opw})$.  Furthermore, $\RY'^n_{e_0}$ is independent of $\RF^{(n)}$. This can be obtained by combining the perfect omniscience condition, which implies that $H(\RZ'^n_V|\RY'^n_{e_0},\RF^{(n)})=0$ for the leaf node, and the condition on the rate of the communication which is  $H(\RF^{(n)})=H(\RZ'^n_V)-ns=H(\RZ'^n_V)-H(\RY'^n_{e_0})$. Therefore, we have $H(\RY'^n_{e_0}|\RF^{(n)})= H(\RY'^n_{e_0},\RF^{(n)})-H(\RF^{(n)})=H(\RZ'^n_V,\RY'^n_{e_0},\RF^{(n)})-H(\RF^{(n)})=H(\RZ'^n_V)-H(\RF^{(n)})=H(\RY'^n_{e_0})$. The third equality is because $\RY'^n_{e_0}$ and $\RF^{(n)}$ are linear functions of $\RZ'^n_V$. Finally,
\begin{align*}
    H(\RK^{(n)}| \RF^{(n)}, \RZ^n_{\opw}) &=  H(\RK^{(n)}| \RF^{(n)}, \RZ'^n_{\opw},\RZ''^n_{\opw})\\
    &\utag{a}{=}  H(\RK^{(n)}| \RF^{(n)}, \RZ'^n_{\opw})\\
    &\utag{b}{=} H(\RK^{(n)}| \RF^{(n)})\\
    &= H(\RY'^n_{e_0}| \RF^{(n)})\\
    &\utag{c}{=}H(\RY'^n_{e_0})\\
    &= H(\RK^{(n)})
\end{align*}
where (a) follow from the independence of $\RZ''^n_{\opw}$ and $(\RK^{(n)}, \RF^{(n)}, \RZ'^n_{\opw})$, (b) is due that the fact that $\RF^{(n)}$ aligns perfectly with $\RZ'^n_{\opw}$, i.e., $H(\RZ'^n_{\opw}|\RF^{(n)})=0$ and (c) is because $\RY'^n_{e_0}$ is independent of $\RF^{(n)}$.

Thus we have shown that a secret key of rate $\frac{1}{n}H(\RK^{(n)})=\frac{1}{n}H(\RY'^n_{e_0})=s$ is achievable with a communication of rate $\frac{1}{n}H(\RF^{(n)})=H(\RZ'_V)-s= (|E|-1)s$. So the pair $((|E|-1)\wskc,\wskc)= ((|E|-1)s,s)$  is achievable, which is as desired.
\end{proof}

To extend this result to the general tree-PIN case, we will prove the following lemma, which allows us to carry out a reduction to an irreducible source without changing $\wskc(R)$. This lemma along with the above theorem on irreducible sources proves Theorem~\ref{thm:rateconstrained_treepin}. 

\begin{lemma}\label{lem:irred:rate} 
 If a tree-PIN source with a linear wiretapper $(\RZ_V,\RZ_{\opw})$ is not irreducible then there exists an irreducible source $(\tRZ_V, \tRZ_{\opw})$ such that 
 \begin{align*}
\wskc(\RZ_V\| \RZ_{\opw})(R) = \wskc(\tRZ_V\|\tRZ_{\opw})(R),\\ 
H(\RY_e|\op{mcf}(\RY_e\wedge\RZ_{\opw})) = H(\tRY_e)
\end{align*}
for all $e \in E$ and all discussion rates $R \geq 0$.
\end{lemma}
\begin{proof}
Since $(\RZ_V, \RZ_{\opw})$ is not irreducible, there exists an edge $e \in E$ such that $\RG_e := \op{mcf}(\RY_e\wedge \RZ_{\opw})$ is a non-constant function. Similar to the proof of Lemma~\ref{lem:irred:rate}, we linearly transform $\RY_e$ and $\RZ_{\opw}$ to  $(\RG_e, \tRY_e)$ and $(\RG_e, \tRZ_{\opw} )$, respectively where $ H(\tRY_e)=H(\RY_e|\op{mcf}(\RY_e\wedge\RZ_{\opw}))$. Let us consider a  new tree-PIN  source $\tRZ_V$, which is the same as $\RZ_V$ except that  $\tRY_e$ and $\tilde{n}_e$ are associated to the edge $e$, and the wiretapper side information is  $\tRZ_{\opw}$. Note that $(\tRZ_V, \tRZ_{\opw})$ is also a tree-PIN source with a linear wiretapper, and $\RG_e$ is independent of $(\tRZ_V, \tRZ_{\opw})$.

Since any valid scheme on reduced model $(\tRZ_V, \tRZ_{\opw})$ can be used as a valid scheme on original model $(\RZ_V, \RZ_{\opw})$, we have 
$$\wskc(\RZ_V\| \RZ_{\opw})(R) \geq \wskc(\tRZ_V\|\tRZ_{\opw})(R).$$

To prove the reverse inequality, $\wskc(R):=\wskc(\RZ_V\| \RZ_{\opw})(R) \leq \wskc(\tRZ_V\|\tRZ_{\opw})(R)$, consider capacity-optimal schemes. Fix $\delta>0$, and let $(\RF^{(n)},\RK^{(n)})$ be a WSKA scheme whose key rate and communication rate satisfy
\begin{gather}
    \wskc(R)-\delta < \liminf \frac{1}{n}\log|\mc{K}^{(n)}| < \wskc(R), \label{eq:rate_k}\\
    \limsup \frac{1}{n}\log|\mc{F}^{(n)}| < R. \nonumber
\end{gather}
Let $L:=\liminf \frac{1}{n}\log|\mc{K}^{(n)}|$. Fix an $\epsilon>0$ small enough that $(L-7\epsilon, L+7\epsilon) \subseteq (\wskc(R)-\delta, \wskc(R))$.
We restrict to a subsequence of $\left(\frac{1}{n}\log|\mc{K}^{(n)}|\right)_{n\geq 1}$ whose limit is $L$. And, with an abuse of notation, we still index this sequence\footnote{Designing a protocol on a subsequence can easily be extended to all the integers without sacrificing the asymptotic performance of the protocol.} with $n$. Therefore, for all $n$ large enough,  $L-\epsilon <\frac{1}{n}\log|\mc{K}^{(n)}| < L+\epsilon$. Again, along this subsequence, we have $\limsup \frac{1}{n}\log|\mc{F}^{(n)}| < R$ because of the properties of subsequential limits. Moreover, we have  $0 \leq \log|\mc{K}^{(n)}|-H(\RK^{(n)}|\RZ^n_{\opw},\RF^{(n)})=\left[\log|\mc{K}^{(n)}|-H(\RK^{(n)})\right] +I(\RK^{(n)} \wedge \RZ^n_{\opw},\RF^{(n)})< \epsilon$ and $\Pr\left[ \exists j\in V \text{ s.t. }  \RK_j^{(n)}\neq \RK^{(n)}\right]< \epsilon$ for all $n$ large enough.

Note that the above secrecy  condition implies that $0\leq\left[\log|\mc{K}^{(n)}|-H(\RK^{(n)})\right] < \epsilon$ and  $I(\RK^{(n)}\wedge\tRZ^n_{\opw},\RG^n_e,\RF^{(n)})=I(\RK^{(n)}\wedge\RZ^n_{\opw},\RF^{(n)})<  \epsilon$. The latter condition implies that $I(\RK^{(n)}\wedge\tRZ^n_{\opw},\RF^{(n)}|\RG^n_e)<  \epsilon$ and $I(\RK^{(n)}\wedge \RG^n_e)<  \epsilon$. From all the above conditions, we have
\begin{gather}
L-\epsilon <\frac{1}{n}\log|\mc{K}^{(n)}| < L+\epsilon,\label{eq:rate1}\\
    0\leq\left[\log|\mc{K}^{(n)}|-H(\RK^{(n)})\right] < \epsilon,\label{eq:rate2}\\
    I(\RK^{(n)}\wedge\tRZ^n_{\opw},\RF^{(n)}|\RG^n_e)<  \epsilon,\label{eq:rate3}\\ I(\RK^{(n)}\wedge \RG^n_e)<  \epsilon.\label{eq:rate4}
\end{gather}
Let $\mc{K}^{(n)}_g$ be the range of $\RK^{(n)}$ when $\RG_e^n=g$, and similarly, $\mc{F}^{(n)}_g$ denote the range of $\RF^{(n)}$ when $\RG_e^n=g$. Note that since $\mc{K}^{(n)}_g \subseteq \mc{K}^{(n)}$ and $\mc{F}^{(n)}_g \subseteq \mc{F}^{(n)}$ for any $g$, we have $\log|\mc{K}^{(n)}_g|\leq \log|\mc{K}^{(n)}|$ and $\log|\mc{F}^{(n)}_g|\leq \log|\mc{F}^{(n)}|$, hence
\begin{gather}
    \sum_{g}\Pr(\RG_e^n=g)\log|\mc{K}^{(n)}_g|\leq  \log|\mc{K}^{(n)}|.\label{eq:rate5}
\end{gather}
We will show that there exists a realization $g^{*}$ of $\RG_e^n$ for which all the desired conditions are met.

From \eqref{eq:rate2}, \eqref{eq:rate4} and \eqref{eq:rate5}, we obtain
\begin{gather*}
    H(\RK^{(n)}|\RG^n_e)\leq \sum_{g}\Pr(\RG_e^n=g) \log|\mc{K}^{(n)}_g|\leq \log|\mc{K}^{(n)}| < H(\RK^{(n)}|\RG^n_e)+2\epsilon.
\end{gather*}
We can conclude from the above chain of inequalities that 
\begin{gather*}
    0\leq \sum_{g}\Pr(\RG_e^n=g)\left[\frac{1}{n} \log|\mc{K}^{(n)}|-\frac{1}{n} \log|\mc{K}^{(n)}_g|\right]< 2\epsilon,\\
    0\leq \sum_{g}\Pr(\RG_e^n=g)\left[ \log|\mc{K}_g^{(n)}|-H(\RK^{(n)}|\RG^n_e=g)\right]< 2\epsilon.
\end{gather*}
These inequalities together with \eqref{eq:rate3} give
\begin{align*}
    0\leq &\sum_{g}\Pr(\RG_e^n=g)\left\lbrace\left[\frac{1}{n} \log|\mc{K}^{(n)}|-\frac{1}{n} \log|\mc{K}^{(n)}_g|\right]+\left[ \log|\mc{K}_g^{(n)}|-H(\RK^{(n)}|\RG^n_e=g)\right]\right. \\ &\left.+\, I(\RK^{(n)}\wedge\tRZ^n_{\opw},\RF^{(n)}|\RG^n_e=g)+\Pr\left( \exists j\in V \text{ s.t. } \RK_j^{(n)}\neq \RK^{(n)}\Bigm\vert \RG^n_e=g\right)\right\rbrace < 5\epsilon.
\end{align*}
Note that the averaged quantity (the term in the curly bracket) is non-negative; hence there exists a realization of $\RG_e^{(n)}=g^*$ such that 
\begin{align*}
    0\leq &\left[\frac{1}{n} \log|\mc{K}^{(n)}|-\frac{1}{n} \log|\mc{K}^{(n)}_{g^*}|\right]+\left[ \log|\mc{K}_{g^*}^{(n)}|-H(\RK^{(n)}|\RG^n_e={g^*})\right]\\ &+\,I(\RK^{(n)}\wedge\tRZ^n_{\opw},\RF^{(n)}|\RG^n_e={g^*})+\Pr\left( \exists j\in V \text{ s.t. } \RK_j^{(n)}\neq \RK^{(n)}\Bigm\vert \RG^n_e=g\right) < 5\epsilon.
\end{align*}
Since all the summands are non-negative, we have 
\begin{subequations}
\label{eq:rate6}
\begin{gather}
    0\leq \left[\frac{1}{n} \log|\mc{K}^{(n)}|-\frac{1}{n} \log|\mc{K}^{(n)}_{g^*}|\right]< 5\epsilon,\label{eq:rate7}\\
    0\leq \left[ \log|\mc{K}_{g^*}^{(n)}|-H(\RK^{(n)}|\RG^n_e={g^*})\right]< 5\epsilon,\label{eq:rate8}\\ 
    0\leq I(\RK^{(n)}\wedge\tRZ^n_{\opw},\RF^{(n)}|\RG^n_e={g^*}) < 5\epsilon,\label{eq:rate9}\\
    0\leq \Pr\left( \exists j\in V \text{ s.t. } \RK_j^{(n)}\neq \RK^{(n)}\Bigm\vert \RG^n_e=g^*\right) < 5\epsilon\label{eq:rate10}.
\end{gather}
\end{subequations}
Let $\RK^{(n)}_{j, g^*}$ (resp. $\RF^{(n)}_{g^*}$) be the function  $\RK^{(n)}_j$ (resp. $\RF^{(n)}$) restricted to $\RG^n_e={g^*}$. So $\RK^{(n)}_{j, g^*}$ and $\RF^{(n)}_{g^*}$ are  functions solely of $\tRZ_V^n$ and possible private randomness. For the source $(\tRZ_V, \tRZ_{\opw})$, users run the scheme $(\RF^{(n)}_{g^*},\RK^{(n)}_{1, g^*}, \ldots, \RK^{(n)}_{m, g^*})$ to generate a secret key using the source $\tRZ_V^n$ and the private randomness. It is enough to argue that this a valid WSKA scheme. To see this, let $P_{\RK^{(n)},\tRZ_V^n, \tRZ_{\opw}^n,\RS_V,\RG^n_e}$ be the joint distribution of $\RK^{(n)}$ and the source random variables, where $\RS_V$ is the local randomness generated by the users. Consider a random variable $\RK^{(n)}_{g^*}$ whose joint distribution with the other random variables is $P_{\RK^{(n)}_{g^*},\tRZ_V^n, \tRZ_{\opw}^n,\RS_V}=P_{\RK^{(n)}|\tRZ_V^n, \tRZ_{\opw}^n,\RS_V,\RG^n_e=g^*}\cdot P_{\tRZ_V^n, \tRZ_{\opw}^n,\RS_V}$. Since $\RG^n_e$ is independent of $(\tRZ_V^n, \tRZ_{\opw}^n,\RS_V)$, we have $H(\RK^{(n)}_{g^*})=H(\RK^{(n)}|\RG^n_e={g^*})$, $I(\RK^{(n)}_{g^*}\wedge\tRZ^n_{\opw},\RF^{(n)})=I(\RK^{(n)}\wedge\tRZ^n_{\opw},\RF^{(n)}|\RG^n_e={g^*})$ and $\Pr\left( \exists j\in V \text{ s.t. } \RK_{j, g^*}^{(n)}\neq \RK_{g^*}^{(n)}\right)=\Pr\left( \exists j\in V \text{ s.t. } \RK_j^{(n)}\neq \RK^{(n)}\Bigm\vert \RG^n_e=g^*\right)$. As $\epsilon>0$ is arbitrary, the conditions \eqref{eq:rate10}, \eqref{eq:rate9}, and \eqref{eq:rate8} imply that recoverability and secrecy conditions are met, proving that this is a valid WSKA scheme. Moreover, the rate of this code satisfies
$\liminf \frac{1}{n} \log|\mc{K}^{(n)}_{g^*}|> \liminf \frac{1}{n} \log|\mc{K}^{(n)}|-6 \epsilon > L-7\epsilon$. From the assumption that $(L-7\epsilon, L+7\epsilon) \subseteq (\wskc(R)-\delta, \wskc(R))$, we have $\liminf \frac{1}{n} \log|\mc{K}^{(n)}_{g^*}|>\wskc(R)-\delta$. Furthermore, the rate of the communication $\RF^{(n)}_{g^*}$ satisfies $\limsup\frac{1}{n}\log|\mc{F}^{(n)}_{g^*}|\leq \limsup\frac{1}{n}\log|\mc{F}^{(n)}| < R$.
Thus,
\begin{align*}
    \wskc(\tRZ_V\|\tRZ_{\opw})(R)>\wskc(\RZ_V\| \RZ_{\opw})(R)-\delta.
\end{align*}
As $\delta$ is arbitrary, we have $\wskc(\tRZ_V\|\tRZ_{\opw})(R)
\geq \wskc(\RZ_V\| \RZ_{\opw})(R).$
This shows that  $\wskc(\RZ_V\| \RZ_{\opw})(R)= \wskc(\tRZ_V\|\tRZ_{\opw})(R)$. Therefore, we can repeat this process until the source becomes irreducible without affecting $\wskc(\RZ_V\| \RZ_{\opw})(R)$.
\end{proof}

The result of Theorem~\ref{thm:rateconstrained_treepin} follows by putting the above lemma and theorem together.

\bibliographystyle{IEEEtran}
\bibliography{IEEEabrv,ref}

\end{document}